\documentclass[11pt]{article}
\usepackage{amsmath}
\usepackage{graphicx}
\usepackage{enumerate}
\usepackage{natbib}
\usepackage{url} %
\usepackage{algorithm}
\usepackage[noend]{algpseudocode}
\usepackage{setspace}
\usepackage{bm}

\newcommand{\red}[1]{{\color{red}#1}}

\usepackage{ upgreek }
\usepackage{graphicx} 
\usepackage{textcomp} 
\usepackage{amsfonts} 
\usepackage{amsmath}
\usepackage{amssymb}
\usepackage{yfonts}
\usepackage{mathtools}
\usepackage{mathrsfs}
\usepackage{latexsym}
\usepackage{array}
\usepackage{amsthm} 
\usepackage{calc}
\usepackage{perpage}
\usepackage{textcomp}
\usepackage{verbatim} 
\usepackage{marvosym}
\usepackage{hieroglf}
\usepackage{multirow} %
\usepackage{rotating} %
\usepackage{dsfont} %
\usepackage{booktabs}
\usepackage[english]{babel}
\usepackage{hyperref} 
\usepackage{url}
\usepackage{color}
\usepackage{tikz}
\usetikzlibrary{shapes,arrows,positioning,calc}
\usepackage{float}
\usepackage{soul}
\usepackage{multirow}
\usepackage{siunitx}
\usepackage{float}
\usepackage{booktabs}
\usepackage{siunitx}
\usepackage{soul}
\usepackage{subfigure}

\newtheorem{theorem}{Theorem}[section]

\newtheorem{lemma}{Lemma}[section]

\newtheorem{remark}{Remark}[section]
\newtheorem{example}{Example}[section]

\newtheorem{assumption}{Assumption}[section]

\newcommand{\Dkonv}{\stackrel{d}{\rightarrow}}
\newcommand{\eps}{\varepsilon}

\newcommand{\normsup}[1]{\ensuremath{\!|\!| #1 | \! |_{\infty}}}

\newcommand{\normtwo}[1]{ \ensuremath{\left\| #1 \right\|_{2}} }

\newcommand{\norm}[1]{\ensuremath{\!|\!| #1 | \! |}}

\newcommand{\tri}[1]{{\left\vert\kern-0.25ex\left\vert\kern-0.25ex\left\vert #1 
    \right\vert\kern-0.25ex\right\vert\kern-0.25ex\right\vert}}
\newcommand{\normop}[1]{{\left\vert\kern-0.25ex\left\vert\kern-0.25ex\left\vert #1 
    \right\vert\kern-0.25ex\right\vert\kern-0.25ex\right\vert}_\infty}
\newcommand{\normf}[1]{{\left\vert\kern-0.25ex\left\vert\kern-0.25ex\left\vert #1 
    \right\vert\kern-0.25ex\right\vert\kern-0.25ex\right\vert}_{\op{F}}}
 \newcommand{\normspec}[1]{{\left\vert\kern-0.25ex\left\vert\kern-0.25ex\left\vert #1 
    \right\vert\kern-0.25ex\right\vert\kern-0.25ex\right\vert}_{2}}

\newcommand{\bigo}{\mathcal{O}}

\newcommand{\1}{{\rm 1}\mskip -4,5mu{\rm l} }

\newcommand{\argmin}{\mathop{\mathrm{arg\,min}}}
\newcommand{\argmax}{\mathop{\mathrm{arg\,max}}}

\newcommand{\op}[1]{\operatorname{#1}}  %

\def\E{\mathbb{E}}

\def\G{\mathbb{G}}

\def\R{\mathbb{R}}

\def\Z{\mathbb{Z}}

\def\Ll{\mathcal{L}}

\newcommand{\rp}{\mathbb{R}^p}
\newcommand{\rd}{\mathbb{R}^d}
\newcommand{\rpp}{\mathbb{R}^{p\times p}}
\newcommand{\rnp}{\mathbb{R}^{n\times p}}

\newcommand{\rnn}{\mathbb{R}^{n\times n}}

\newcommand{\mpr}{\mathbb{P}}

\def\Var{\mathbb{V}\mbox{ar}\,}

\def\cov{\mbox{Cov}\,}

\newcommand{\graph}{\mathcal{G}}

\newcommand{\grad}{\nabla}

\newcommand{\vecx}{\mathbf{x}}
\newcommand{\vecb}{{\boldsymbol{\beta}}}

\newcommand{\vecz}{{\mathbf{z}}}
\newcommand{\vecg}{{\boldsymbol{\gamma}}}
\newcommand{\response}{Y}

\newcommand{\vecw}{{\mathbf{w}}}

\newcommand{\est}{{\hat{\boldsymbol{\gamma}}}}
\newcommand{\estb}{{\hat{\boldsymbol{\beta}}}}

\newcommand{\Dmat}{\ensuremath{D}}

\begin{document}
	
	\title{Spectral clustering with variance information for group structure estimation in panel data
		\thanks{%
			We are grateful to professors H.J. Wang and Y. Zhang for sending us the code for their simulations in~\cite{zhang2019} . We thank Professor D. Millimet for kindly sharing this data with us. The data we use here is the same as in \cite{millimet2003environmental}. We also thank the AE and anonymous referees for constructive comments on an earlier version of this manuscript that motivated us to consider the local analysis in section 2.2 and resulted in a greatly improved manuscript.}}
	\author{Lu Yu\thanks{Department of Statistical Sciences, University of Toronto. E-mail: \texttt{stat.yu@mail.utoronto.ca} }
		\and
		Jiaying Gu\thanks{Department of Economics, University of Toronto. E-mail: \texttt{jiaying.gu@utoronto.ca}}
		\and
		Stanislav Volgushev\thanks{Department of Statistical Sciences, University of Toronto. E-mail: \texttt{stanislav.volgushev@utoronto.ca} }
	}
	
	\maketitle
	\bigskip
	\begin{abstract}
		Consider a panel data setting where repeated observations on individuals are available. Often it is reasonable to assume that there exist groups of individuals that share similar effects of observed characteristics, but the grouping is typically unknown in advance. We first conduct a local analysis which reveals that the variances of the individual coefficient estimates contain useful information for the estimation of group structure. We then propose a method to estimate unobserved groupings for general panal data models that explicitly accounts for the variance information. Our proposed method remains computationally feasible with a large number of individuals and/or repeated measurements on each individual. The developed ideas can also be applied even when individual-level data are not available and only parameter estimates together with some quantification of estimation uncertainty are given to the researcher.
		 A thorough simulation study demonstrates superior performance of our method than existing methods and we apply the method to two empirical applications. 
		
	\end{abstract}
	
	\noindent%
	\vfill
	
	\addtolength{\textheight}{.5in}

	\section{Introduction}

	Panel data models are a standard empirical tool in statistics, economics, marketing, and financial research. The conventional modeling approach is to assume that all individual heterogeneity can be summarized by an individual specific intercept, often known as the fixed effects, while assuming all covariates have a common effect among all the individuals, such that information can be pooled across individuals to gain efficiency for estimating these common parameters. However, heterogeneous responses towards observed control variables are often better supported by empirical evidence, especially as detailed individual level data becomes more available.

	An increasingly popular approach to model unobserved heterogeneity in the effects of covariates on individual responses is to assume the existence of a {finite} number of homogeneous groups. Here, parameters in a potentially non-linear model\footnote{Examples include quantile regression and discrete outcome models. } are assumed to take common values within groups but differ across groups. The main challenge is to learn the unobserved group structure from observed data. {An alternative way to model unobserved heterogeneity is through latent factors (e.g., \cite{bai2009panel}). This approach also has discrete heterogeneity in the sense that a small number of unobserved factors drive the co-movement of a large number of time series. Both group pattern and factor structure are useful empirical tools, but they have different interpretations. In this paper, we focus on group patterns.}
	
	The existing literature can be roughly categorized into three categories. Methods from the first category rely on minimizing a loss function that incorporates different coefficients for all individuals combined with a penalty which encourages the coefficient estimates to be similar. \cite{SSP} propose the classifier-LASSO (C-LASSO) approach, which is applicable to both linear and nonlinear models. Differences among individual parameters are penalized through a LASSO type penalty, and consistent grouping can be achieved if the penalty parameter is chosen properly. \cite{wang2018homogeneity} propose a Panel-CARDS penalty which extends the idea of homogeneity pursuit in \cite{ke2015homogeneity} from cross-sectional models to panel data models. \cite{gu2019panel} propose to use the convex clustering penalty of \cite{hocking2011clusterpath} in panel data quantile regression models with grouped individual intercepts and common slope parameters. 
	
	An alternative approach is to relate the group structure estimation problem to clustering; here clusters in the coefficient vectors correspond to latent groups of individuals. Estimating clusters has a long history in statistics and economics. Among the many clustering algorithms, the $k$-mean algorithm by \cite{macqueen1967some} is one of the most popular and commonly used methods. However, instead of directly applying $k$-mean methods on the estimated individual parameters, \cite{LinNg} and \cite{BM} propose to incorporate the regression loss function and re-estimate the group-specific coefficients in an iterative fashion. Originally proposed for linear regression models, this approach has also been extended to quantile regression models by \cite{zhang2019} and \cite{Leng}. Further advancement of this literature has considered time varying group membership, for example \cite{miao2020panel}, \cite{okui2021heterogeneous} and \cite{lumsdaine2021estimation}. 
	
	Both the penalization-based and clustering-based approaches described above require the repeated fitting of large regression models which involve all individuals and all individual-specific parameters in a large-scale minimization problem. This can be computationally costly especially for large scale datasets, which become more and more common in practice. In addition, the extensions of the $k$-means approach discussed above rely on iterative algorithms with random initialization which require repeated application with many different starting points. Motivated by those computational challenges, \cite{CM} propose an estimator for linear panel data models with grouped intercepts and common slope. Their approach is shown to guarantee the same  theoretical properties as \cite{BM} but is computationally much faster. It should be pointed out however that their approach seems to be difficult to extend to non-linear panels.
	\cite{wang2021} propose to use ordered individual-specific regression estimators to convert the grouping problem into a change-point detection setting and apply binary segmentation to learn the underlying group structure. This approach can be applied to both linear and nonlinear panel data models. It is computationally efficient because the individual-specific regressions only need to be estimated once rather than in an iterative fashion. They further show that by considering the spectral decomposition of an outer product of the individual parameter estimates and then applying binary segmentation on the leading eigenvectors can lead to improved group estimation. 
	
	In the present paper, we propose a novel approach that retains the computational advantages of working with individual-specific regressions but explicitly takes into account the uncertainty in the corresponding estimates. This information is particularly important in settings where different entries of a coefficient vector are estimated with different degrees of precision and hence carry varying amounts of information about the underlying population coefficients. To motivate the specific form of reweighting we use, we first conduct a simplified analysis in a local alternative framework. In the simplest case where there are only two groups in the population, we study the probability of classifying an individual to one of two groups when the separation between group centers tends to zero at a certain rate. This analysis targets a simplified iteration step which is the key ingredient of most existing iterative procedures for estimating group membership.
	
This local analysis motivates us to weigh the differences between coefficient estimates of different individuals by an estimated variance-covariance matrix.	The resulting weighted differences can not be interpreted as a Euclidean distance. This renders many classical clustering approaches such as the vanilla $k$-means algorithm or extensions of homogeneity pursuit and binary segmentation inapplicable. We handle this challenge by interpreting the weighted distances as a quantification of dissimilarity between individuals. With this interpretation, we can apply any clustering approach that works with general measures of dissimilarity. We consider two popular approaches: \emph{k-medoids} \cite{schubert2019faster} and \emph{spectral clustering} \cite{ng2002}. In simulation studies, we find that both approaches outperform existing proposals. In finite samples, the spectral clustering approach works better than the k-medoids approach and we provide high level assumptions which guarantee consistent group structure recovery asymptotically.

	The remaining paper is organized as follows. In Section \ref{sec: local} we present the simple local analysis motivates our approach. Section~\ref{sec:methodology} contains a detailed description of the proposed estimation procedure and illustrates it on several specific models that were previously considered in the literature. Section~\ref{sec:theoryhighlevel} contains theoretical guarantees on correct group estimation under high-level conditions. Those conditions are verified for several examples in Section~\ref{sec:theoryexamples}. A simulation study is presented in Section~\ref{sec:sims}. An empirical illustration analyzing the heterogeneous relationship between income and pollution level among different states using data from the United States is given in Section~\ref{sec:data}. We also apply our approach to the commuting zone summary statistics provided by \cite{chetty2018impacts} to analyze group patterns of intergenerational income mobility. Section \ref{sec:conclusion} concludes. All proofs and some additional plots are deferred to the supplementary material.

	\section{Setting and proposed methodology}	
	
	\subsection{General setting}
	
	Assume that we have repeated observations $(\vecx_{it},Y_{it})_{t=1,...,T}$ from individuals $i=1,...,n$. Our goal is to assign the individuals into $G^*$ groups such that individuals in the same group share a set of characteristics. For now, let $G^*$ be given, a data-driven choice of $G^*$ will be discussed at a later point. 
	
	Specifically, assume that the characteristics of individual $i$ are described by a vector of parameters $\vecg_i$ and that we are interested in grouping individuals according to sub-vectors $\vecb_i \in \R^p$ of $\vecg_i$. For instance, $\vecg_i$ can be coefficients in a non-linear model linking the response $Y_{it}$ to the covariates $\vecx_{it}$ and $\vecb_i$ can be the full vector $\vecg_i$, a sub-vector thereof, or simply the intercept term in a regression model. Specific examples are provided in Section~\ref{sec:ex}.
	
	A popular approach to such problems, pioneered by \cite{LinNg} and \cite{BM}, is to interpret this as a clustering problem and apply an iterative approach in the spirit of Lloyd's k-means clustering algorithm. For concreteness, assume that we only have two groups and that the coefficient vectors $\vecg_i = (\alpha_i,\vecb_i)$\footnote{since the $\alpha_i$ will be left unrestricted, they correspond to the individual specific effects} can be estimated by minimizing a loss function $\Ll$ via
	\[
	(\hat \alpha_i, \hat \vecb_i) = \argmin_{\alpha,\vecb} \sum_{t=1}^T \Ll(\vecx_{it},Y_{it};\alpha,\vecb).
	\]
	Roughly speaking, procedures in the spirit of \cite{LinNg, BM} consist of an initialization step where individuals are assigned to groups in a randomized fashion, followed by iterative re-assignments until convergence. In the k'th iteration step, denote the group centers from step $k-1$ by $\hat \vecb_1^{(k-1)}, \hat \vecb_2^{(k-1)}$. Now individual $i$ is assigned to group $1$ iff\footnote{\cite{BM} consider linear least squares models where the individual-specific intercepts $\alpha_i$ can be differenced out. The method presented here is a canonical generalization of their approach to non-linear models where differencing out individual effects may not be possible.}
	\begin{equation}\label{eq:BMsimple}
		\inf_{\alpha} \sum_{t=1}^T \Ll(\vecx_{it},Y_{it};\alpha,\hat \vecb_1^{(k-1)}) < \inf_{\alpha} \sum_{t=1}^T \Ll(\vecx_{it},Y_{it};\alpha,\hat \vecb_2^{(k-1)}).
	\end{equation}
	This approach has been adopted to quantile regression by~\cite{zhang2019}. In practice, it has two potential drawbacks. First, for large $n,T$ the cost of each iteration step can be expensive. Second and more importantly, if only initial estimators $\hat \alpha_i, \hat \vecb_i$ but not individual level data are available, this approach is infeasible to implement.
	
	Assuming that we only have access to estimators $\hat \alpha_i, \hat \vecb_i$ and covariance estimates $\hat \Sigma_i$ for $\hat \vecb_i$, a natural alternative to the iteration step is to assign individual $i$ to group $1$ iff 
	\begin{equation}\label{eq:PAMsimple}
		\|\hat \Sigma_i^{-1/2}(\hat\vecb_i - \hat \vecb_1^{(k-1)})\|_2 < \|\hat \Sigma_i^{-1/2}(\hat\vecb_i - \hat \vecb_2^{(k-1)})\|_2.
	\end{equation}
	For a motivation, note that the problem of assigning individual $i$ to group $1$ or $2$ reduces to classifying an individual into one of two classes. The rule in~\eqref{eq:PAMsimple} can now be viewed as an approximate Bayes rule in classification: if $\hat \Sigma_i$ are fixed and $\hat \vecb_i - \vecb_{i}^* \sim N(0,\hat \Sigma_i)$ and the population parameters $\vecb_{i}^*$ satisfy $\vecb_{i}^* \in \{\hat \vecb_1^{(k-1)}, \hat \vecb_2^{(k-1)}\}$,~\eqref{eq:PAMsimple} reduces to the Bayes rule which is known to be optimal for minimizing classification error.

	 At this point, it is natural to wonder whether the rule in~\eqref{eq:BMsimple} or in~\eqref{eq:PAMsimple} should be used. We next argue that, in a simplified but general setting, the classification error of rule~\eqref{eq:PAMsimple} is (asymptotically) always at least as good as that of~\eqref{eq:BMsimple}.

	\subsection{Loss functions versus weighted distances of estimators: a local analysis} \label{sec: local}
	
	To keep the presentation focused and notation simple, consider a single individual and drop the index $i$ throughout this section. Assume that the true parameter that generated the data is $\vecg^* = (\alpha^*,\vecb^*)$ and that we want to decide based on observations $(\vecx_t,Y_t)_{t=1,\dots,T}$ whether the data are generated from parameter $(\alpha_1,\vecb_1)$ or $(\alpha_2,\vecb_2)$ where $\vecb_1,\vecb_2$ are given and $\alpha_1,\alpha_2 \in \R$ are unspecified. Let $\Gamma$ denote the parameter space and define
	\[
	(\hat \alpha, \hat \vecb) :=  \argmin_{(\alpha,\vecb) \in \Gamma} \sum_{t=1}^T \Ll(\vecx_{t},Y_{t};\alpha,\vecb).
	\]
	Denote by $\hat \Sigma$ a consistent estimator of the asymptotic variance of $\hat \vecb$. Define 
	\[
	\hat k^{BM} = 1 \iff \inf_{\alpha} \sum_{t=1}^T \Ll(Y_t - \alpha - \vecx_t^\top \vecb_1) < \inf_\alpha \sum_{t=1}^T\Ll(Y_t - \alpha - \vecx_t^\top \vecb_2)  
	\] 
	and 
	\[
	\hat k^{PAM} = 1 \iff \| \hat\Sigma^{-1/2} (\hat \vecb - \vecb_1) \|_2 <  \| \hat\Sigma^{-1/2} (\hat \vecb - \vecb_2) \|_2.  
	\]
	We also consider a more general approach for a general weight matrix $K_T$ that can depend on the sample size and on the available data
	\[
	\hat k^{PAM,K_T} = 1 \iff \|K_T(\hat \vecb - \vecb_1) \|_2 <  \| K_T(\hat \vecb - \vecb_2) \|_2.  
	\]
	This includes the case of no weighting by setting $K_T$ to be the identity matrix. We will now compare those rules in a local alternative regime where $\vecb_1 = \vecb^*, \vecb_2 = \vecb^* + T^{-1/2}\bm{\Delta}$. Assume that the loss function $\Ll$ has the following properties.
	
	\begin{assumption}\label{Ass:gen}
		\begin{enumerate} Assume that $(\vecx_1,Y_1),\dots,(\vecx_T,Y_T)$ are i.i.d. and that further
			\item[(i)] The map $m: \vecg \mapsto \E[\Ll(\vecx_t,Y_t;\vecg)]$ is twice continuously differentiable in a neighbourhood of $\vecg^*$ with symmetric Hessian matrix $A_\vecg$ of full rank.%
			\item[(ii)] The map $g: \vecg \mapsto \Ll(\vecx_t,Y_t;\vecg)$ is differentiable at $\vecg^*$ on a set $\mathcal{Z}$ such that $\mpr((\vecx_t,Y_t) \in \mathcal{Z}) = 1$ and there exists a measurable function $\dot g$ such that almost surely $|\Ll(\vecx_t,Y_t;\vecg_1) - \Ll(\vecx_t,Y_t;\vecg_2)| \le \dot g(\vecx_t,Y_t) \|\vecg_1 - \vecg_2\|$ for all $\vecg_1,\vecg_2$ in a neighborhood of $\vecg^*$ and $\E[\dot g(\vecx_t,Y_t)^2] < \infty$.
			\item[(iii)] For any $\vecb$ in a neighbourhood $\mathcal{B}$ of $\vecb^*$ the function $\alpha \mapsto m(\alpha,\vecb)$ has a well separated (uniformly in $\vecb$) global minimizer $\alpha^*_\vecb$, i.e. for every $\eps > 0$ we have 
			\[
			\inf_{\vecb \in \mathcal{B}} \inf_{|\alpha - \alpha^*_\vecb| > \eps} \big(m(\alpha,\vecb) - m(\alpha^*_\vecb,\vecb) \big) > 0.
			\]
			\item[(iv)] The value $\vecg^*$ is in the interior of the parameter space $\Gamma$. Either the parameter space $\Gamma$ is compact or the parameter space is convex and the function $\vecg \mapsto \Ll(\vecx_t,Y_t;\vecg)$ is convex almost surely.   
		\end{enumerate}
	\end{assumption}
	
	It is routine to verify that all of the above conditions hold for two important examples that we will discuss throughout this paper: quantile regression and logistic regression. More generally, parts (i) and (ii) of the assumptions are fairly mild and standard conditions for establishing asymptotic normality and expansions for m-estimators, see for instance Theorem 5.23 and the discussion around it in \cite{vdV}. Conditions (iii) and (iv) are added because the proof relies not only on expansions for the original estimator but also for the minimizer of the perturbed objective $\sum_{t=1}^T \Ll(Y_t - \alpha - \vecx_t^\top \vecb)$ where $\vecb \neq \vecb^*$. We have opted for simple to state and verify conditions rather than the most general possible ones. The proof of Theorem~\ref{thm:loc} reveals that it is the expansions~\eqref{eq:expansionL}--\eqref{eq:ratetildealpha} in the proof rather than the specific conditions we state above that are needed to establish this result. Such expansions can also be established for data with serial dependence but we do not pursue this direction here as it does not add any insights to our main message. 
	
	To state the next result introduce some additional notation. For square matrices $M$ of dimension $p+1$ consider the following block structures 
	\[
	M = \left[ 
	\begin{array}{cc}
		M_{11} & M_{12}
		\\
		M_{21} & M_{22}
	\end{array}
	\right]
	\]
	with $M_{11} \in \R$.
	
	\begin{theorem}\label{thm:loc}
	Assume that Assumption~\ref{Ass:gen} holds and that $\vecb_1 = \vecb^*, \vecb_2 = \vecb^* + T^{-1/2}\bm{\Delta}$, $\bm{\Delta} \neq 0$. Let $A = A_{\vecg^*}$, $B = Var(\grad_\vecg \Ll(\vecx,Y;\vecg^*))$ and assume that $B$ is of full rank. Then $\sqrt{T}(\hat \vecb - \vecb^*) \Dkonv N(0,\Sigma_\vecb)$ where $\Sigma_\vecb = [A^{-1} B A^{-1}]_{22}$. Assume that $\hat \Sigma = \Sigma_\vecb + o_\mpr(1)$. Then
		\begin{equation}\label{eq:asympcompprob}
		\lim_{T\to \infty} \mpr\Big( \hat k^{PAM} = 1 \Big) \geq \lim_{T\to \infty} \mpr\Big( \hat k^{BM} = 1 \Big).
		\end{equation} 
		Define $D = [[A^{-1}]_{22}]^{-1}, C := B_{22} + \frac{B_{11}}{A_{11}^2}A_{21}A_{21}^\top - 2 \frac{A_{21}B_{21}^\top}{A_{11}}$. Equality in~\eqref{eq:asympcompprob} holds if and only if $C^{1/2}\bm\Delta$ is a scalar multiple of $C^{-1/2}D\bm\Delta$. Further, if $K_T \to K$ entry-wise in probability for a fixed matrix $K$ with finite entries
		\[
		\lim_{T\to \infty} \mpr\Big( \hat k^{PAM} = 1 \Big) \geq \lim_{T\to \infty} \mpr\Big( \hat k^{PAM,K_T} = 1 \Big).
		\]
	\end{theorem}
	
	A similar result under even weaker conditions continues to hold if there is no individual-specific $\alpha$ and all parameters are estimated globally. The proof of this result is similar in spirit but even simpler and we omit the details for the sake of brevity.  
	
	Note that when $\Ll$ is a correctly specified negative log-likelihood function, standard regularity conditions yield $A=B$ which further implies $C=D$ by the block matrix inversion formula. In this case $C^{-1/2}D = C^{1/2}$ so the asymptotic probabilities for rules~\eqref{eq:BMsimple} and~\eqref{eq:PAMsimple} selecting the correct center are equal for any $\bm\Delta$. Correct specification of $\Ll$ is sufficient but not necessarry. The equality $C=D$ continues to hold in the case where $A$ is a scalar multiple of $B$ which is the case in least squares or quantile regression with homoscedastic errors, for instance. However, in general models such as quantile regression or ordinary least squares estimation with heteroscedasticity or in the presence of temporal dependence, $A$ is not a scalar multiple of $B$ in general and thus also $C \neq D$. Since rule~\eqref{eq:PAMsimple} is always at least as good as~\eqref{eq:BMsimple} asymptotically, this suggests that~\eqref{eq:PAMsimple} would be preferable whenever the asymptotic covariance matrix can be estimated consistently, even when~\eqref{eq:BMsimple} is feasible.    
	
	The second statement of Theorem~\ref{thm:loc} implies that the proposed scaling with $\hat\Sigma^{-1/2}$ is asymptotically optimal among all possible choices of scale matrix that converge to a fixed matrix.
	
	Although the results presented above only work in a very idealized setting and can not be directly utilized to analyze the performance of rules~\eqref{eq:BMsimple} and~\eqref{eq:PAMsimple} when applied inside an iterative procedure, the findings strongly suggest that using the objective function in iteration for group centers might not be optimal from a statistical perspective. Instead, using information on the (asymptotic) variance of the estimators $\hat \vecb_i$ can lead to more efficient procedures. This motivates the ideas in the following section. 
	
	\begin{remark}
	The key to proving the first statement of Theorem~\ref{thm:loc} is an asymptotic expansion for the probabilities appearing in~\eqref{eq:asympcompprob}. Specifically, we derive the following limits
	\[
	P(\hat k^{PAM} = 1) \to \Phi(\|\Sigma_\vecb^{-1/2}\Delta_\vecb\|_2/2).
	\]
	in equation~\eqref{eq:limpobPAM} and 
	\[
	P(\hat k^{BM} = 1) \to \Phi\Big( \frac{\Delta\Big[ [A^{-1}]_{22}\Big]^{-1}\Delta}{2(\Delta^\top C \Delta)^{1/2}} \Big)
	\]
	in equation~\eqref{eq:limpobBM} in the proof of Theorem~\ref{thm:loc}. This is where the matrix $C$ comes into play. Given those expansions,~\eqref{eq:asympcompprob} follows by an application of the Cauchy-Schwarz inequality as follows
	\[
	\frac{\Delta^\top D \Delta}{(\Delta^\top C \Delta)^{1/2}} 
	= 
	\frac{\Delta^\top C^{1/2} C^{-1/2} D \Delta}{(\Delta^\top C\Delta)^{1/2}} 
	\leq 
	\frac{\|\Delta^\top C^{1/2}\|_2 \|C^{-1/2} D \Delta\|_2}{(\Delta^\top C \Delta)^{1/2}} 
	= 
	(\Delta^\top \Sigma_\vecb^{-1} \Delta)^{1/2}.
	\]
	This inequality is strict unless $C^{1/2}\Delta$ is a scalar multiple of $C^{-1/2} D \Delta$. 
	\end{remark}

	\subsection{Proposed methodology through the lens of clustering}\label{sec:methodology}
	
	The discussion up to this point focused on variants of the k-means algorithm for grouping individuals. However, k-means is not the only clustering method which is available and other approaches have been observed to have superior performance in certain settings. Many methods of this type work with general measures of dissimilarity between units and attempt to cluster units that are most similar to each other. Given the developments in the previous sections, a natural measure of dissimilarity is given by
	\begin{equation} \label{eq:hatVij}
		\hat V_{ij} := \norm{\hat \Sigma_{i,j}^{-1/2}(\hat\vecb_i - \hat\vecb_j) }_2\,,    
	\end{equation}
	where typically $\hat \Sigma_{i,j} = \hat \Sigma_i + \hat \Sigma_j$ and $\hat \Sigma_i$ estimates the variance of $\hat\vecb_i - \vecb$. Note that for consistent estimators $\hat \vecb_i$, $\hat \Sigma_i$ will typically converge to zero. This measure of dissimilarity can be computed based on summary statistics and variance estimates and does not require individual level data. The importance of taking variance information into account was illustrated in a simplified setting in Theorem~\ref{thm:loc} and is also confirmed in our simulations. As pointed out by the Associate Editor, using covariance estimates or diagonal versions thereof for $\hat \Sigma_i$ has the added benefit of making the procedure scale invariant.%

	Two popular clustering approaches in the literature that work with general measures of dissimilarity are \textit{$k$-medoids}~\cite{schubert2019faster} and~\textit{spectral clustering} \cite{ng2002, chung1997, von2007}. Similarly to $k$-means clustering, the $k$-medoids problem is NP-hard to solve exactly. In practice, approximate solutions to this problem are obtained by employing the algorithm \textit{Partitioning Around Medoids} (PAM)~\cite{reynolds2006clustering, schubert2019faster,kaufman2009finding}. We refer to~\cite[Section 4.1, Chapter 2]{kaufman2009finding} for more details about the PAM algorithm. %
	As we observe in simulations, using the PAM algorithm with dissimilarity measure~\eqref{eq:hatVij} can already lead to substantial gains relative to the iterative k-means style approaches of~\cite{LinNg,BM,zhang2019}. However, extensive simulations showed that in all settings considered spectral clustering leads to even more accurate group estimation than PAM, and hence we focus on spectral clustering in the theoretical developments that follow. Simulation evidence for the superiority of spectral clustering over to PAM is presented in Section~\ref{sec:sims}.%

	Since there are many variations of spectral clustering that are available in the literature, a detailed description of the specific version we use is given in Algorithm~\ref{SpectralClustering}\footnote{We do not claim any novel contributions to this specific algorithm, the details and explanation are presented here for the reader's convenience.}.
	
	\begin{algorithm}[ht]
		\caption{Spectral Clustering}\label{SpectralClustering}
		\hspace*{\algorithmicindent} \textbf{Input}: Number of clusters $G^*$, dissimilarity matrix $\hat V:=(\hat V_{ij})$ computed in~\eqref{eq:hatVij}. \\
		\hspace*{\algorithmicindent} \textbf{Output}:  Clusters $\hat I_1,\dots,\hat I_{G^*}$.
		\begin{algorithmic}[1]
			
			\State Compute the empirical adjacency matrix $\hat A \in \rnn$ with entries $\hat A_{ij} := e^{-\hat V_{ij}}$ for $i\neq j$ and $\hat A_{ij}=1$ for $i=j$. 
			
			\State Compute the empirical degree matrix $\hat D:=\operatorname{diag}(\hat{\Dmat}_1,\dots,\hat{\Dmat}_n),$ where $\hat{\Dmat}_i:=\sum_{j=1}^n \hat A_{ij},i=1,\dots,n.$ 
			
			\State Calculate the normalized graph Laplacian ~$\hat{L}:=\hat{\Dmat}^{-1/2}(\hat\Dmat-\hat{A})\hat{\Dmat}^{-1/2}.$ 
			
			\State Find $G^*$ orthonormal eigenvectors corresponding to the $G^*$ smallest eigenvalues of~$\hat{L}$, and form the matrix~$\hat Z \in \R^{n\times G^*}$ by stacking those vectors in columns. Normalize the rows of $\hat Z,$ to have $\ell^2$-norm $1$ and denote the resulting matrix by $\hat U.$
			
			\State Apply standard $k$-means clustering with $G^*$ clusters taking the rows of $\hat U$ as input vectors, and return the clusters $\hat I_1,\dots,\hat I_{G^*}.$ 
		\end{algorithmic}
	\end{algorithm}
	
	\normalsize
	
	To intuitively understand the motivation behind the above algorithm observe that the dissimilarities $\hat V_{ij}$ can be expected to be large if individuals $i,j$ are from different groups. In the limit $T \to \infty$ those distances will tend to infinity, and thus $\hat A_{ij} \approx 0$ whenever $i,j$ are from different groups. Similarly, $\hat V_{ij}$ can be expected to be bounded when $i,j$ are in the group, and thus $\hat A_{ij}$ will usually be bounded away from zero for such pairs. Thus after rearranging the order of individuals we see that $\hat V_{ij}$ will be approximately block diagonal with non-zero entries in the blocks. It is now straightforward to see that $\hat L$ will have exactly $G^*$ zero eigenvalues if there are $G^*$ such blocks and all other eigenvalues will be strictly positive. Moreover, the eigen-space corresponding to zero eigenvalues will have an orthogonal basis consisting of vectors that have non-zero entries in the exact components corresponding to different groups, see also the discussion surrounding equation~\eqref{eq:groupvect} and Lemma~\ref{eigenspace} in the supplementary material. For a more detailed discussion of the intuition and alternative formulations of the spectral clustering algorithm see \cite{von2007} and the literature cited therein. Although the last step of the algorithm uses the standard $k$-means algorithm, we note that it is applied on the rows of $\hat U$ which is a standard clustering problem with $n$ data points in Euclidean space. No refitting of models on individual level or large scale models as in \cite{BM} is required.
	
	Some additional comments on specific choices that we made in Algorithm 1 are in order. First, in step (1), we apply an exponential kernel to the dissimilarity matrix. Other monotone transformations can be used, for instance the Gaussian kernel is another popular choice. Our simulation exercise confirms that both the exponential kernel and the Gaussian kernel perform similarly. Second, in step (3), we apply a normalization to the graph Laplacian for the spectral clustering analysis. A line of seminal works (\cite{luxburg2004limits} and \cite{von2008consistency}) investigate the convergence of the normalized and unnormalized versions of the popular spectral clustering algorithm. They demonstrate that the normalized spectral clustering converges under very general conditions, while the unnormalized spectral clustering is only consistent under strong additional assumptions, which are not always satisfied in real data. These works give strong evidence for the superiority of normalized spectral clustering. 
	
	\begin{remark} \cite{wang2021} also observe that the spectral decomposition of a certain matrix that is derived from individual-specific estimators contains information on the latent group structure. However, there are several crucial differences between their and our approach. Most importantly, we explicitly take into account the uncertainty that is associated with individual-specific estimators while \cite{wang2021} work directly with raw estimators. Moreover, \cite{wang2021} do not apply spectral clustering directly but rather use certain eigenvectors as input to a binary segmentation algorithm. For a simulation-based comparison with that method, see section~\ref{sec:simlog}.  
		
	The idea to use spectral clustering for grouping different entities also appeared in \cite{vandelft2021}. The setting in the latter paper is very different from ours since \cite{vandelft2021} consider grouping locally stationary functional time series and do not take into account estimation uncertainty when constructing their dissimilarity measure between observations. Still, some parts of our theoretical analysis under high-level assumptions are related to theirs, additional comments on this can be found in Remark~\ref{rem:compmeth}. 
	\end{remark}
	So far we discussed an algorithm for assigning individuals to $G^*$ groups for any given $G^*$. In some settings, $G^*$ will be chosen based on domain knowledge about the problem at hand. If no such knowledge is available, we propose to select the $G^*$ that maximizes the relative eigen--gap (\cite{von2007}) of a modified graph Laplacian~$\tilde L.$ More precisely, consider the scaled dissimilarity~$\tilde V_{ij}:=\frac{2}{\sqrt{\log n\log T} }\hat V_{ij}.$ Use $\tilde V_{ij}$ as input to Algorithm~\ref{SpectralClustering} and obtain $\tilde L$ as output from step 3 of that algorithm. Consider the values $\tilde\lambda_i:=1-\hat\lambda_i,i=1,\dots,n,$ with $\hat\lambda_1\le\cdots\le \hat\lambda_n$ denoting the ordered eigenvalues of $\tilde L$. The estimated number of groups is   
	\begin{equation}\label{eq:selGstar}
		\hat G = \argmax_{g=1,\dots,n-1} \frac{|\tilde\lambda_{g+1}-\tilde\lambda_{g}|}{\tilde\lambda_{g+1}}\,,
	\end{equation}
	The motivation for using the scaling in $\tilde V_{ij}$ is that, under technical assumptions made later, this scaling ensures $\tilde V_{ij} \to 0$ for all $i,j$ in the same group. Without this scaling, the heuristic tends to have a small probability of not selecting a correct number of groups as $T$ increases.  
	
	Similar heuristic eigen-gap methods for estimating the number of groups can also be found in~\cite{vandelft2021, john2020spectrum, little2017path}, among many others. 
	
	\begin{remark}
		There are at least two other popular approaches to selecting the number of groups or equivalently the number of clusters. The first type of method combines cross-validation with the idea that ``true" cluster assignment should be stable under small perturbations of the data. This idea was exploited in \cite{wang2010} for selecting the number of clusters in a general setting and adapted by \cite{zhang2019} to selecting the number of groups for panel data quantile regression. However, as pointed out in \cite{ben2006sober}, methods that select the number of clusters based on stability can fail for certain cluster configurations. One such example will be given in the simulation section, see Model 2 in section~\ref{sec:qrslope} . The second drawback of such methods is that clustering stability can only be defined when there are at least two clusters. Hence, by construction, stability methods always select at least two clusters and fail if there is only a single cluster in the data.
		
		The second method uses information criteria which select the number of clusters that maximize a sum of objective function plus penalty, see for instance \cite{SSP,gu2019panel,wang2021} among many others. The main drawback of such approaches is that information criteria need to be derived case by case as they differ depending on the specific form of the objective function making them difficult to use for applied researchers. We note that this is different from the classical setting involving AIC and BIC in a maximum likelihood framework where only the number of parameters in the model matters. Moreover, computation of such information criteria typically requires access to raw data which might not always be available as in our second application. {The information criteria method also involves the heaviest computation burden because to construct the information criteria statistics, all candidate models with varying values of $G$ need to be estimated which can be costly (See computation time comparison in Section \ref{sec:simlog}). }
		
		We also conduct an extensive simulation comparing different methods of selecting the number of groups in Section \ref{sec:simlog} and \ref{sec:simnumgr}. Results show that our heuristic approach works reasonably well in all settings considered. Unsurprisingly, we also find that there is no universally dominating method.
	\end{remark}
	
	\subsection{Examples}\label{sec:ex}
	
	The setting above is generic and so far we did not assume anything about the specific structure of the estimators. In the remainder of this section, we provide several illustrative examples of model specifications that were considered previously and show how those examples fit into the proposed framework.

	\begin{example}[\textit{Logistic regression regression with individual-specific intercepts and grouping on slopes}] \label{ex:logreg}
		Consider binary responses $Y_{it} \in \{0,1\}$ and assume that
		\[
		\mpr(Y_{it} = 1) = \frac{\exp(\alpha_i + \vecx_{it}^\top\vecb_i )}{1+\exp(\alpha_i + \vecx_{it}^\top\vecb_i)} = \frac{\exp( \vecz_{it}^\top\vecg_i )}{1+\exp(\vecz_{it}^\top\vecg_i)}\,,
		\] 
		where $\vecz_{it}^\top = (1,\vecx_{it}^\top)$ and $\vecg_i^\top = (\alpha_i,\vecb_i^\top)$.
		We leave the $\alpha_i\in\R$ unrestricted and assume that certain sub-vectors of $\vecb_i\in\R^{p}$ have a group structure. 
	\end{example}
	
	\cite{SSP} considers a similar model; they assume a Gaussian link function for the binary response. \cite{ando2021bayesian} also considers the logit model with individual specific slope coefficients and a factor structure on the individual fixed effects. Their way of modeling unobserved heterogeneity is different from ours as we focus on group patterns. 
	
		\begin{example}[\textit{Quantile regression with individual-specific intercepts and grouping on slopes}] \label{ex:qrslope}
		Given the observations are $(\vecx_{it},Y_{it}),$
		assume that the conditional quantile function of the response $Y_{it}$ given covariates $\vecx_{it}$ for individual $i$ satisfies
		\[
		q_{i, \tau}( \vecz_{it}) = \alpha_i(\tau) + \vecx_{it}^\top\vecb_i(\tau) = \vecz_{it}^\top\vecg_i(\tau)\,,
		\]  
		where $\alpha_i(\tau)\in\R$ are unrestricted and we search for a group structure on $\vecb_i(\tau)\in\rp$. 
	\end{example}
	
	This setting was also considered in \cite{zhang2019}, \cite{Leng}. \cite{zhang2019} propose an iterative algorithm based on the $k$-mean algorithm in \cite{BM} to learn group structure. \cite{Leng} use a $k$-means type of iterative algorithm, but allow for time fixed effect while grouping both the individual fixed effects and the slope coefficients. This model will be considered in Section~\ref{sec:data} where coefficients of the panel quantile regression model will be utilized to analyze heterogeneous relationship between income and pollution level among different states in  the US.

	\begin{example}[\textit{Quantile regression with joint slope and grouping on intercepts}] \label{ex:qrint}
		Assume that the conditional quantile function of response $Y_{it}$ given covariates $\vecx_{it}$ for individual $i$ is
		\[
		q_{i, \tau}( \vecx_{it}) = \alpha_i(\tau) + \vecx_{it}^\top\vecb(\tau)\,,
		\]  
		where the vector of slope coefficients $\vecb(\tau)\in\rp$ is assumed to be the same across individuals. 
	\end{example}  
	
	This model was first considered in \cite{koenker2004quantile}, who proposed to regularize the individual fixed effects via $\ell_1$ penalization. \cite{Lamarche} considers the optimal choice of the penalty parameters in this approach. There has been an active literature on panel data quantile regression, mainly focusing on estimation of the common parameters $\beta(\tau)$ (e.g., \cite{kato2012asymptotics}, \cite{galvao2016}, \cite{harding2017} and \cite{galvao2018}). \cite{zhang2019subgroup} and \cite{gu2019panel} consider group structure on $\alpha_i(\tau)\in\R$. %

	\section{Theoretical Analysis}

\subsection{Generic spectral clustering results} \label{sec:theoryhighlevel}
	
In this section, we provide high-level conditions on the estimators $\hat \vecb_i\in\rp$ and $\hat \Sigma_{i,j}\in\rpp$ which ensure that the correct group structure is recovered with probability tending to one as $n,T$ tend to infinity. Formally, assume that the true coefficients $\vecb_1,...,\vecb_n$ take $G^*$ different values, say $\vecb_1^*,...,\vecb_{G^*}^*$ and the true group membership is given by
\[
\vecb_i = \vecb_k^* \quad \Leftrightarrow \quad i \in I^*_k, \quad k = 1,...,G^*\,,
\]   
where $I^*_k\subseteq\{1,\dots,n\}, k=1,\dots,G^*$ denote the true underlying groups. Naturally, we assume $I^*_k \cap I^*_\ell = \emptyset$ for $k \neq \ell$. We begin by providing an analytical non-asymptotic result which guarantees perfect classification in terms of certain abstract quantities. More precisely, define
\begin{align*}
A_{1,max} &:= \max_{i,j \mbox{ in different groups}} \hat A_{ij}\,,
\\
A_{0,min} &:= \min_{i,j \mbox{ in same group}} \hat A_{ij}\,
\\
A_{0,max} &:= \max_{i,j \mbox{ in same group}} \hat A_{ij}\,.
\end{align*}

\begin{theorem}\label{th:specnonasy}
A sufficient condition for perfect classification is 
\begin{equation}\label{eq:nonasyperfectA}
\frac{A_{1,max}}{A_{0,min}}  \sqrt{\frac{A_{0,max}^3}{A_{0,min}^3}}
\le 2^{-8.5} (n G^*)^{-1/2} \sqrt{\frac{\min_k |I_k^*|^3}{n \max_k |I_k^*|^2} }
\end{equation}
\end{theorem}	

Theorem~\ref{th:specnonasy} holds for fixed $n,T$ and is proved in a purely analytic way. The result does not assume anything about temporal or cross--sectional dependence. On a high level, this result corresponds to intuition as the inequality in~\eqref{eq:nonasyperfectA} becomes more difficult to satisfy for a larger number of groups $G^*$ or when groups have more unbalanced sizes leading to a larger ratio $\max_k|I_k^*|/\min_k |I_k^*|$. Having more individuals (larger $n$) also intuitively makes the problem harder. In order to achieve perfect classification, a large minimal dissimilarity between individuals from different groups, i.e. a small $A_{1,max}$, relative to $A_{0,min}$, is required. The ratio $\frac{A_{0,max}}{A_{0,min}}$ describes the spread of similarity measures among individuals that belong to the same group. Having a large spread here makes the problem harder, which again corresponds to intuition. Note that this is only a sufficient condition, and sharper results might be possible.  However, we are not aware of any necessary and sufficient conditions guaranteeing the success of spectral clustering or sharp expansions for the proportion of correctly grouped units. %

\begin{remark}\label{rem:compmeth}
The proof relies on the type of arguments that appeared in earlier work on spectral clustering, in particular \cite{ng2002}, \cite{von2007} and \cite{vandelft2021}. However, the setting we consider is different from any of the works mentioned above and the arguments need to be modified accordingly. The work of \cite{vandelft2021} is closest in spirit, but our analysis is complicated by the fact that we allow the number of individuals $n$ to diverge while the number of entities to be clustered was fixed in \cite{vandelft2021}. In order to deal with this complication, we leverage the fact that our construction of the similarity matrix gives rise to the different order for the diagonal blocks and off-diagonal blocks of the empirical Laplacian matrix. Taking advantage of this difference in order together with spectral information contained in the diagonal blocks of the empirical Laplacian matrix allows us to handle a diverging number of individuals.
\end{remark}
	
Below, we will provide more specific assumptions on the minimal separation of group centers and quality of initial estimators which guarantee that the probability of the events in~\eqref{eq:nonasyperfectA} tend to one. In the assumptions below, we allow for data from triangular arrays where the values of $\vecb_{i}$ and $\Sigma_{i,j}$ change with $n,T$. To keep the presentation simple this is not emphasized in the notation. We also allow the number of groups $G^*$ to grow with $n$. %

\begin{assumption}\label{asm:est_vec}
The estimators $\estb_i$ are uniformly consistent with rate $a_{n,T}$, i.e. 
\[
a_{n,T} := \sup_{i\in\{1,\dots,n\}} \,\norm{\estb_i-\vecb_i}_2 = o_\mpr(1). %
\]
\end{assumption}

	\begin{assumption}\label{asm:est_cov}
	There exists a sequence $b_T \to \infty$ and matrices $\Sigma_{i,j}$ (which may depend on $n,T$) such that
	$$
	\sup_{i,j} \normspec{b_T \hat \Sigma_{i,j} - \Sigma_{i,j}}  = o_\mpr(1)\,,
	$$
	where $\normspec{\cdot}$ denotes the spectral norm.
	Moreover, assume
	\begin{equation}\label{eq:ev-sigmaij}
		0<m < \lambda_{\min}(\Sigma_{i,j})\le \lambda_{\max}(\Sigma_{i,j}) < M<\infty \quad \forall i\neq j \in\{1,\dots,n\} 
	\end{equation}
	with some fixed constants $0 < m \leq M < \infty$ that do not depend on $n,T$. 
\end{assumption}

Assumptions~\ref{asm:est_vec} and~\ref{asm:est_cov} impose minimal restrictions on the quality of the initial estimates $\estb_i$ and $\hat \Sigma_{i,j}$. We emphasize that the matrices $\Sigma_{i,j}$ in Assumption~\ref{asm:est_cov} are not required to be equal to the true asymptotic covariance matrices of $\estb_i - \estb_j$ for the theory to work. This is a useful result because in some environments researchers only have access to individual estimates and the associated coordinate-by-coordinate standard deviation, but the covariances estimate are missing. In such cases, our method can still be used by setting the off-diagonal elements of $\hat \Sigma_{i,j}$ to zero.  Assumption \ref{asm:est_cov} will hold provided that the variance estimators on the diagonal converge to non-negative values. While setting off-diagonal entries to zero might not be optimal in the asymptotic setting of Theorem~\ref{thm:loc}, simulations indicate that in finite samples the performance can be close to using consistent estimates of the covariance. When covariances are difficult to estimate, using only the diagonal entries can even enhance finite sample performance as we will see in Section~\ref{sec:simlog}. Similarly, this assumption can be satisfied if there is dependence across individuals but this dependence is ignored when estimating the covariance of $\estb_i - \estb_j$. Again, ignoring this dependence will not lead to procedures with best possible performance but might work reasonably well if the dependence across individuals is mild.

In all examples we consider later the individuals will be assumed independent and the estimators $\estb_i$ will satisfy 
$
\sqrt{T}(\estb_i - {\vecb_i}) \stackrel{\mathcal{D}}{\longrightarrow} \mathcal{N}_p(0,\Sigma_i), i = 1,\dots,n.
$
By independence among individuals, the weak convergence above holds jointly for any given pair of individuals and the corresponding limits will be independent. In this case, we will set $b_T = T$, $\Sigma_{i,j} := \Sigma_i+\Sigma_j$, $\hat\Sigma_{i,j} := \hat \Sigma_i + \hat\Sigma_j$ where $T\hat\Sigma_i$ will be consistent estimators of $\Sigma_i$. 

The bound in $a_{n,T}$ is uniform over a potentially growing number of individuals $n$ and typically be of the form $a_{n,T} = \bigo_\mpr(\sqrt{T^{-1}\log n})$ where the additional $\sqrt{\log n}$ factor is to ensure uniformity. 

We now have the following result
\begin{theorem}\label{thm:clusternew}
Under Assumptions~\ref{asm:est_vec}, \ref{asm:est_cov} let $\Delta_{min} := \min_{k \neq \ell} \|\vecb_k^* - \vecb_\ell^*\|_2$. Assume that $a_{n,T} = o_\mpr(\Delta_{min})$, $n \geq 3$ and 
\begin{equation}\label{eq:assminsep}
\log n = o(b_T^{1/2}\Delta_{min}).  
\end{equation}
Then the true group structure is recovered with probability tending to one as $T \to \infty$.
\end{theorem}

In order to achieve perfect classification with probability going to one, Theorem~\ref{thm:clusternew} requires lower bounds on the minimal separation $\Delta_{min}$ which is required to grow faster than the uniform estimation error and than $b_T^{-1/2}\log n$. In the special setting discussed above the Theorem where $b_T = T, a_{n,T} = \bigo_\mpr(\sqrt{T^{-1}\log n})$, this corresponds to assuming that $\Delta_{min} \gg T^{-1/2}\log n$. For groups with fixed separation across centers where $\Delta_{min}$ is a constant, this leads to the requirement $\log n = o(T^{1/2})$ which allows the number of individuals to grow very quickly with $n$. If the minimal separation tends to zero, the requirements on $\log n$ relative to $\sqrt{T}$ become more stringent. 

Given the non-asymptotic bound in~\eqref{eq:nonasyperfectA}, it would also be possible to conduct a more detailed analysis in the case where the orders of $\Delta_{min}$ and $a_{n,T}$ match but $\Delta_{min}$ is sufficiently large so as to dominate a constant multiple of $a_{n,T}$ with a certain probability. Such an analysis would reveal more nuanced view on the role of $G^*$ and $\max_k |I_k^*|, \min_k |I_k^*|$ but does not lead to any specific insights except that large $G^*$ and imbalanced groups make the problem harder.

\subsection{Verification of high level conditions for specific examples} \label{sec:theoryexamples}
	
	In this section, we provide specific conditions in Example~\ref{ex:logreg}--Example~\ref{ex:qrslope} which guarantee that the high-level conditions~$\ref{asm:est_vec}$ and~\ref{asm:est_cov} are satisfied. The set of examples that we consider is by no means exhaustive for the possible applications of our methodology. Rather, it is intended as a demonstration that our high-level conditions can be verified in several different settings including the presence of individual-specific and joint parameters, binary outcomes, and non-smooth objective functions.   
	
	\subsubsection{Logistic regression with individual-specific intercepts and grouping on the slopes (Example~\ref{ex:logreg})}\label{sec:logtheory}
	The coefficient vector $\vecg_i^\top := (\alpha_i,\vecb_i^\top)$ is estimated via maximum likelihood, i.e.
	\[
	\est_i := \argmax_{\vecg\in\R^{p+1}} \frac{1}{T}\sum_{t=1}^T\Big[Y_{it}\vecz_{it}^\top\vecg-\log(1+\exp(\vecz_{it}^\top\vecg))\Big], \quad i = 1,\dots,n\,.
	\] 
	{The exact form of the asymptotic variance differs depending on whether the data exhibit temporal dependence. We begin by discussing the case that the observations $(\vecx_{it},Y_{it})$ are i.i.d. across $t$ and independent across $i$ and discuss the case with temporal dependence across $t$ later in this section. Throughout, the values of $\vecg_i^*$ are allowed to depend on $n,T$.} 
	
	Recall that in the i.i.d. case under standard assumptions the estimator $\hat \gamma_i$ is asymptotically normal with asymptotic variance given by
	$$
	\widetilde{\Sigma}_i = \Biggl( \E\Big[ \frac{e^{\vecz_{i1}^\top\vecg_i^*}}{(1+e^{\vecz_{i1}^\top\vecg_i^*})^2}\vecz_{i1}\vecz_{i1}^\top  \Big] \Biggr)^{-1}.
	$$
	The canonical plug-in estimator of~$\widetilde{\Sigma}_i$ takes the form
	$$
	\hat{\widetilde \Sigma}_i = \Biggl(\frac{1}{T}\sum_{t=1}^T \frac{e^{\vecz_{it}^\top\hat\vecg_i}}{(1+e^{\vecz_{it}^\top\hat\vecg_i})^2}\vecz_{it}\vecz_{it}^\top\Biggr)^{-1}.
	$$
	Denote by $\check \Sigma_i$ the lower $p\times p$ sub-matrix of $\hat{\widetilde \Sigma}_i$. Then we set
	\begin{equation}\label{eq:hatsigma-log}
		\hat \Sigma_{i,j} := T^{-1}(\check \Sigma_i + \check \Sigma_j)\,.
	\end{equation}

	Consider the following assumptions.
	\begin{assumption}\label{asm:logistic_eigenvalue} Assume that for a constant $L >0$ independent of $i,n,T$
	\[
	1/L<\Big\{\lambda_{\min}\bigl( \E\bigl[ \vecz_{i1}\vecz_{i1}^\top \bigr]\bigr)\Big\}
	< \Big\{\lambda_{\max}\bigl( \E\bigl[ \vecz_{i1}\vecz_{i1}^\top \bigr]\bigr)\Big\}
	< L
	\]
	and there exists $\kappa_1 < \infty$ independent of $n,T$ such that $\sup_i \|\vecg_i^*\| \leq \kappa_1$.
	\end{assumption}

	\begin{assumption}\label{asm:logistic_bounded}
		Assume $\sup_{i,t}\{\,\norm{\vecz_{it}}_2\}< \kappa <\infty$ a.s. for a constant $\kappa$ that does not depend on $n,T$.
	\end{assumption}

	Assumption~\ref{asm:logistic_eigenvalue} places mild restrictions on the design matrix. The boundedness condition in Assumption~\ref{asm:logistic_bounded} is made for the sake of simplicity; it can be relaxed to designs with bounded moments at the cost of additional technicalities in the proofs.

	\begin{theorem}\label{mle_est}
		Assume Assumptions~\ref{asm:logistic_eigenvalue} and \ref{asm:logistic_bounded} hold, that data are i.i.d. across $t$ and independent across $i$, and $T \to \infty, \log n/T \to 0$. \\
		\noindent (i)  It holds that
		\begin{equation}\label{eq:log-mle1}
			\sup_{i\in\{1,\dots,n\}}\normtwo{\hat\vecg_i-\vecg_i^*}=\bigo_\mpr\Bigg(\sqrt{\frac{\log n}{T}}\Bigg)\,.
		\end{equation}
		\noindent (ii)
		Under the same assumptions the estimators $\hat \Sigma_{i,j}$ in~\eqref{eq:hatsigma-log} satisfy 
		\begin{equation}\label{eq:log-mle2}
			\sup_{i \neq j} \normspec{T\hat \Sigma_{i,j} - \Sigma_{i,j}} = o_\mpr(1)\,,
		\end{equation}
		\noindent where $\Sigma_{i,j}$ denotes the lower $p\times p$ submatrix of $\widetilde \Sigma_i + \widetilde \Sigma_j$. Furthermore $\Sigma_{i,j}$ satisfy~\eqref{eq:ev-sigmaij}\,.
	\end{theorem}
	Theorem~\ref{mle_est} implies that Assumptions~\ref{asm:est_vec} and~\ref{asm:est_cov} hold with $a_{n,T} = \sqrt{T^{-1}\log n}$, $b_T = T$ and $\Sigma_{i,j}$ corresponding to the scaled asymptotic variance matrix of $\hat\vecb_i - \hat\vecb_j$. In particular,~\eqref{eq:assminsep} is satisfied provided that $\Delta_{min} \gg (\log n)/\sqrt{T}$. 
	
	Note that the results directly imply that Assumptions~\ref{asm:est_vec} and~\ref{asm:est_cov} continue to hold for any sub-vectors of $\hat\vecg_i$. This covers settings where we want to leave some coefficients individual-specific and only perform grouping on a part of the full coefficient vector. 
	
	{We now proceed to consider the case of temporal dependence. 
		\begin{assumption}
			\label{B1} For each $i \geq 1$, the process $(\vecx_{it},Y_{it})_{t \in \Z}$ is strictly stationary and $\beta$-mixing. 
			Let $\beta_{i}(j)$ denote the $\beta$-mixing coefficient of the process $(\vecx_{it},Y_{it})_{t \in \Z}$. Assume that there exist constants $b_\beta \in (0,1), C_\beta > 0$ independent of $i,n,T$ such that
			\begin{equation*}
				\sup_i \beta_i(j) \leq \beta(j), \quad \forall j \geq 1,
			\end{equation*}
			where $\beta(j):=C_{\beta} b_\beta^j.$
		\end{assumption}
		Such exponential mixing assumptions are often made in the literature, see for instance \cite{kato2012asymptotics} and \cite{galvao2018} in the context of quantile regression. }
	
	{
		The available data $(\vecx_{it},Y_{it})_{t = 1,\dots, T}$ are an observed stretch from the strictly stationary process $(\vecx_{it},Y_{it})_{t \in \Z}$. Under this assumption, the asymptotic variance of the estimator~$\est_i$ is of the form
		\[
		\widetilde{\Sigma}_i=B_i^{-1}H_iB_i^{-1}
		\]
		with
		\begin{align*}
			B_i&:= \E\Bigg[ \frac{e^{\vecz_{i1}^\top\vecg_i^*}}{(1+e^{\vecz_{it}^\top\vecg_i^*})^2}\vecz_{i1}\vecz_{i1}^\top  \Bigg] \\
			H_i&:= \E\big[\vecw_{i1}\vecw_{i1}^\top \big]+\sum_{j=1}^{\infty}\E \big[\vecw_{i1}\vecw_{i,1+j}^\top+\vecw_{i,1+j}\vecw_{i1}^\top \big] \,,
		\end{align*}
		where $\vecw_{it}:=Y_{it}\vecz_{it}-\frac{e^{\vecz_{it}^\top\vecg_i^*}\vecz_{it}}{1+e^{\vecz_{it}^\top\vecg_i^*}}$.
		A possible sandwich estimator of the asymptotic variance~$\widetilde{\Sigma}_i$ takes the form
		\[
		\hat{\widetilde \Sigma}_i = \widehat{B}_{iT}^{-1}\widehat{H}_{iT}\widehat{B}_{iT}^{-1}\,,
		\]
		where
		\begin{align*}
			\widehat{B}_{iT}&=  \frac{1}{T}\sum_{t=1}^T \frac{e^{\vecz_{it}^\top\hat\vecg_i}}{(1+e^{\vecz_{it}^\top\hat\vecg_i})^2}\vecz_{it}\vecz_{it}^\top \\
			\widehat{H}_{iT}& =  \frac{1}{T}\sum_{t=1}^T  \widehat \vecw_{it} \widehat \vecw_{it}^\top 
			+\sum_{1\le j\le m_T} \Big(1-\frac{j}{T}\Big) \Bigg(\frac{1}{T}\sum_{t=1}^{T-j}\big( \widehat \vecw_{it} \widehat \vecw_{i,t+j}^\top + \widehat \vecw_{i,t+j} \widehat \vecw_{it}^\top \big) \Bigg)\\
			\widehat \vecw_{it}& = Y_{it}\vecz_{it}-\frac{e^{\vecz_{it}^\top\est_i}\vecz_{it}}{1+e^{\vecz_{it}^\top\est_i}}\,,
		\end{align*}
		and $m_T>0$ denotes the bandwidth parameter tending to be infinity as $T$ goes to infinity.
		Denote by $\widehat \Sigma_i$ the lower $p\times p$ sub-matrix of $\hat{\widetilde \Sigma}_i$. Then we set
		\begin{equation}\label{eq:hatsigma-log1}
			\hat \Sigma_{i,j} := T^{-1}(\widehat \Sigma_i + \widehat \Sigma_j)\,.
		\end{equation}
	}
	
	{
		\begin{theorem}\label{mle_est1}
			Let Assumptions~\ref{asm:logistic_eigenvalue}, \ref{asm:logistic_bounded} and~\ref{B1} hold.
			Assume $T$ grows at most polynomially in $n$ and $(\log n)^3 = o(T)$. Assume that the smallest and largest eigenvalues of $H_i$ are bounded away from zero an infinity uniformly in $n,T$. \\
			\noindent (i)  It holds that
			\begin{equation}\label{eq:log-mle1}
				\sup_{i\in\{1,\dots,n\}}\normtwo{\hat\vecg_i-\vecg_i^*}=\bigo_\mpr\Bigg(\sqrt{\frac{\log n}{T}}\Bigg)\,.
			\end{equation}
			\noindent (ii)
			In addition, if $m_T\to\infty$ as $T \to \infty$ and $\frac{m_T^3\log (n\vee m_T)}{T}=o(1)$, the estimators $\hat \Sigma_{i,j}$ satisfy 
			\begin{equation}\label{eq:log-mle2}
				\sup_{i \neq j} \normspec{T\hat \Sigma_{i,j} - \Sigma_{i,j}} = o_\mpr(1)\,,
			\end{equation}
			\noindent where $\Sigma_{i,j}$ denotes the lower $p\times p$ submatrix of $\widetilde \Sigma_i + \widetilde \Sigma_j$. Furthermore $\Sigma_{i,j}$ satisfy~\eqref{eq:ev-sigmaij}\,.
		\end{theorem}
		Theorem~\ref{mle_est1} implies that Assumptions~\ref{asm:est_vec} and~\ref{asm:est_cov} hold with $a_{n,T} = \sqrt{T^{-1}\log n}$, $b_T = T$ and $\Sigma_{i,j}$ corresponding to the scaled asymptotic variance matrix of $\hat\vecb_i - \hat\vecb_j$. 
		 In particular,~\eqref{eq:assminsep} is satisfied provided that $\Delta_{min} \gg (\log n)/\sqrt{T}$ and we need the additional condition $(\log n)^3 = o(T)$. }

	\subsubsection{Quantile regression with individual-specific intercepts and grouping on the slopes (Example~\ref{ex:qrslope})} \label{sec:qrslope}

	Consider the quantile regression panel data model 
	\begin{equation*}
		q_{i,\tau}(\vecz_{it})= \vecz_{it}^\top\vecg^*_i(\tau), ~~t=1,\dots,T, i=1,\dots,n\,,
	\end{equation*}
	where $q_{i,\tau}(\vecz_{it}):=\inf\bigl\{y: \mpr(\response_{it}<y|\vecz_{it})\ge\tau\bigr\}$ is the conditional $\tau\,$-\,quantile of $\response_{it}$ given $\vecz_{it}.$

	We will {first assume that $(\vecz_{it},Y_{it})$ are i.i.d. across $t$ for each $i$ and independent across $i$. An extension to temporal dependence as in Assumption~\ref{B1} will be considered below. The distribution of $(\vecz_{it},Y_{it})$ and the values of $\vecg_i$ are allowed to vary with $n,T$.} 
	
	Consider the quantile regression estimator~$\est_i^\top=(\hat\alpha_i,\hat\vecb_i^\top):$
	\begin{equation*}
		\est_i :=\argmin_{\vecg\in\R^{p+1}}\frac{1}{T}\sum_{t=1}^T \rho_{\tau} (\response_{it} - \vecz_{it}^\top\vecg)\,,
	\end{equation*}
	where $\rho_{\tau}(u):=\{\tau-\1(u\le 0)\}u$ denotes the check function.

	Under mild regularity assumptions (in particular, this is true under Assumptions~\ref{A1}--\ref{A3} given below) this estimator is asymptotically normal with asymptotic covariance matrix of the form $\widetilde\Sigma_i = B_i^{-1}H_iB_i^{-1}$ where
	\begin{equation}
		\label{eq:AiBi}
		H_i := \tau (1-\tau) \E[\vecz_{i1}\vecz_{i1}^\top], \quad B_i = \E[f_{\response_{i1}\mid \vecz_{i1}}(q_{i,\tau}(\vecz_{i1})\mid \vecz_{i1})\vecz_{i1}\vecz_{i1}^\top]\,,
	\end{equation}
	with $f_{\response_{i1}|\vecz_{i1}}(y|\vecz)$ as the density function of the conditional distribution $F_{\response_{i1}|\vecz_{i1}}(y|\vecz).$ 
	
	A common way to estimate $\tilde\Sigma_{i}$ uses the Hendricks-Koenker sandwich covariance matrix estimator (\cite{HK1992}) which takes the following form
	\begin{equation}
		\label{eq:HKvar}
		\hat{\widetilde\Sigma}_{iT} := \widehat{B}_{iT}^{-1} \widehat{H}_{iT} \widehat{B}_{iT}^{-1}\,,~~\text{with}
	\end{equation}
	\begin{align*}
		\label{eq3}
		\widehat{B}_{iT} :=  \frac{1}{T} \sum_{t=1}^{T} \widehat{f}_{it} \vecz_{it}\vecz_{it}^\top\,,
		~\widehat{H}_{iT} := \tau (1-\tau)\frac{1}{T} \sum_{t=1}^{T}  \vecz_{it}\vecz_{it}^\top\,,
		~\widehat{f}_{it}  := \frac{2d_T}{\vecz_{it}^\top (\est_{i}(\tau + d_T) - \est_{i}(\tau - d_{T}))}\,.
	\end{align*}
	Here $d_T$ denotes a smoothing parameter that should converge to zero at an appropriate rate. Let $\hat \Sigma_{iT}$ denote the lower $p \times p$ submatrix of $\hat{\widetilde\Sigma}_{iT}$
	and set 
	\begin{equation}\label{eq:hatsigma}
		\hat \Sigma_{i,j} := T^{-1}\Big(\hat \Sigma_{iT} + \hat \Sigma_{jT}\Big)\,.
	\end{equation}
	
	We now verify Assumptions~\ref{asm:est_vec} and~\ref{asm:est_cov}, under the following conditions.\\
	
	\begin{assumption}
		\label{A1} Assume that $\normtwo{\vecz_{it}}\leq \kappa < \infty$, and that $c_{\lambda}\leq\lambda_{\min}(\E [\vecz_{it} \vecz_{it}^\top])\leq\lambda_{\max}(\E [\vecz_{it} \vecz_{it}^\top])\leq C_{\lambda}$ holds uniformly in  $i$ for some fixed constants $c_{\lambda}>0$ and $\kappa,C_{\lambda} <\infty$ that are independent of $n,T$.\\
	\end{assumption}
	
	\begin{assumption}
		\label{A2} {Define $\mathcal{Z}:=[-\kappa,\kappa]^{p+1}$.}The conditional distribution $F_{\response_{i1}|\vecz_{i1}}(y|\vecz)$ is twice differentiable w.r.t. $y$, with the corresponding derivatives $f_{\response_{i1}|\vecz_{i1}}(y|\vecz)$ and $f'_{\response_{i1}|\vecz_{i1}}(y|\vecz)$. 
		Assume that 
		$$
		\sup_i \sup_{y \in \mathbb{R},\vecz\in \mathcal{Z}} |f_{\response_{i1}|\vecz_{i1}}(y|\vecz)| \leq f_{max} < \infty, \quad
		\sup_i \sup_{y \in \mathbb{R},\vecz\in \mathcal{Z}} |f'_{\response_{i1}|\vecz_{i1}}(y|\vecz)| \leq \overline{f'}  < \infty.
		$$
		where $f_max, \overline{f'}$ are independent of $n,T$.
	\end{assumption}
	
	\begin{assumption}
		\label{A3} Denote by $\mathcal{T}$ an open neighborhood of $\tau$. Assume that uniformly across $i$, there exists a constant $f_{\min} < f_{\max}$ independent of $n,T$ such that 
		$$
		0 < f_{\min} \leq \inf_i \inf_{\eta \in \mathcal{T}} \inf_{\vecz \in \mathcal{Z}} f_{\response_{i1}|\vecz_{i1}}(q_{i,\eta}(\vecz)|\vecz)\,.
		$$
	\end{assumption}
	
	\begin{assumption}
		\label{A4} 
		Assume that $d_T = o(1)$, as $T \to \infty$ and 
		$$
		\frac{\log(nT)}{T d_T^{4/3}} = o(1)\,.
		$$
	\end{assumption}
	
	Assumptions~\ref{A1}-\ref{A4} are fairly standard in the quantile regression literature and have been imposed in \cite{kato2012asymptotics} and \cite{galvao2018} among many others.

	\begin{theorem}\label{quantreg_est}
		Let Assumptions~\ref{A1}-\ref{A3} hold.
		Assume $\log n = o(T)$ and $\min(n,T) \to \infty$. Assume that the data are i.i.d. across $t$ and independent across $i$.\\
		\noindent (i)  It holds that~
		$$
		\sup_{i\in\{1,\dots,n\} }\normtwo{\est_i-\vecg^*_i}=\bigo_\mpr\biggl(\sqrt{\frac{\log n}{T}}\biggr).
		$$
		In particular, Assumption~\ref{asm:est_vec} holds with $a_{n,T}=\sqrt{\frac{\log n}{T}}$ provided that $\log n = o(T)$. \\
		\noindent (ii)
		If in addition to the above Assumption~\ref{A4} holds, then Assumption~\ref{asm:est_cov} is also satisfied with $b_T := T$, $\Sigma_{i,j}$ denoting the lower $p\times p$ sub-matrix of  $\widetilde\Sigma_i+\widetilde\Sigma_j,$ and $\hat \Sigma_{i,j}$ defined in~\eqref{eq:hatsigma}.
	\end{theorem}
	
	Theorem~\ref{mle_est} implies that Assumptions~\ref{asm:est_vec} and~\ref{asm:est_cov} hold with $a_{n,T} = \sqrt{T^{-1}\log n}$, $b_T = T$ and $\Sigma_{i,j}$ corresponding to the scaled asymptotic variance matrix of~~$\hat\vecb_i - \hat\vecb_j$. 
	
	Similarly to the discussion in Section~\ref{sec:logtheory}, the results directly imply that Assumptions~\ref{asm:est_vec} and~\ref{asm:est_cov} continue to hold for any sub-vectors of $\hat\vecg_i$.  
	
	{We now consider the dependent case. Since we need to account for the temporal dependence structure, the asymptotic covariance matrix for the estimator~$\hat\vecg_i$ is now of the form $\widetilde\Sigma_i = B_i^{-1}\widetilde H_iB_i^{-1}$ where the matrix~$B_i$ is defined in~\eqref{eq:AiBi} as in the independent case, whereas the matrix $\widetilde H_i$ is defined in the following way incorporating the dependence
		\begin{equation*}
			\widetilde H_i := \tau (1-\tau) \E[\vecz_{i1}\vecz_{i1}^\top] +\sum_{j=1}^\infty \E[\vecw_{i1}\vecw_{i,1+j}^\top+\vecw_{i,1+j}\vecw_{i1}^\top] \,,
		\end{equation*}
		with $\vecw_{i1}= \vecz_{i1}\big(\tau-\1(\response_{it}\le q_{i,\tau}(\vecz_{i1})) \big).$
		This motivates the following generalized version of the Hendricks-Koenker sandwich covariance matrix estimator~$\hat{\widetilde\Sigma}_{iT}$
		\begin{equation}
			\label{eq:HKvar_dp}
			\hat{\widetilde\Sigma}_{iT} := \widehat{B}_{iT}^{-1} \widehat{H}'_{iT} \widehat{B}_{iT}^{-1}\,,
		\end{equation}
		where $\widehat{B}_{iT} :=  \frac{1}{T} \sum_{t=1}^{T} \widehat{f}_{it} \vecz_{it}\vecz_{it}^\top$ is defined in the same way as in the independent case and the estimator $\widehat{H}'_{iT} $ is defined via
		\begin{align*}
			\label{eq:hatH}
			~\widehat{H}'_{iT} := \tau (1-\tau)\frac{1}{T} \sum_{t=1}^{T}  \vecz_{it}\vecz_{it}^\top+ \sum_{1\le j\le m_T}\Big(1-\frac{j}{T}\Big)\Big(\frac{1}{T}\sum_{t\in T_j}(\widehat \vecw_{it} \widehat \vecw_{i,t+j}^\top+  \widehat \vecw_{i,t+j}\widehat \vecw_{it}^\top )\Big)
		\end{align*}
		Here, $T_j:=\{1\le t\le T-j\},$ $m_T>0$ denotes the bandwidth parameter tending to be infinity as $T$ goes to infinity, and
		\[
		\widehat \vecw_{it} :=\vecz_{it}\Big(\tau-\1\big(\response_{it}\le \hat\vecg_i(\tau)^\top\vecz_{it}\big)\Big)\,.
		\]
		
		To establish the asymptotic consistency of the covariance estimator, we need following additional assumptions.
	}	
	{	
		\begin{assumption}
			\label{B2} For each $i=1,...,n$ and $j>1$, the random vector $(Y_{i1},Y_{i,1+j})$ has a density conditional on $(\vecz_{i1},\vecz_{i, 1+j})$ and this density is bounded uniformly across $i,j$ and $n,T$.
		\end{assumption}
		A similar assumption was made in~\cite{kato2012asymptotics}.
		\begin{assumption}
			\label{B3} Assume that $d_T = o(1)$ and $m_T\to\infty$ as $T \to \infty$, and 
			\begin{equation*}
				\frac{\log n}{ T d_T^2 }=o(1),\qquad \frac{m_T^3\log n}{T}=o(1)\,.
			\end{equation*}
		\end{assumption}
		This assumption is similar to Assumption~\ref{A4} and imposes a restriction on the relative growth of the time dimension compared to the number of individuals.
	}
	
	{	
		\begin{theorem}\label{quantreg_est1}
			Let Assumptions~\ref{A1}-\ref{A3}, and \ref{B1}-\ref{B2} hold. 
			Assume $T$ grow at most polynomial in $n$, $(\log n)^3 = o(T)$, and $\min(n,T) \to \infty$. Assume that the smallest and largest eigenvalues of $\widetilde H_i$ are bounded away from zero and infinity uniformly in $i,n,T$.\\
			\noindent (i) It holds that~
			\begin{equation}\label{eq:mix_est}
				\sup_{i\in\{1,\dots,n\} }\normtwo{\est_i-\vecg^*_i}=\bigo_\mpr\biggl(\sqrt{\frac{\log n}{T}}\biggr).
			\end{equation}
			In particular, Assumption~\ref{asm:est_vec} holds with $a_{n,T}=\sqrt{\frac{\log n}{T}}$. \\
			\noindent (ii) If in addition to the above Assumption~\ref{B3} holds, then Assumption~\ref{asm:est_cov} is also satisfied with $b_T := T$, $\Sigma_{i,j}$ denoting the lower $p\times p$ sub-matrix of  $\widetilde\Sigma_i+\widetilde\Sigma_j$ and $\hat \Sigma_{iT}$ denoting the lower $p \times p$ submatrix of
			\begin{equation}\label{eq:hatsigma}
				\hat \Sigma_{i,j} := T^{-1}\Big(\hat{\widetilde{\Sigma}}_{iT} + \hat{\widetilde{\Sigma}}_{jT}\Big)\,.
			\end{equation}
		\end{theorem}
		
	}

	\subsubsection{Quantile regression with common slope and grouping on the intercepts (Example~\ref{ex:qrint})} \label{sec:qrint}

	Consider the quantile regression panel data model 
	\begin{equation*}
		q_{i,\tau}(\vecx_{it})=\alpha_i^*(\tau)+\vecx_{it}^\top\vecb^*(\tau), ~~t=1,\dots,T, i=1,\dots,n\,,
	\end{equation*}
	where $q_{i,\tau}(\vecx_{it}):=\inf\bigl\{y: \mpr(\response_{it}<y|\vecx_{it})\ge\tau\bigr\}$ denotes the conditional $\tau\,$-\,quantile of $\response_{it}$ given $\vecx_{it}.$ In contrast to the setting in Section~\ref{sec:qrslope}, we assume that the slope coefficient $\vecb^*$ is common across individuals and are only interested in grouping the intercepts. This model was considered in~\cite{gu2019panel}, who used a lasso-type penalty to enforce grouping on the intercepts $\alpha_i^*$. The latter paper also demonstrated that putting this kind of regularization on $\alpha_i^*$ can result in improved estimation of $\vecb^*$ compared to leaving $\alpha_i^*$ unrestricted.

	Assume that $(\vecx_{it},Y_{it})$ are i.i.d. across $t$ for each $i$ and independent across $i$. Since only intercepts contain the grouping information, we aim to use the estimates for $\alpha_i^*,$ and their variance estimates to construct the similarity matrix. At this point, there are two possibilities for estimating $\alpha_i^*$: (1) run individual-specific quantile regressions ignoring the fact that $\vecb^*$ is common across individuals or (2) put all individuals into a single large model in order to borrow information across individuals to improve the efficiency in estimating the joint coefficient vector $\vecb^*$. 
	
	Approach (1) has computational advantages, especially if $n$ is large, but can also result in a loss of efficiency. The theoretical treatment of (1) easily follows from minor modifications of the results in Section~\ref{sec:qrslope}, and we hence focus on the second approach. Define
	\begin{equation} \label{eq:alphabetaqrint}
		\big(\tilde\alpha_1(\tau),\cdots,\tilde\alpha_n(\tau),\tilde\vecb(\tau) \big):=\argmin_{\alpha_1,\dots,\alpha_n,\vecb}\frac{1}{nT}\sum_{i=1}^n\sum_{t=1}^T \rho_{\tau} (\response_{it} - \alpha_i-\vecx_{it}^\top\vecb)\,.
	\end{equation}
	
	In what follows, we assume that $n \to \infty$ which is the more relevant case for group structure detection. In this case, it is possible to obtain simplified estimators for the asymptotic variance of $\tilde \alpha_i$. Those estimators will be motivated next. 
	
	The main insight is that under $n \to \infty$ the estimation of $\vecb^*$ has a negligible effect of the asymptotic variance of $\tilde \alpha_i$ since $\vecb^*$ is estimated at a faster rate due to borrowing information across individuals. Further observe that, defining $\hat e_{it}  = \response_{it}-\vecx_{it}^\top\tilde\vecb$, we have
	\begin{equation}  \label{eq:alphabetaqrint2}
		\tilde\alpha_i = \argmin_{\alpha\in\R}\frac{1}{T}\sum_{t=1}^T \rho_{\tau} (\hat e_{it} - \alpha),~~i=1,\dots,n\,.
	\end{equation}
	Thus $\tilde\alpha_{i}$ is approximately the sample quantile of $\{\hat e_{it}, t=1,\dots,T\}$, which should be close to the sample quantile of $\{e_{it}, t=1,\dots,T\},$ where $e_{it} := \response_{it} - \vecx_{it}^\top\vecb^*$.
	
	If $n \to \infty$ this idea can be formalized by applying a modification of Lemma 7 in \cite{galvao2018} (after noting that the proof of the latter result can be modified to weaken the assumption $n(\log T)^2 /T \to 0$ made in there). Denoting the sample quantile of $\{e_{it}, t=1,\dots,T\}$ by $\hat \alpha_i$, the latter result implies
	$$
	\sup_{i=1,\dots,n} |\hat \alpha_i - \tilde \alpha_i| = \bigo_{\mpr}\Big( \normtwo{\tilde \vecb - \vecb^*} + T^{-1}\log (n \vee T) \Big ).
	$$
	Observing that by the proof of Theorem 3.2 in \cite{kato2012asymptotics} we have $\normtwo{\tilde \vecb - \vecb^*}= o_{\mpr}(T^{-1/2}),$ when $n \to \infty$ (note that this part of their result does not require the restrictive growth assumption on $n$ which is needed for unbiased asymptotic normality of $\tilde\vecb$), this implies $|\hat \alpha_i - \tilde \alpha_i|  = o_{\mpr}(T^{-1/2})$ uniformly over $i$ and thus the asymptotic distributions of $\hat \alpha_i$ and $\tilde \alpha_i$ coincide. Now classical results on the distribution of sample quantiles imply that the asymptotic variance of  $\hat \alpha_i$ is  
	\begin{equation}\label{eq:sigma1}
		\Sigma_i = \tau(1-\tau)/f_{e_i}^2(q_{e_i}(\tau)) \,,
	\end{equation}
	where $f_{e_i}, q_{e_i}$ denote the (unconditional) density and quantile function of $e_{i1}$, respectively. 
	
	This motivates the following version of $\hat \Sigma_{i,j}$: for a bandwidth parameter $d_T$ define
	\[
	\hat{\widetilde\Sigma}_{iT} := \tau(1-\tau) \Bigg( \frac{\tilde\alpha_i(\tau+d_T)-\tilde\alpha_i(\tau-d_T)}{2d_T} \Bigg)^2, \quad i=1,\dots,n
	\]
	and compute
	\begin{equation}\label{eq:hatsigma1}
		\hat \Sigma_{i,j} := T^{-1}\Big(\hat{\widetilde \Sigma}_{iT} + \hat{\widetilde \Sigma}_{jT}\Big)\,.
	\end{equation}
	
	\begin{theorem}\label{qrint_est}
		Let Assumptions~\ref{A1}-\ref{A3} with $\vecz_{it}=\vecx_{it}$ hold.
		Assume $\log n = o(T)$, $\min(n,T) \to \infty$.\\
		(i) It holds that~
		$$
		\sup_{i\in\{1,\dots,n\} } |\tilde\alpha_i-\alpha^*_i| =\bigo_\mpr\biggl(\sqrt{\frac{\log(n\vee T)}{T}}\biggr).
		$$
		In particular, Assumption~\ref{asm:est_vec} holds with $a_{n,T}=\sqrt{\frac{\log(n\vee T)}{T}}$. \\
		(ii)
		If in addition to the above $\log(n\vee T)/(n d_T^2) = o(1)$, then Assumption~\ref{asm:est_cov} is also satisfied with $b_T := T$, $\Sigma_{i,j} = \Sigma_i + \Sigma_j$ where $\Sigma_i,\Sigma_j$ are defined in~\eqref{eq:sigma1}, and $\hat \Sigma_{i,j}$ defined in~\eqref{eq:hatsigma1}.
	\end{theorem}
	
	We next consider the case of temporal dependence. Under Assumptions~\ref{B1}--\ref{B2} the asymptotic variance takes he form 
	\[
	\Sigma_i = \frac{1}{f_{e_i}^2(q_{e_i}(\tau))}\sum_{t \in \Z} \cov(\1\{e_{i0} \leq q_{e_i}(\tau)\}, \1\{e_{it} \leq q_{e_i}(\tau)\}). 
	\]
	This can be estimated consistently by
	\[
	\hat{\Sigma}_{iT} := \Bigg( \frac{\tilde\alpha_i(\tau+d_T)-\tilde\alpha_i(\tau-d_T)}{2d_T} \Bigg)^2 \Big(\tau(1-\tau) + \sum_{1 \le j \le m_T} (1-j/T)\sum_{t \in T_j} (\widehat \vecw_{it} \widehat \vecw_{i,t+j}^\top+  \widehat \vecw_{i,t+j}\widehat \vecw_{it}^\top )  \Big)
	\]
	where $T_j:=\{1\le t\le T-j\},$ $m_T>0$ denotes the bandwidth parameter tending to be infinity as $T$ goes to infinity, and
	\[
	\widetilde \vecw_{it} := \tau-\1\big\{\response_{it}\le \tilde\vecb_i(\tau)^\top\vecx_{it} + \hat\alpha_i(\tau)\big\}\,.
	\]

	\begin{theorem}\label{qrint_est1}
		Let Assumptions~\ref{A1}-\ref{A3}, and \ref{B1}-\ref{B2} with $\vecz_{it}=\vecx_{it}$ hold. 
		Assume $T$ grows at most polynomially in $n$, $(\log n)^3 = o(T)$, and $\min(n,T) \to \infty$.\\
		(i) It holds that~
		\[
		\sup_{i\in\{1,\dots,n\} } |\tilde\alpha_i-\alpha^*_i| =\bigo_\mpr\biggl(\sqrt{\frac{\log n}{T}}\biggr)\,.
		\]
		where $\tilde\alpha_i$ is defied in~\eqref{eq:alphabetaqrint2}.
		In particular, Assumption~\ref{asm:est_vec} holds with $a_{n,T}=\sqrt{\frac{\log n}{T}}$. \\
		(ii)
		If in addition to the above Assumption~\ref{B3} holds, then Assumption~\ref{asm:est_cov} is also satisfied with $b_T := T$, $\Sigma_{i,j} = \Sigma_i + \Sigma_j,$ where $\Sigma_i,\Sigma_j$ are defined in~\eqref{eq:sigma1}, and $\hat \Sigma_{i,j}$ defined in~\eqref{eq:hatsigma1}.
	\end{theorem}

	\section{Numerical experiments} \label{sec:sims}
	
	In Section~\ref{sec:simlog} and Section~\ref{sec:simquant}, we report the performance of different algorithms in terms of assigning individuals to groups when the true number of groups is specified. We consider two performance metrics: perfect matching, which corresponds to the proportion of times that the exact group assignment is found, and average matching. The latter is computed as follows. Define the true cluster assignment as a set $\omega^*:=\{\omega^*_1,\dots,\omega^*_n\},$ where $\omega^*_i\in\{1,\dots,G^*\}$ denotes the $\omega^*_i$-th group to which the individual $i$ belongs. 
	Define the set of permutations of the labels $\Phi:=\{\phi: \phi \text{ is a bijection from } \{1,\dots,G^*\} \text{ to } \{1,\dots,G^*\}\}.$
	Define the estimated membership as a set~$\hat\omega:=\{\hat\omega_1,\dots,\hat\omega_n\},$ where $\hat\omega_i\in\{1,\dots,G^*\}$ denotes the estimated group number of the $i$-th individual.
	We define the average percentage of correct classification of the estimated membership~$\hat\omega$ as
	\[
	\max_{\phi\in\Phi}\frac{1}{n}\sum_{i=1}^n \1\{\phi\big(\omega^*_i\big)= \hat\omega_i\}\,.
	\]
	A similar approach was taken in \cite{SSP, gu2019panel, Leng}. The performance of the heuristic~\eqref{eq:selGstar} for selecting the number of groups is considered in Section~\ref{sec:numgrouplog} for logistic regression and in Section~\ref{sec:simnumgr} for quantile regression.  Additional models and simulation settings are discussed in the  supplement.

	\subsection{Logistic regression} \label{sec:simlog}
	In this section, we consider logistic regression with individual-specific intercepts and groupings on the slopes specified as 
	\[
	Y_{it} = \1\{\alpha_i + \vecx_{it}^\top \vecb_{g_i}\geq \epsilon_{it}\}\,,
	\]
	where $\epsilon_{it}$ follows a logistic distribution, $\alpha_i = 1$ for all $i$ and $g_i \in \{1,2,3\}$ with equal proportions, and $\vecx_{it}^\top:=(x_{1it},x_{2it})$. Moreover,
	\[
	\vecb_1=\binom{-4}{1},\vecb_2=\binom{0}{1},\vecb_3=\binom{4}{1}\,.
	\]
	
	We consider two different data generating processes for the covariates $\vecx_{it}$. For Model 1, 
	\[
	x_{1it} = 0.5 \alpha_i + \eta_i + z_{1it}, \quad \text{and}\quad x_{2it} = 0.5 \alpha_i + \eta_i + z_{2it}\,,
	\]
	where $\eta_i \sim N(0,1)$ and $z_{1it} \sim N(0, 4)$ and $z_{2it} \sim N(0, 0.04).$ 
	
	Here, the data generating process is constructed such that the coefficient of $x_2$ is not informative on the group structure while at the same time it is estimated less precisely. On the contrary, the coefficient of $x_1$ is informative on group structure and also precisely estimated. 
	
	For model 2, we switch the labels of $x_1$ and $x_2$. This is a more challenging DGP because the coordinate of $\vecb$ that contains group information is estimated with a lot of noise; see the scatter plot of $\{\hat{\vecb}_i\}_{i=1,\dots,n}$ in Figure \ref{fig:betalogit} for a data realization from Model 1 versus Model 2.
	
		\begin{figure}[H]
		\centering
		\subfigure[]{	\includegraphics[width=1\textwidth]{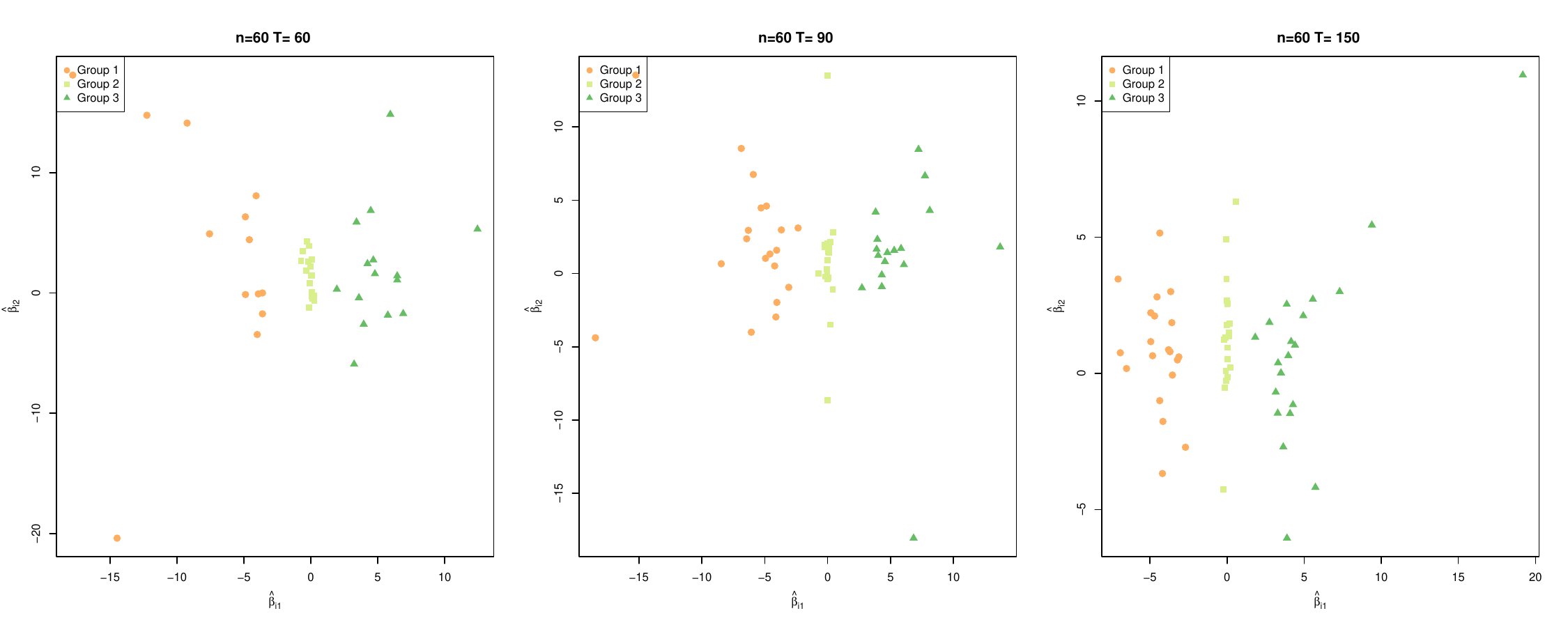}}
		\subfigure[]{	\includegraphics[width=1\textwidth]{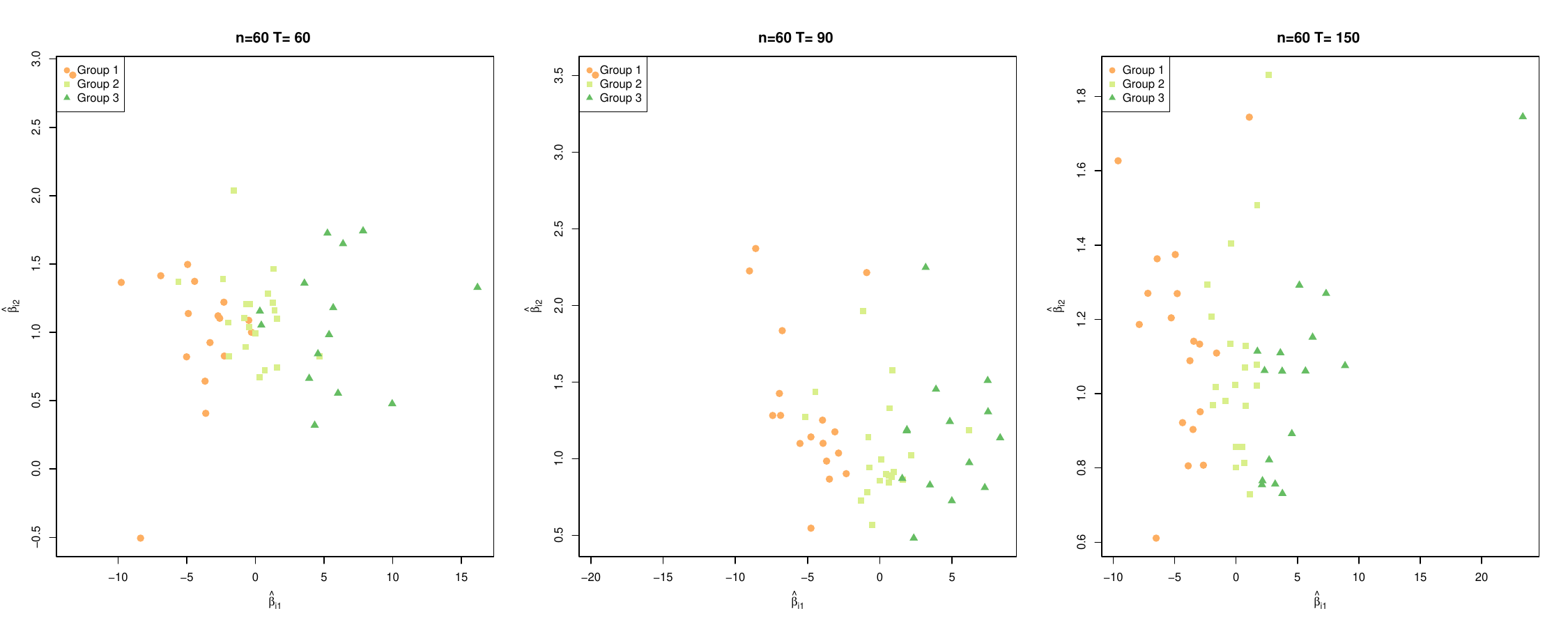}}
		\caption{\small{Scatter plots of $\{\hat{\vecb}_i\}_{i=1}^n$ for Model 1 (Figure (a)) and Model 2 (Figure (b)).}}
		\label{fig:betalogit}
	\end{figure} 
	
	\subsubsection{Clustering with a known number of groups}
	We first compare our method with the C-LASSO proposed by \cite{SSP}. The C-LASSO approach proposed in \cite{SSP} considers minimizing the following objective: 
	\[
	\frac{1}{nT} \sum_i \sum_t \psi(Y_{it}, \vecx_{it}, \vecb_i, \hat \alpha_i(\vecb_i)) + \frac{\lambda}{n} \sum_i \prod_{k=1}^{K_0} \norm{\vecb_i - \eta_k}\,.
	\]
	This itself is not a convex optimization, but at each $k$, we can focus on only the $k$-th element in the product term in the penalty, resulting in a convex program. For details of the implementation, we refer to our supplement or \cite{SSP}. {In addition, we also consider an interative $k$-means approach in the spirit of \cite{BM}. In particular, in each iterative step, we re-estimate group membership based on the logit likelihood function, and then refit the model until coefficients converge. We also compare to the Sequential Binary Segmentation Algorithm (SBSA) in \cite{wang2021} (labeled as SBSA in Table \ref{tab:logistic}. The SBSA method applies the binary segmentation algorithm to detect break-points in eigenvectors from the spectral decomposition of the outer product of $\hat \vecb$ that corresponds to the $\min(p, G)$ largest eigenvalues where $p$ is the number of covariates. }

	\begingroup
	\tabcolsep = 4pt
	\def\arraystretch{1}
	
	\begin{table}[H]
		\small
		\begin{center}
			\resizebox{\textwidth}{!}{\begin{tabular}{rrccccrrrrrrr|ccccrrrrrrr}
					\toprule
					\multicolumn{24}{c}{\bfseries Model 1}\tabularnewline
					\cline{1-24}
					\multicolumn{2}{c}{\bfseries  }&\multicolumn{1}{c}{\bfseries }&\multicolumn{10}{c|}{\bfseries Perfect Match}&\multicolumn{1}{c}{\bfseries }&\multicolumn{10}{c}{\bfseries Average Match}\tabularnewline
					\cline{1-2} \cline{3-13} \cline{15-24}
					\multicolumn{1}{c}{n}&\multicolumn{1}{c}{T}&\multicolumn{1}{c}{}&\multicolumn{1}{c}{S}&\multicolumn{1}{c}{PAM}&\multicolumn{1}{c}{S-Diag}&\multicolumn{1}{c}{S-Iden}&\multicolumn{4}{c}{C-LASSO}&\multicolumn{1}{c}{$k$-means}&\multicolumn{1}{c}{SBSA}&\multicolumn{1}{c}{}&\multicolumn{1}{c}{S}&\multicolumn{1}{c}{PAM}&\multicolumn{1}{c}{S-Diag}&\multicolumn{1}{c}{S-Iden}&\multicolumn{4}{c}{C-LASSO}&\multicolumn{1}{c}{$k$-means}&\multicolumn{1}{c}{SBSA}\tabularnewline
					\midrule
			$30$&$ 60$&&$0.83$&$0.53$&$0.88$&$0.00$&$0.32$&$0.67$&$0.78$&$0.40$&$0.03$&$0.07$&&$0.992$&$0.919$&$0.995$&$0.478$&$0.948$&$0.979$&$0.988$&$0.959$&$0.857$ & $0.631$\tabularnewline
			$30$&$ 90$&&$0.93$&$0.70$&$0.99$&$0.00$&$0.44$&$0.80$&$0.86$&$0.55$&$0.05$&$0.16$&&$0.998$&$0.963$&$0.999$&$0.448$&$0.968$&$0.990$&$0.993$&$0.978$&$0.881$&$0.741$\tabularnewline
			$30$&$150$&&$0.99$&$0.98$&$0.99$&$0.09$&$0.77$&$0.93$&$0.98$&$0.93$&$0.04$&$0.70$&&$1.000$&$0.999$&$1.000$&$0.590$&$0.987$&$0.997$&$0.999$&$0.998$&$0.884$&$0.919$\tabularnewline
		$60$&$ 60$&&$0.77$&$0.40$&$0.82$&$0.01$&$0.14$&$0.47$&$0.55$&$0.21$&$0.00$&$0.00$&&$0.995$&$0.914$&$0.996$&$0.433$&$0.948$&$0.976$&$0.981$&$0.967$&$0.858$&$0.632$\tabularnewline
		$60$&$ 90$&&$0.80$&$0.57$&$0.88$&$0.01$&$0.21$&$0.57$&$0.70$&$0.54$&$0.00$&$0.16$&&$0.996$&$0.964$&$0.998$&$0.407$&$0.959$&$0.985$&$0.990$&$0.987$&$0.873$&$0.733$\tabularnewline
		$60$&$150$&&$0.92$&$0.88$&$0.96$&$0.12$&$0.62$&$0.88$&$0.91$&$0.85$&$0.00$&$0.66$&&$0.999$&$0.990$&$0.999$&$0.577$&$0.989$&$0.998$&$0.998$&$0.996$&$0.876$&$0.914$\tabularnewline
																\hline
					\cline{1-24}
					\multicolumn{24}{c}{\bfseries }\tabularnewline
					\multicolumn{24}{c}{\bfseries Model 2}\tabularnewline
					\cline{1-24}
					\multicolumn{2}{c}{\bfseries  }&\multicolumn{1}{c}{\bfseries }&\multicolumn{10}{c|}{\bfseries Perfect Match}&\multicolumn{1}{c}{\bfseries }&\multicolumn{10}{c}{\bfseries Average Match}\tabularnewline
					\cline{1-2} \cline{3-13} \cline{15-24}
					\multicolumn{1}{c}{n}&\multicolumn{1}{c}{T}&\multicolumn{1}{c}{}&\multicolumn{1}{c}{S}&\multicolumn{1}{c}{PAM}&\multicolumn{1}{c}{S-Diag}&\multicolumn{1}{c}{S-Iden}&\multicolumn{4}{c}{C-LASSO}&\multicolumn{1}{c}{$k$-means}&\multicolumn{1}{c}{SBSA}&\multicolumn{1}{c}{}&\multicolumn{1}{c}{S}&\multicolumn{1}{c}{PAM}&\multicolumn{1}{c}{S-Diag}&\multicolumn{1}{c}{S-Iden}&\multicolumn{4}{c}{C-LASSO}&\multicolumn{1}{c}{$k$-means}&\multicolumn{1}{c}{SBSA}\tabularnewline
					\midrule
			$30$&$ 60$&&$0$&$0$&$0$&$0$&$0$&$0$&$0$&$0$&$0$&$0$&&$0.691$&$0.643$&$0.682$&$0.541$&$0.583$&$0.726$&$0.675$&$0.583$&$0.700$&$0.486$\tabularnewline
			$30$&$ 90$&&$0$&$0$&$0$&$0$&$0$&$0$&$0$&$0$&$0$&$0$&&$0.753$&$0.675$&$0.739$&$0.539$&$0.632$&$0.787$&$0.748$&$0.632$&$0.722$&$0.517$\tabularnewline
			$30$&$150$&&$0.01$&$0.01$&$0.02$&$0$&$0$&$0.02$&$0$&$0$&$0$&$0$&&$0.869$&$0.759$&$0.862$&$0.582$&$0.729$&$0.864$&$0.837$&$0.729$&$0.743$&$0.563$\tabularnewline
		$60$&$ 60$&&$0$&$0$&$0$&$0$&$0$&$0$&$0$&$0$&$0$&$0$&&$0.699$&$0.625$&$0.677$&$0.475$&$0.577$&$0.728$&$0.664$&$0.577$&$0.684$&$0.469$\tabularnewline
		$60$&$ 90$&&$0$&$0$&$0$&$0$&$0$&$0$&$0$&$0$&$0$&$0$&&$0.791$&$0.665$&$0.769$&$0.490$&$0.648$&$0.802$&$0.746$&$0.648$&$0.715$&$0.502$\tabularnewline
		$60$&$150$&&$0$&$0$&$0$&$0$&$0$&$0$&$0$&$0$&$0$&$0$&&$0.883$&$0.729$&$0.864$&$0.501$&$0.746$&$0.872$&$0.840$&$0.746$&$0.714$&$0.550$\tabularnewline
														\bottomrule
			\end{tabular}}
			\caption{\small{Comparison of group membership estimation. $S$ refers to our proposed spectral clustering method, PAM refers to the PAM method applied on the dissimilarity measure $V$. S-Diag refers to the spectral clustering approach but we plug in the diagnonal of the variance-covariance matrix estimate and likewise, S-Iden is similar but the variance-covariance matrix of individual estimate is taken to be the identity matrix. For C-LASSO, the four columns of the results are based on tuning parameter constants $c = 0.05 \times \{1, \frac{1}{4}, \frac{1}{8}, \frac{1}{32}\}$. The k-means approach is adapted from \cite{BM} where we iteratively cluster individuals with a refit to update group coefficients until convergence. For the iterative k-means method, we use 20 random starting groupings and a maximum of 100 iterations for each starting grouping. Then we take the grouping that minimizes the loss function. %
					\label{tab:logistic}}}
		\end{center}
	\end{table}
	\endgroup
	
	\vspace{-.25in}
	
	The first few rows in Table \ref{tab:logistic} report the performance of four different grouping methods for several combinations of $n$ and $T$ based on Model 1. We evaluate the performance by the proportion of perfect matches out of 100 simulation repetitions and the average matches described at the beginning of Section~\ref{sec:sims}. {The spectral clustering method works consistently better than the PAM approach (labeled PAM).  
		
From local analysis in Section \ref{sec: local}, one might expect that PAM and the iterative k-mean method should perform similarly. However, that analysis is asymptotic and inspecting the scatter plot of $\{\hat{\vecb}_i\}_{i=1,\dots,n}$ in Figure \ref{fig:betalogit} suggests that the coefficient estimates are not yet approximately Gaussian around their true values. Hence the asymptotic analysis may not provide a sufficiently accurate description of finite sample performance at the sample sizes considered in this simulation. We also note that there are non-convergence issues with the iterative k-means method. For $n = 30, L = 60$, in $12\%$ of the cases none of the random initialization lead to convergence after $100$ iterations and in $35\%$ of all cases the initialization that led to the best likelihood function did not correspond to convergence after 100 iterations.
			
In addition, the non-Gaussian shape of the point clouds may also provide an explanation for the superior performance of the spectral clustering method over PAM, since spectral clustering is known to have an advantage for clusters with non-elliptical shapes. 
		
For small $T$, using the diagonalized estimated variance-covariance matrix (labeled S-Diag) actually performs slightly better than using the full estimated variance-covariance matrix (labeled S). In finite samples, the off-diagonal terms of the variance-covariance matrix can be poorly estimated and using just the diagonal variance information seems to provide a small margin of better performance. However, discarding variance information completely (labeled S-Iden) clearly shows much worse performance. The iterative $k$-mean method performs much worse than spectral clustering with variance information in terms of perfect match. Its average match performance is in fact better than the spectral method when the variance information is not accounted for. This shows that spectral clustering needs to be applied together with variance information for good performance. We also note that the iteration k-means method sometimes does not converge after 100 iterations, which may explain why performance is not improving monotonically as $T$ increases. The SBSA approach in \cite{wang2021} has better performance than the k-mean method for perfect match proportion in Model 1, but is still inferior to CLASSO, PAM, and spectral clustering with variance information. }

	For C-LASSO, the penalty tuning parameter $\lambda$ is set at $c  T^{-1/3}\Var(Y_{it})$ as recommended by the authors with a few different values of $c$ specified in the caption of Table \ref{tab:logistic}. We see that the C-LASSO can perform very well for a suitably chosen constant, and our method matches that or overperforms sometimes. However, it can perform poorly if the tuning parameter constant is not chosen carefully. This imposes challenges for its practical usage. 
	
	Performance for Model 2 is reported in the last few rows in Table \ref{tab:logistic}. We clearly see that this is a much more challenging DGP with almost all methods failing to recover perfect match for group membership. In terms of average matches, our method still performs comparable or sometimes better than all other methods for all combinations of $n$ and $T$. 
	
	\subsubsection{Estimating the number of groups}\label{sec:numgrouplog}		
	
	Simulation results in Table \ref{tab:logistic} assume the researchers know the correct number of groups $G$. We report in Table \ref{tab:logitestG} the performance of the proposed method for estimation of $G$. The comparison is made with the information criteria proposed in \cite{SSP}. For Model 1, both methods work very well while for Model 2, the information criteria of \cite{SSP} works much better across most combinations of $n,T$. The information criteria relies on the whole sample to estimate $G$ while our heuristic approach only requires information on individual based estimates.
	
	\begin{table}[!h] \label{tab:logitestG}
		\begin{center}
			\begin{tabular}{rrcrrrrrcrrrrr}
				\toprule 
				\multicolumn{14}{c}{\bfseries Model 1}\tabularnewline
				\cline{1-14}
				\multicolumn{2}{c}{\bfseries  }&\multicolumn{1}{c}{\bfseries }&\multicolumn{5}{c}{\bfseries Heuristic}&\multicolumn{1}{c}{\bfseries }&\multicolumn{5}{c}{\bfseries IC-CLASSO}\tabularnewline
				\cline{1-2} \cline{4-8} \cline{10-14}
				\multicolumn{1}{c}{n}&\multicolumn{1}{c}{T}&\multicolumn{1}{c}{}&\multicolumn{1}{c}{1}&\multicolumn{1}{c}{2}&\multicolumn{1}{c}{3}&\multicolumn{1}{c}{4}&\multicolumn{1}{c}{$\geq$5}&\multicolumn{1}{c}{}&\multicolumn{1}{c}{1}&\multicolumn{1}{c}{2}&\multicolumn{1}{c}{3}&\multicolumn{1}{c}{4}&\multicolumn{1}{c}{$\geq$5}\tabularnewline
				\hline
				$30$&$ 60$&&$0.10$&$0$&\textbf{0.89}&$0$&$0.01$&&$0$&$0.02$&\textbf{0.98}&$0$&$0$\tabularnewline
				$30$&$ 90$&&$0$&$0$&\textbf{0.98}&$0.01$&$0.01$&&$0$&$0$&\textbf{0.99}&$0.01$&$0$\tabularnewline
				$30$&$150$&&$0$&$0$&\textbf{1.00}&$0$&$0$&&$0$&$0$&\textbf{0.94}&$0.06$&$0$\tabularnewline
				$60$&$ 60$&&$0.02$&$0$&\textbf{0.98}&$0$&$0$&&$0$&$0$&\textbf{1.00}&$0$&$0$\tabularnewline
				$60$&$ 90$&&$0$&$0$&\textbf{1.00}&$0$&$0$&&$0$&$0$&\textbf{1.00}&$0$&$0$\tabularnewline
				$60$&$150$&&$0$&$0$&\textbf{1.00}&$0$&$0$&&$0$&$0$&\textbf{0.94}&$0.06$&$0$\tabularnewline
				\hline 
				\cline{1-14}
				\multicolumn{14}{c}{\bfseries }\tabularnewline
				\multicolumn{14}{c}{\bfseries Model 2}\tabularnewline
				\cline{1-14}
				\multicolumn{2}{c}{\bfseries  }&\multicolumn{1}{c}{\bfseries }&\multicolumn{5}{c}{\bfseries Heuristic}&\multicolumn{1}{c}{\bfseries }&\multicolumn{5}{c}{\bfseries IC-CLASSO}\tabularnewline
				\cline{1-2} \cline{4-8} \cline{10-14}
				\multicolumn{1}{c}{n}&\multicolumn{1}{c}{T}&\multicolumn{1}{c}{}&\multicolumn{1}{c}{1}&\multicolumn{1}{c}{2}&\multicolumn{1}{c}{3}&\multicolumn{1}{c}{4}&\multicolumn{1}{c}{$\geq$5}&\multicolumn{1}{c}{}&\multicolumn{1}{c}{1}&\multicolumn{1}{c}{2}&\multicolumn{1}{c}{3}&\multicolumn{1}{c}{4}&\multicolumn{1}{c}{$\geq$5}\tabularnewline
				\hline
				$30$&$ 60$&&$0.93$&$0.07$&\textbf{0.00}&$0$&$0$&&$0$&$0.85$&\textbf{0.15}&$0$&$0$\tabularnewline
				$30$&$ 90$&&$0.76$&$0.20$&\textbf{0.04}&$0$&$0$&&$0$&$0.79$&\textbf{0.20}&$0.01$&$0$\tabularnewline
				$30$&$150$&&$0.61$&$0.16$&\textbf{0.22}&$0.01$&$0$&&$0$&$0.52$&\textbf{0.40}&$0.07$&$0.01$\tabularnewline
				$60$&$ 60$&&$0.92$&$0.08$&\textbf{0.00}&$0$&$0$&&$0$&$0.99$&\textbf{0.01}&$0$&$0$\tabularnewline
				$60$&$ 90$&&$0.87$&$0.11$&\textbf{0.02}&$0$&$0$&&$0$&$0.92$&\textbf{0.08}&$0$&$0$\tabularnewline
				$60$&$150$&&$0.47$&$0.16$&\textbf{0.36}&$0.01$&$0$&&$0$&$0.62$&\textbf{0.36}&$0.02$&$0$\tabularnewline
				\bottomrule
		\end{tabular}\end{center}
		\caption{Estimation of G for Model 1 and Model 2 with true $G = 3$. IC-CLASSO is based on a combination of the log-likelihood evaluation and a penalty term that depends on a turning parameter, group size and $n$ and $T$.}
	\end{table}
	
	\subsubsection{Computation times}
	In what follows, Table \ref{tab:logittime} reports the computation times for our method versus C-LASSO,  the iterative k-means method which requires iteration with the whole sample as well as SBSA proposed in \cite{wang2021}. The run time of SBSA is very similar to our proposed method because it also only required the use of individual coefficients. For C-LASSO the reported times are based on maximum 20 iterations for optimization. The final estimates are obtained when the objective function differs less than 0.001 and when the $\ell_2$ norm of the estimates group centers differ by less than $0.1\%$ or when the maximum iterations are reached. For the iterative k-mean method, we take 20 random start of group membership and pick the best estimates that minimizes the loss function criteria. For each random starting, the maximum number of iteration is 100. With known $G$, our method and the SBSA method has the least computational time while the iterative k-means has the largest. This is because quite often the k-mean algorithm does not converge before the maximum 100 iterations is reached. Our algorithm spends most of its computation time on individual based estimates. For the C-LASSO method, the individual estimates are computed as initial estimates before applying a re-optimization with the penalty terms for group center estimates. Because the optimization problem is only convex for optimizing over one group center while fixing the others, it has to optimize group by group, which increases computation times. The iterative k-means method takes the most computation time, as it requires individual loops to decide group membership until convergence as well as refitting to obtain group center estimates. In practice, we observe that it may take a very large number of iterations to converge. When $G$ is not known, the computation time of our method does not increase because the heuristic method recycles already computed similarity measures to estimate $G$. The SBSA method uses a IC criteria to estimate $G$. Because grouping is obtained very fast for each candidate model and the IC criteria just needs to evaluate the likelihood of each estimated candidate model, the increase in computation time is also very minimal. Both the C-LASSO and k-means rely on information criteria to estimate $G$ which requires fitting of all candidate models with varying $G$. Hence computation times grow at least linearly with the number of candidate models. The reported times in Table \ref{tab:logittime} are based on candidate models with $G = \{1,2,3,4,5\}$.

	\begingroup
	\tabcolsep = 5pt
	\def\arraystretch{1}
	
	\begin{table}[ht]
		\small
		\centering
		\begin{tabular}{rrrrrrrcrrrr}
			\toprule
			\multicolumn{2}{c}{} & \multicolumn{4}{r}{Known G} &\multicolumn{1}{c}{}& \multicolumn{4}{r}{Estimate G}\\ 
			\cline{1-2} \cline{3-7} \cline{9-12}
			& n & T & Spectral & C-LASSO & Kmeans &SBSA &  & Spectral & C-LASSO & Kmeans & SBSA\\ 
			\hline
			& 30 & 60 & 0.32 & 5.88 & 13.23 &0.38   && 0.32 & 31.75 & 86.27 & 0.44\\ 
			& 30& 90 & 0.38 & 6.55 & 24.76 & 0.44&& 0.38 & 39.94 & 122.75 & 0.50\\ 
			& 30 & 150 & 0.42 & 6.87 & 52.69 & 0.43&& 0.43 & 48.09 & 245.97 & 0.52\\ 
			& 60 & 60 & 1.26 & 10.15 & 48.27 & 0.77&& 1.28 & 59.48 & 445.17 & 0.94\\ 
			& 60 & 90 & 1.47 & 11.10 & 197.70 & 0.80&& 1.52& 79.38 & 742.79&1.12\\ 
			& 60 & 150 & 1.61 & 12.30 & 225.92 & 0.85&& 1.61 & 123.48 & 801.23 & 1.34\\ 
			\bottomrule
		\end{tabular}
		\caption{\small{Comparison of computation time in seconds for Model 1: the left panel includes computation times when we assume $G$ is known. The right panel includes computation times when we have to estimate $G$. For our proposed method, we use the heuristic method to estimate $G$ and for all other methods, we use some form of information criteria to estimate $G$ from the set $\{1,2,3,4,5\}$. Timings are averages of 5 data realizations. }}
		\label{tab:logittime} 
	\end{table}
	\endgroup

	\subsection{Quantile regression} \label{sec:simquant}
	In this section, we consider quantile regression with individual-specific intercepts and grouping on the slopes as in Example~\ref{ex:qrslope}, and with joint slope and grouping of intercepts from Example~\ref{ex:qrint}. We focus on the clustering performance with a given (correctly specified) number of groups, the performance of the proposed heuristic, and several other methods for selecting the number of groups is considered in Section~\ref{sec:simnumgr}.

	\subsubsection{Quantile regression individual--specific intercepts and grouping on slopes} \label{sec:simqrslope}
	
	Recall the model specification in Example~\ref{ex:qrslope}:
	$
	q_{it}(\tau) = \alpha_i(\tau) + \vecx_{it}^\top\vecb_i(\tau) = \vecz_{it}^\top\vecg_i(\tau)\,.
	$
	This setting was also considered in~\cite{zhang2019} and we will compare the performance of the proposed method with theirs. 
	Simulations are done in the \textbf{quantreg} package in R. Covariance estimates are computed using the function \texttt{summary.rq()} with option \texttt{se="nid"} and default bandwidth choice \texttt{hs=true}.

	We consider three models. Model 1 corresponds to Model 3 from~\cite{zhang2019}.%
	
	\textbf{Model 1:} 
	$$
	y_{it}=\alpha_i +\vecx_{it}^\top\vecb_{g_i}+0.5x_{2it}e_{it}\,,
	$$
	where
	$$
	\alpha_i\overset{iid}{\sim} U(0,1), i=1,\dots,n,~~g_i\text{ are sampled randomly with equal probablities from }\{1,2,3\}.
	$$
	Set 
	$$
	\vecb_1=\begin{pmatrix} 0.1\\  0.1\end{pmatrix}, 
	\vecb_2=\begin{pmatrix}0.2 \\  0.2\end{pmatrix}, 
	\vecb_3=\begin{pmatrix}0.3\\  0.3\end{pmatrix}, 
	$$
	$e_{it}\overset{iid}{\sim}N(0,1) \text{ or } e_{it}\overset{iid}{\sim} t(3),$
	and
	$$
	\vecx_{it}^\top=(x_{1it},x_{2it})\,,~\text{with}~x_{1it}=0.3\alpha_i +z_{1it}, \text{ with } z_{1it}\overset{iid}{\sim}N(0,1),\text{and}~x_{2it}\sim U(0,1).
	$$

	{Results for $\tau = 0.5$ are reported in Table~\ref{tab:model1a}. We considered the PAM method as well as several variants of spectral clustering. In particular, $S^g$ refers to spectral clustering when we apply the Gaussian kernel instead of the exponential kernel in Algorithm 1. S-Diag is spectral clustering when we use only the diagonal entries in the variance-covariance matrix and sets the off--diagonal entries to zero. S-Iden is spectral clustering when we do not use the variance covariance information of the coefficient estimates. k-$\op{mean}^\circ$ applies the k-mean clustering algorithm on $\hat \vecb$ and ZWZ19 is the method in \cite{zhang2019} which adapts the iterative method of \cite{BM} to the quantile regression case. }
	
	Spectral clustering shows uniformly best performance in terms of average and perfect matching across all settings considered. The approach of~\cite{zhang2019} comes close in terms of average matching and is better than both methods which ignore variance information (S-Iden and k-$\op{means}^\circ$) but is slightly worse than PAM and the spectral method. This agrees with the theoretical analysis in Section \ref{sec: local} which suggests that for a heteroscedastic model as in Model 1, the loss function based approach implicitly takes into account variance information, but is not as efficient as using the dissimilarity measure as in Algorithm 1. Surprisingly,~\cite{zhang2019} shows much worse performance in terms of perfect matching. A closer look at the results revealed that in this model the method of~\cite{zhang2019} often assigns one individual to the wrong group, resulting in good average matching but inferior perfect matching performance. Despite our best efforts at varying various parameters of~\cite{zhang2019} (e.g. criteria for termination and number of random starting points), we were not able to alleviate this issue. Among spectral methods, using the Gaussian kernel to transform the dissimilarity measure does not lead to improvements in terms of performance. Using just the diagonal of the variance-covariance matrix also yields almost identical performance than using the estimated full variance-covariance matrix. We do note that the PAM method is slightly worse than the corresponding spectral clustering method. This seems to be a persistent phenomenon we observe in all the simulations for quantile regression. 
	{Moreover, we provide the scatter plot of $\{\hat{\vecb}_i\}_{i=1,\dots,n}$ in the online Appendix (Section~\ref{app:plots} Figure \ref{fig:beta}) for a data realization where the proposed method achieves perfect matching but all the other methods fail. The figure suggests that there seems to be clear separation along the first coordinate, but not in the second coordinate of $\hat \beta$, which is driven by the fact that the second coordinate is estimated with more noise. However, this information is not available to the researcher. Accounting for the variance information improves the performance compared to methods that do not account for this. The loss function based method of \cite{zhang2019} implicitly accounts for this to some extent, but less well than reweighting. }
	
The simulation findings suggest that the main improvements in our proposal are due to using variance information, while using spectral clustering instead of PAM only leads to modest additional gains. The results here are also consistent with the local analysis in Theorem~\ref{thm:loc} since this is a model with heteroscedastic errors.

	The second model we consider has four groups, with pairs of group centres being close together. Both entries of the coefficient vector carry information about the group structure, but one of them is estimated more precisely than the other one.
	
	\textbf{Model 2:} 
	$$
	y_{it}=\alpha_i +\vecx_{it}^\top\vecb_{g_i}+0.5x_{2it}e_{it}\,,
	$$
	where 
	$$
	\alpha_i=1, i=1,\dots,n,~~
	g_i\text{ are sampled randomly with equal probabilities from}\{1,2,3,4\}.
	$$
	Set
	$$
	\vecb_1=\begin{pmatrix} 0.1\\  0.1\end{pmatrix}, 
	\vecb_2=\begin{pmatrix}0.2 \\  0.2\end{pmatrix}, 
	\vecb_3=\begin{pmatrix}3\\  3\end{pmatrix}, 
	\vecb_4=\begin{pmatrix}3.1\\  3.1\end{pmatrix}, 
	$$
	$e_{it}\overset{iid}{\sim} N(0,1) \text{ or } t(3),$
	and set
	$$
	\vecx_{it}^\top=(x_{1it},x_{2it})\,,~\text{with}~x_{1it}=0.3\alpha_i +z_{1it},~\text{where}~z_{1it}\overset{iid}{\sim}N(0,1),~\text{and}~x_{2it}\sim U(0,1)\,.
	$$
	
	Results for $\tau = 0.5$ are reported in Table~\ref{tab:model2a}.
	The results are fairly similar to those of Model 1, the proposed method has the best performance with respect to perfect and average match. The design of this DGP is also used later for estimation of $G$ to demonstrate the drawback of stability based method proposed in \cite{wang2010}. Again, the main performance boost comes from using reweighting and using spectral clustering instead of PAM only leads to small additional accuracy gains.

	\begingroup
	\tabcolsep = 3pt
	\def\arraystretch{1}
	
	\begin{table}[h]
		\centering
		\scalebox{0.8}{
			\begin{tabular}{cccccccccc|cccccccc}
				\toprule
				\multirow{2}{*}{n}& \multirow{2}{*}{T} & \multicolumn{1}{c}{}& \multicolumn{7}{c|}{\textbf{Perfect Match}} &
				\multicolumn{7}{c}{\textbf{Average Match}} \\
				&  & &  {$\op{S}^g$} & {S} & {PAM} &  {S-Diag} & {S-Iden}&  $k$-$\op{means}^\circ$ & {ZWZ19} & &{$\op{S}^g$} & {S} & {PAM} &  {S-Diag} &{S-Iden}&  $k$-$\op{means}^\circ$ & {ZWZ19} \\
				\midrule
				\multicolumn{18}{c}{$N(0,1), \tau = 0.5$}\\
				\midrule
				$30$&$ 60$&&$0.23$&$0.23$&$0.17$&$0.22$&$0$&$0$&$0.09$&&$0.95$&$0.95$&$0.93$&$0.95$&$0.64$&$0.63$&$0.90$\tabularnewline
				$30$&$ 90$&&$0.62$&$0.62$&$0.56$&$0.61$&$0$&$0$&$0.39$&&$0.98$&$0.98$&$0.98$&$0.98$&$0.70$&$0.70$&$0.97$\tabularnewline
				$30$&$120$&&$0.81$&$0.81$&$0.77$&$0.80$&$0$&$0$&$0.67$&&$0.99$&$0.99$&$0.99$&$0.99$&$0.75$&$0.75$&$0.98$\tabularnewline
				$60$&$ 60$&&$0.06$&$0.06$&$0.04$&$0.05$&$0$&$0$&$0.01$&&$0.95$&$0.95$&$0.94$&$0.95$&$0.63$&$0.61$&$0.92$\tabularnewline
				$60$&$ 90$&&$0.40$&$0.40$&$0.34$&$0.39$&$0$&$0$&$0.16$&&$0.98$&$0.98$&$0.98$&$0.98$&$0.70$&$0.69$&$0.97$\tabularnewline
				$60$&$120$&&$0.72$&$0.72$&$0.65$&$0.71$&$0$&$0$&$0.45$&&$0.99$&$1.00$&$0.99$&$0.99$&$0.76$&$0.76$&$0.99$\tabularnewline
				$90$&$ 60$&&$0.02$&$0.02$&$0.01$&$0.02$&$0$&$0$&$0.00$&&$0.96$&$0.96$&$0.94$&$0.95$&$0.63$&$0.61$&$0.92$\tabularnewline
				$90$&$ 90$&&$0.27$&$0.26$&$0.20$&$0.26$&$0$&$0$&$0.07$&&$0.99$&$0.99$&$0.98$&$0.98$&$0.70$&$0.69$&$0.97$\tabularnewline
				$90$&$120$&&$0.62$&$0.62$&$0.57$&$0.62$&$0$&$0$&$0.34$&&$1.00$&$1.00$&$0.99$&$1.00$&$0.77$&$0.76$&$0.99$\tabularnewline
				\midrule
				\multicolumn{18}{c}{$t(3), \tau = 0.5$}\\
				\midrule
				$30$&$ 60$&&$0.12$&$0.11$&$0.07$&$0.10$&$0$&$0$&$0.04$&&$0.92$&$0.92$&$0.89$&$0.92$&$0.61$&$0.60$&$0.86$\tabularnewline
				$30$&$ 90$&&$0.42$&$0.42$&$0.38$&$0.41$&$0$&$0$&$0.23$&&$0.97$&$0.97$&$0.96$&$0.97$&$0.67$&$0.67$&$0.94$\tabularnewline
				$30$&$120$&&$0.73$&$0.73$&$0.69$&$0.73$&$0$&$0$&$0.51$&&$0.99$&$0.99$&$0.99$&$0.99$&$0.72$&$0.72$&$0.98$\tabularnewline
				$60$&$ 60$&&$0.01$&$0.01$&$0.01$&$0.01$&$0$&$0$&$0.00$&&$0.93$&$0.93$&$0.91$&$0.93$&$0.60$&$0.58$&$0.88$\tabularnewline
				$60$&$ 90$&&$0.22$&$0.23$&$0.16$&$0.21$&$0$&$0$&$0.06$&&$0.97$&$0.97$&$0.97$&$0.97$&$0.67$&$0.66$&$0.95$\tabularnewline
				$60$&$120$&&$0.53$&$0.52$&$0.45$&$0.52$&$0$&$0$&$0.26$&&$0.99$&$0.99$&$0.99$&$0.99$&$0.72$&$0.72$&$0.98$\tabularnewline
				$90$&$ 60$&&$0.00$&$0.00$&$0.00$&$0.00$&$0$&$0$&$0.00$&&$0.94$&$0.93$&$0.91$&$0.93$&$0.60$&$0.58$&$0.89$\tabularnewline
				$90$&$ 90$&&$0.10$&$0.09$&$0.06$&$0.09$&$0$&$0$&$0.02$&&$0.97$&$0.97$&$0.97$&$0.97$&$0.66$&$0.66$&$0.95$\tabularnewline
				$90$&$120$&&$0.39$&$0.39$&$0.34$&$0.38$&$0$&$0$&$0.13$&&$0.99$&$0.99$&$0.99$&$0.99$&$0.73$&$0.72$&$0.98$\tabularnewline
				\bottomrule
		\end{tabular}}
		\caption{\small {Membership estimation based on Spectral (the proposed method), ZWZ19 \cite{zhang2019}, and the vanilla $k-$means method without variance information for Model 1 with $\tau=0.5$ and two error distributions. \label{tab:model1a}}}%
\end{table}
\endgroup

\vspace{.1in}

\begingroup
\tabcolsep = 3pt
\def\arraystretch{1}

\begin{table}[h]
	\centering
	\scalebox{0.8}{
		\begin{tabular}{cccccccccc|cccccccc}
			\toprule
			\multirow{2}{*}{n}& \multirow{2}{*}{T} & \multicolumn{1}{c}{}& \multicolumn{7}{c|}{\textbf{Perfect Match}} &
			\multicolumn{7}{c}{\textbf{Average Match}} \\
			&  & & {$\op{S}^g$} & {S} & {PAM} &  {S-Diag} & {S-Iden}&  $k$-$\op{means}^\circ$ & {ZWZ19} & &  {$\op{S}^g$} & {S} & {PAM} &  {S-Diag} &{S-Iden}&  $k$-$\op{means}^\circ$ & {ZWZ19} \\
			\midrule
			\multicolumn{18}{c}{$N(0,1), \tau = 0.5$}\\
			\midrule
			$30$&$ 60$&&$0.31$&$0.32$&$0.26$&$0.31$&$0$&$0$&$0.06$&&$0.95$&$0.96$&$0.94$&$0.96$&$0.70$&$0.69$&$0.84$\tabularnewline
			$30$&$ 90$&&$0.67$&$0.69$&$0.60$&$0.67$&$0$&$0$&$0.23$&&$0.98$&$0.99$&$0.98$&$0.98$&$0.74$&$0.73$&$0.88$\tabularnewline
			$30$&$120$&&$0.87$&$0.88$&$0.83$&$0.87$&$0$&$0$&$0.32$&&$0.99$&$1.00$&$0.99$&$1.00$&$0.79$&$0.78$&$0.88$\tabularnewline
			$60$&$ 60$&&$0.13$&$0.13$&$0.07$&$0.13$&$0$&$0$&$0.01$&&$0.96$&$0.96$&$0.95$&$0.96$&$0.68$&$0.67$&$0.84$\tabularnewline
			$60$&$ 90$&&$0.49$&$0.50$&$0.43$&$0.50$&$0$&$0$&$0.09$&&$0.99$&$0.99$&$0.98$&$0.99$&$0.75$&$0.73$&$0.86$\tabularnewline
			$60$&$120$&&$0.77$&$0.78$&$0.74$&$0.78$&$0$&$0$&$0.19$&&$1.00$&$1.00$&$1.00$&$1.00$&$0.81$&$0.78$&$0.86$\tabularnewline
			$90$&$ 60$&&$0.05$&$0.05$&$0.03$&$0.05$&$0$&$0$&$0.00$&&$0.96$&$0.97$&$0.96$&$0.97$&$0.69$&$0.68$&$0.83$\tabularnewline
			$90$&$ 90$&&$0.38$&$0.38$&$0.30$&$0.36$&$0$&$0$&$0.05$&&$0.99$&$0.99$&$0.99$&$0.99$&$0.75$&$0.74$&$0.84$\tabularnewline
			$90$&$120$&&$0.71$&$0.71$&$0.61$&$0.70$&$0$&$0$&$0.08$&&$1.00$&$1.00$&$1.00$&$1.00$&$0.81$&$0.79$&$0.83$\tabularnewline
			\midrule
			\multicolumn{18}{c}{$t(3), \tau = 0.5$}\\
			\midrule
			$30$&$ 60$&&$0.20$&$0.20$&$0.14$&$0.19$&$0$&$0$&$0.02$&&$0.93$&$0.94$&$0.91$&$0.94$&$0.67$&$0.66$&$0.80$\tabularnewline
			$30$&$ 90$&&$0.51$&$0.53$&$0.44$&$0.52$&$0$&$0$&$0.11$&&$0.97$&$0.98$&$0.97$&$0.98$&$0.72$&$0.71$&$0.84$\tabularnewline
			$30$&$120$&&$0.74$&$0.76$&$0.70$&$0.75$&$0$&$0$&$0.19$&&$0.99$&$0.99$&$0.99$&$0.99$&$0.76$&$0.76$&$0.86$\tabularnewline
			$60$&$ 60$&&$0.03$&$0.04$&$0.02$&$0.04$&$0$&$0$&$0.00$&&$0.94$&$0.95$&$0.93$&$0.95$&$0.66$&$0.66$&$0.79$\tabularnewline
			$60$&$ 90$&&$0.30$&$0.31$&$0.26$&$0.30$&$0$&$0$&$0.02$&&$0.98$&$0.98$&$0.98$&$0.98$&$0.72$&$0.71$&$0.82$\tabularnewline
			$60$&$120$&&$0.61$&$0.62$&$0.56$&$0.60$&$0$&$0$&$0.08$&&$0.99$&$0.99$&$0.99$&$0.99$&$0.77$&$0.75$&$0.83$\tabularnewline
			$90$&$ 60$&&$0.01$&$0.01$&$0.00$&$0.01$&$0$&$0$&$0.00$&&$0.95$&$0.95$&$0.93$&$0.95$&$0.67$&$0.66$&$0.78$\tabularnewline
			$90$&$ 90$&&$0.18$&$0.18$&$0.14$&$0.18$&$0$&$0$&$0.01$&&$0.98$&$0.98$&$0.98$&$0.98$&$0.72$&$0.71$&$0.80$\tabularnewline
			$90$&$120$&&$0.49$&$0.50$&$0.41$&$0.49$&$0$&$0$&$0.03$&&$0.99$&$0.99$&$0.99$&$0.99$&$0.78$&$0.76$&$0.79$\tabularnewline			\bottomrule
	\end{tabular}}
	\caption{\small {Membership estimation based on Spectral (the proposed method), ZWZ19 \cite{zhang2019}, and the vanilla $k-$means method without variance information for Model 2 with $\tau=0.5$ and two error distributions. \label{tab:model2a}}}%
\end{table}

\endgroup

\textbf{Model 3:} 
{The last model we consider has the same specification as Model 1, except that we allow individuals to have varying time period lengths and individuals with shorter  panel length are expected to be estimated with larger standard error. This resembles many macroeconomic settings where individual units have varying panel length and hence individual based estimates are of very different quality. In the simulation, the panel lengths are a random draw from $\{30,60,90\}$ with equal probabilities. Results are summarized in Table \ref{tab:model5a}. The overall performance deteriorates in comparison to Table \ref{tab:model1a} since some individuals with shorter panel length are estimated with more noise. The spectral clustering methods ($S^g$ and S) perform comparably. Using just the diagonal information of the covariance matrix yields equally good performance, but not using the covariance information at all clearly performs worse. The vanilla k-means method again performs very similarly to spectral clustering without accounting for variance information. The method proposed by \cite{zhang2019} is competitive, improves upon estimates not using variance information, but is slightly inferior to PAM and our proposed spectral clustering methods. }

\vspace{.1in}

\begingroup
\tabcolsep = 3pt
\def\arraystretch{1}

\begin{table}[h]
\centering
\scalebox{0.8}{
	\begin{tabular}{ccccccccc|cccccccc}
		\toprule
		\multirow{2}{*}{n}& \multicolumn{1}{c}{}&& \multicolumn{7}{c|}{\textbf{Perfect Match}} &
		\multicolumn{7}{c}{\textbf{Average Match}} \\
		&  &  {$\op{S}^g$} & {S} & {PAM} &  {S-Diag} & {S-Iden}&  $k$-$\op{means}^\circ$ & {ZWZ19} & & {$\op{S}^g$} & {S} & {PAM} &  {S-Diag} &{S-Iden}&  $k$-$\op{means}^\circ$ & {ZWZ19} \\
		\midrule
		\multicolumn{17}{c}{$N(0,1), \tau = 0.5$}\\
		\midrule
		$30$&&$0.09$&$0.10$&$0.05$&$0.08$&$0$&$0$&$0.04$&&$0.92$&$0.93$&$0.86$&$0.92$&$0.61$&$0.61$&$0.88$\tabularnewline
		$60$&&$0.01$&$0.02$&$0.00$&$0.01$&$0$&$0$&$0.00$&&$0.94$&$0.94$&$0.87$&$0.93$&$0.59$&$0.59$&$0.90$\tabularnewline
		$90$&&$0.00$&$0.00$&$0.00$&$0.00$&$0$&$0$&$0.00$&&$0.93$&$0.93$&$0.88$&$0.93$&$0.59$&$0.58$&$0.90$\tabularnewline
		\midrule
		\multicolumn{17}{c}{$t(3), \tau = 0.5$}\\
		\midrule
		$30$&&$0.04$&$0.04$&$0.02$&$0.04$&$0$&$0$&$0.01$&&$0.89$&$0.90$&$0.82$&$0.89$&$0.58$&$0.58$&$0.84$\tabularnewline
		$60$&&$0.00$&$0.00$&$0.00$&$0.00$&$0$&$0$&$0.00$&&$0.91$&$0.91$&$0.83$&$0.91$&$0.57$&$0.56$&$0.86$\tabularnewline
		$90$&&$0.00$&$0.00$&$0.00$&$0.00$&$0$&$0$&$0.00$&&$0.91$&$0.91$&$0.83$&$0.91$&$0.56$&$0.55$&$0.87$\tabularnewline
		\bottomrule
\end{tabular}}
\caption{\small {Membership estimation based on Spectral (the proposed method), ZWZ19 \cite{zhang2019}, and the vanilla $k-$means method without variance information for Model 3 with $\tau=0.5$ and two error distributions. \label{tab:model5a}}}%
\end{table}

\endgroup

\subsubsection{Quantile regression with joint slope and grouping on intercepts} \label{sec:simqrint}

In this section, we consider the setting in Example~\ref{ex:qrint}. The spectral clustering approach is based on the estimators for the slopes and variances described in Section~\ref{sec:qrint}. More precisely, recall the definition of $\tilde \alpha_1, \tilde \vecb$ in~\eqref{eq:alphabetaqrint} and $\hat \Sigma_{i,j}$ defined in~\eqref{eq:hatsigma1}. 
The variation matrix~$\hat V$ which we use as input to the spectral clustering algorithm is given by~$\hat V_{ij}:=  \hat \Sigma_{i,j}^{-1/2}|\tilde\alpha_i-\tilde\alpha_j|\,. $ For comparison, we also consider spectral clustering setting all variance estimators set to be equal, the naive $k$-means approach on estimated $\tilde \alpha_i$ from~\eqref{eq:alphabetaqrint}, 
and the convex clustering procedure of \cite{gu2019panel}. Tuning parameters for \cite{gu2019panel} were set as described in the latter paper. 
The following model corresponds to DGP1 location scale shift model in~\cite{gu2019panel}.

\textbf{Model 4:} 
$$
y_{it}=\alpha_i+x_{it}\beta+(1+x_{it}\gamma)e_{it}\,.
$$
where
$e_{it}\overset{iid}{\sim}N(0,1) \text{ or } e_{it}\overset{iid}{\sim} t(3),$
$\alpha_i\in\{1,2,3\}~~\text{with the same proportions},$
and
$\beta=1,\gamma=0.1,x_{it}=\gamma_i+v_{it},$
where $\gamma_i$ and $v_{it}$ are independent and identically distributed from standard normal distribution over $i, t$, respectively.

{Tables~\ref{tab:model4a} summarizes the proportion of perfect classification and the average of the percentage of correct classification based on the proposed method with both exponential and the Gaussian kernel (denoted as S and $S^g$ respectively). Spectral clustering ignoring variance information is denoted as S-Iden, and $k$-means clustering on $\tilde \alpha_i$ is denoted as $k$-$\op{means}^\circ$ along with the PAM method for clustering. The procedure from~\cite{gu2019panel} is denoted as GV.}

{In this model, including variance information is not helpful (S versus S-Iden). A possible explanation for variance information not being useful in this model is that the $\alpha_i$ are one-dimensional and there are no directions of larger or smaller variation in their estimates. The PAM and vanilla k-means performs identical in this model. The key difference is that PAM picks a representative point as the center of a group while k-means will take a cluster based average, this does not materialize any differences for grouping estimation in this Model. The method proposed in \cite{gu2019panel}, which uses convex clustering method to group the intercept shows slightly inferior performance for smaller $T$, but is otherwise comparable for larger $T$. }

\begingroup
\tabcolsep = 3pt
\def\arraystretch{1}

\begin{table}[ht]
\centering
\scalebox{0.8}{
\begin{tabular}{ccccccccc|ccccccc}
	\toprule
	\multirow{2}{*}{n}& \multirow{2}{*}{T} &\multicolumn{1}{c}{}&  \multicolumn{6}{c|}{\textbf{Perfect Match}} &
	\multicolumn{6}{c}{\textbf{Average Match}} \\
	&  && {S}  & {$S^g$}& {S-Iden} & PAM & $k$-$\op{means}^\circ$ & GV && {S}  & {$S^g$}& {S-Iden} & PAM & $k$-$\op{means}^\circ$ & GV \\
	\midrule
	\multicolumn{15}{c}{$N(0,1)$, $\tau=0.5$}\\
	\midrule
	$30$&$15$&&$0.07$&$0.08$&$0.06$&$0.05$&$0.05$&$0.03$&&$0.902$&$0.904$&$0.891$&$0.887$&$0.905$&$0.670$\tabularnewline
	$30$&$30$&&$0.52$&$0.53$&$0.52$&$0.45$&$0.49$&$0.39$&&$0.978$&$0.978$&$0.978$&$0.971$&$0.978$&$0.880$\tabularnewline
	$30$&$60$&&$0.94$&$0.94$&$0.94$&$0.91$&$0.92$&$0.91$&&$0.998$&$0.998$&$0.998$&$0.997$&$0.997$&$0.988$\tabularnewline
	$60$&$15$&&$0.00$&$0.00$&$0.01$&$0.01$&$0.01$&$0.00$&&$0.915$&$0.915$&$0.914$&$0.896$&$0.916$&$0.663$\tabularnewline
	$60$&$30$&&$0.35$&$0.36$&$0.36$&$0.27$&$0.34$&$0.22$&&$0.982$&$0.982$&$0.982$&$0.977$&$0.982$&$0.903$\tabularnewline
	$60$&$60$&&$0.91$&$0.90$&$0.90$&$0.87$&$0.90$&$0.86$&&$0.998$&$0.998$&$0.998$&$0.998$&$0.998$&$0.994$\tabularnewline
	$90$&$15$&&$0.00$&$0.00$&$0.00$&$0.00$&$0.00$&$0.00$&&$0.913$&$0.913$&$0.913$&$0.899$&$0.915$&$0.662$\tabularnewline
	$90$&$30$&&$0.19$&$0.20$&$0.20$&$0.16$&$0.17$&$0.12$&&$0.982$&$0.983$&$0.982$&$0.979$&$0.982$&$0.908$\tabularnewline
	$90$&$60$&&$0.89$&$0.88$&$0.88$&$0.84$&$0.87$&$0.83$&&$0.999$&$0.999$&$0.999$&$0.998$&$0.999$&$0.994$\tabularnewline
	\midrule
	\multicolumn{15}{c}{$t(3)$, $\tau=0.5$}\\
	\midrule
	$30$&$15$&&$0.01$&$0.01$&$0.02$&$0.01$&$0.02$&$0.01$&&$0.875$&$0.876$&$0.851$&$0.855$&$0.877$&$0.638$\tabularnewline
	$30$&$30$&&$0.34$&$0.34$&$0.34$&$0.29$&$0.32$&$0.23$&&$0.962$&$0.962$&$0.957$&$0.953$&$0.962$&$0.836$\tabularnewline
	$30$&$60$&&$0.88$&$0.88$&$0.88$&$0.83$&$0.85$&$0.80$&&$0.996$&$0.996$&$0.995$&$0.993$&$0.995$&$0.975$\tabularnewline
	$60$&$15$&&$0.00$&$0.00$&$0.00$&$0.00$&$0.00$&$0.00$&&$0.882$&$0.882$&$0.876$&$0.861$&$0.885$&$0.595$\tabularnewline
	$60$&$30$&&$0.16$&$0.16$&$0.17$&$0.13$&$0.15$&$0.08$&&$0.971$&$0.971$&$0.970$&$0.964$&$0.970$&$0.850$\tabularnewline
	$60$&$60$&&$0.78$&$0.79$&$0.78$&$0.73$&$0.78$&$0.71$&&$0.996$&$0.996$&$0.996$&$0.994$&$0.996$&$0.983$\tabularnewline
	$90$&$15$&&$0.00$&$0.00$&$0.00$&$0.00$&$0.00$&$0.00$&&$0.887$&$0.888$&$0.885$&$0.863$&$0.890$&$0.550$\tabularnewline
	$90$&$30$&&$0.05$&$0.06$&$0.04$&$0.04$&$0.05$&$0.03$&&$0.969$&$0.969$&$0.969$&$0.962$&$0.969$&$0.851$\tabularnewline
	$90$&$60$&&$0.69$&$0.68$&$0.69$&$0.63$&$0.66$&$0.63$&&$0.996$&$0.996$&$0.996$&$0.995$&$0.996$&$0.984$\tabularnewline
	\bottomrule
\end{tabular}}
\caption{\label{tab:model4a}\small{Membership estimation based on Spectral, $k$-means and the method proposed in \cite{gu2019panel} (GV) for Model 4 with $\tau=0.5$.}}
\end{table}
\endgroup

\subsection{Determining the Number of Groups}\label{sec:simnumgr}

In this section, we compare the proposed heuristic in~\eqref{eq:selGstar} for selecting the number of groups with other proposals from the literature. A general principle for determining the number of clusters using cross-validation (CV) in combination with the stability of cluster assignments was proposed by~\cite{wang2010} and adapted to quantile regression with grouping on the slopes in~\cite{zhang2019}. The underlying idea is directly applicable to any clustering algorithm, and hence we consider two versions: CV-kmeans corresponding to the proposal of~\cite{zhang2019}, and CV-Spectral which uses spectral clustering as proposed in the present paper as the underlying clustering algorithm. The maximum numbers of clusters to consider, denoted by $\op{Gmax}$, is set to $10$ throughout. 
{Results for Model 1 are presented in Table~\ref{tab:model1estGa} and those for Model 2 are summarized in Table \ref{tab:model3estGa}.} {All results reported in this section are based on {500} simulation repetitions.}

For Model 1, the proposed heuristic has the best performance for all settings except for $t(3)$ errors with $n = 30, 60, T = 60$ where the CV--Spectral outperforms slightly. CV--Spectral shows better performance than CV--kmeans consistently. We note that Model 1 is perfectly symmetric with an odd number of groups, this corresponds to a setting that is favorable for stability--based methods. Model 2 demonstrates a situation where stability based method performs badly.

Model 2 corresponds to an even number of groups, and both CV methods fail in this setting because they always pick $2$ groups. In light of the findings in~\cite{ben2006sober}, this is not surprising; see also~\cite{von2010clustering}. The issue is that a wrong grouping with two groups corresponding to coefficients $(0.1,0.1), (0.2,0.2)$ in one group and $(3,3), (3.1,3.1)$ in the other is very stable under variations of the data which leads to confusion of the stability--based methods. In the online appendix, we plot the paths of cross-validated stability scores for different $n,T$ combinations and different realizations of the data (Section~\ref{app:plots}  Figure~\ref{fig:4}). For larger $T$ there is a local minimum at the true number of groups $G = 4$, but the global minima are always at $G = 2$. The proposed heuristic works reasonably well and is able to pick up the correct number of groups as $T$ increases. 

Model 4 corresponds to common slopes and group structure on the intercept (see also Example~\ref{ex:qrint}). Since this setting was also considered in~\cite{gu2019panel}, we  consider the information criterion proposed in there. Results are presented in  Table~\ref{tab:model4estG0.5}. We also include cross--validation with spectral clustering, denoted by CV-spectral, for comparison. Note that CV--kmeans is not applicable in this setting. 

For $\tau = 0.5, n= 30, T = 15$, the best performing method is \cite{gu2019panel} with about a $10\%-15\%$ advantage over the other two methods which show comparable performance. In all other settings, CV-Spectral is the best or close to best (within $5\%$) performer. The heuristic method performs better or is similar to \cite{gu2019panel} for most cases with $n = 90, T \geq 60$ while the results between those two are mixed in other settings.

In conclusion, there is no clear winner that performs best across all models and settings. This is not surprising because selecting the number of clusters is a very difficult problem in general. {This also explains why there exists no unifying approach for selecting the number of groups.} Our proposed eigenvalue heuristic is competitive in most cases considered, and clearly the best on some. Stability--based methods have two major limitations: they cannot select one group by construction, and they can fail for models with stable clusters for the wrong number of groups. The information criterion in~\cite{gu2019panel} can select one group and performs well when $n,T$ are smaller but falls behind when $n$ is large. No information criterion is known for quantile regression models with unrestricted intercepts and grouping on the slopes. Such a criterion could potentially be derived, but it would only be valid in this specific setting and we refrained from taking this route since we aimed to propose a method that is applicable in more generality.

\begingroup
\tabcolsep = 7pt
\def\arraystretch{1}

\begin{table}[ht]
\centering
\scalebox{0.8}{
\begin{tabular}{cc|ccccc|ccccc}
\toprule
\multirow{2}{*}{n}& \multirow{2}{*}{T} &  \multicolumn{5}{c|}{ $ N(0,1)$} &
\multicolumn{5}{c}{ $ t(3)$} \\
&    & \bf{1}  & \bf{2}    & \bf{3}    & \bf{4}    & \bf{$\ge$5}  & \bf{1}    & \bf{2}    & \bf{3}    & \bf{4}    & \bf{$\ge$5}  \\
\midrule
\multicolumn{12}{c}{ CV-Spectral, $\tau=0.5$} \\
\midrule
30 & 60   & -- & 0.05 & \textbf{0.84} &  0.09 & 0.02 & -- & 0.08 & \textbf{0.72} & 0.15 & 0.05 \\
30 & 90   & -- & 0.01 & \textbf{0.98} &  0.01 & 0.00 & -- & 0.03 & \textbf{0.91} & 0.06 & 0.00\\
30 & 120 & -- & 0.01 & \textbf{0.99} &  0.00 & 0.00 & -- & 0.01 & \textbf{0.98} & 0.01 & 0.00\\
60 & 60   & -- & 0.01 & \textbf{0.98} &  0.01 & 0.00 & -- & 0.01 & \textbf{0.95} & 0.02 & 0.02\\
60 & 90   & -- & 0.00 & \textbf{1.00} &  0.00 & 0.00 & -- & 0.00 & \textbf{0.99} & 0.01 & 0.00\\
60 & 120 & -- & 0.00 & \textbf{1.00} &  0.00 & 0.00 & -- & 0.00 & \textbf{1.00} & 0.00 & 0.00\\
\midrule
\multicolumn{12}{c}{ CV-kmeans, $\tau=0.5$} \\
\midrule
30 & 60   & -- & 0.29 & \textbf{0.40} & 0.17 & 0.14 & -- & 0.34 & \textbf{0.32} & 0.17 & 0.17 \\
30 & 90   & -- & 0.13 & \textbf{0.69} & 0.12 & 0.06 & -- & 0.19 & \textbf{0.58} & 0.16 & 0.07 \\
30 & 120 & -- & 0.13 & \textbf{0.80} & 0.06 & 0.01 & -- & 0.13 & \textbf{0.72} & 0.11 & 0.04\\
60 & 60   & -- & 0.09 & \textbf{0.39} & 0.25 & 0.27 & -- & 0.13 & \textbf{0.27} & 0.16 & 0.44\\
60 & 90   & -- & 0.06 & \textbf{0.74} & 0.15 & 0.05 & -- & 0.07 & \textbf{0.63} & 0.18 & 0.12\\
60 & 120 & -- & 0.05 & \textbf{0.85} & 0.08 & 0.02 & -- & 0.03 & \textbf{0.78} & 0.15 & 0.04\\
\midrule
\multicolumn{12}{c}{ Heuristic, $\tau=0.5$} \\
\midrule
30 & 60   & 0.00 & 0.07 & \textbf{0.91} & 0.02 & 0.00 & 0.05 & 0.25 & \textbf{0.69} & 0.01 & 0.00 \\ 
30 & 90   & 0.00 & 0.00 & \textbf{1.00} & 0.00 & 0.00 & 0.00 & 0.00 & \textbf{0.98} & 0.02 & 0.00 \\
30 & 120 & 0.00 & 0.00 & \textbf{1.00} & 0.00 & 0.00 & 0.00 & 0.00 & \textbf{0.99} & 0.01 & 0.00 \\
60 & 60   & 0.01 & 0.00 & \textbf{0.98} & 0.01 & 0.00 & 0.09 & 0.07 & \textbf{0.84} & 0.00 & 0.00 \\
60 & 90   & 0.00 & 0.00 &\textbf{1.00} & 0.00 & 0.00 & 0.00 & 0.00 & \textbf{1.00} & 0.00 & 0.00 \\
60 & 120 & 0.00 & 0.00 & \textbf{1.00} & 0.00 & 0.00 & 0.00 & 0.00 & \textbf{1.00} & 0.00 & 0.00 \\
\bottomrule
\end{tabular}}
\caption{\small{Percentage of estimated number of groups based on CV-Spectral, CV-kmeans, and Heuristic for Model 1 with $\tau=0.5$. The true $G$ is 3 (highlighted column).}
} \label{tab:model1estGa}
\end{table}
\endgroup

\begingroup
\tabcolsep = 7pt
\def\arraystretch{1}

\begin{table}[ht]
\centering
\scalebox{0.8}{
\begin{tabular}{cc|ccccc|cccccc}
\toprule
\multirow{2}{*}{n}& \multirow{2}{*}{T} &  \multicolumn{5}{c|}{ $ N(0,1)$} &
\multicolumn{5}{c}{ $ t(3)$} \\
&    & \bf{1} &  \bf{2}    & \bf{3}    & \bf{4}    & \bf{$\ge$5}  & \bf{1} &  \bf{2}    & \bf{3}    & \bf{4}    & \bf{$\ge$5}  \\
\midrule
\multicolumn{12}{c}{CV-Spectral, $\tau=0.5$} \\
\midrule
40 & 40  & -- & 1.00 & 0.00 & \textbf{0.00} & 0.00 & -- & 1.00 & 0.00 & \textbf{0.00} & 0.00 \\
40 & 80  & -- & 1.00 & 0.00 & \textbf{0.00} & 0.00 & -- & 1.00 & 0.00 & \textbf{0.00} & 0.00\\
40 & 160 &-- & 1.00 & 0.00 & \textbf{0.00} & 0.00 & -- & 1.00 & 0.00 & \textbf{0.00} & 0.00   \\
60 & 40  & -- & 1.00 & 0.00 & \textbf{0.00} & 0.00 & -- & 1.00 & 0.00 & \textbf{0.00} & 0.00    \\
60 & 80  & -- & 1.00 & 0.00 & \textbf{0.00} & 0.00 & -- & 1.00 & 0.00 & \textbf{0.00} & 0.00 \\
60 & 160 &-- & 1.00 & 0.00 & \textbf{0.00} & 0.00 & -- & 1.00 & 0.00 & \textbf{0.00} & 0.00  \\
\midrule
\multicolumn{12}{c}{CV-kmeans, $\tau=0.5$} \\
\midrule
40 & 40  & -- & 1.00 & 0.00 & \textbf{0.00} & 0.00 & -- & 1.00 & 0.00 & \textbf{0.00} & 0.00 \\
40 & 80  & -- & 1.00 & 0.00 & \textbf{0.00} & 0.00 & -- & 1.00 & 0.00 & \textbf{0.00} & 0.00\\
40 & 160 &-- & 1.00 & 0.00 & \textbf{0.00} & 0.00 & -- & 1.00 & 0.00 & \textbf{0.00} & 0.00   \\
60 & 40  & -- & 1.00 & 0.00 &\textbf{0.00} & 0.00 & -- & 1.00 & 0.00 & \textbf{0.00} & 0.00    \\
60 & 80  & -- & 1.00 & 0.00 & \textbf{0.00} & 0.00 & -- & 1.00 & 0.00 & \textbf{0.00} & 0.00 \\
60 & 160 &-- & 1.00 & 0.00 &\textbf{0.00} & 0.00 & -- & 1.00 & 0.00 & \textbf{0.00} & 0.00  \\
\midrule
\multicolumn{12}{c}{Heuristic, $\tau=0.5$} \\
\midrule
40 & 40  & 0.00 & 0.70 & 0.00 & \textbf{0.30} & 0.00 & 0.00 & 0.92 & 0.00 & \textbf{0.08} & 0.00 \\
40 & 80  & 0.00 & 0.00 & 0.00 & \textbf{0.99} & 0.01 & 0.00 & 0.02 & 0.00 & \textbf{0.97} & 0.01 \\
40 & 160 &0.00 & 0.00 & 0.00 & \textbf{0.99} & 0.01 & 0.00 & 0.00 & 0.00 & \textbf{1.00} & 0.00   \\
60 & 40  & 0.00 & 0.49 & 0.00 & \textbf{0.51} & 0.00 & 0.00 & 0.91 & 0.00 & \textbf{0.09} & 0.00    \\
60 & 80  & 0.00 & 0.00 & 0.00 & \textbf{1.00} & 0.00 & 0.00 & 0.01 & 0.00 & \textbf{0.99} & 0.00 \\
60 & 160 &0.00 & 0.00 & 0.00 & \textbf{1.00} & 0.00 & 0.00 & 0.00 & 0.00 & \textbf{1.00} & 0.00  \\
\bottomrule
\end{tabular}}
\caption{ \label{tab:model3estGa}\small{Percentage of estimated number of groups with CV-Spectral, CV-kmeans and Heuristic methods for Model 2 with $\tau=0.5.$}. The true $G$ is 4 (highlighted column).}
\end{table}
\endgroup

\begingroup
\tabcolsep = 7pt
\def\arraystretch{1}

\begin{table}[ht]
\centering
\scalebox{0.8}{
\begin{tabular}{cc|ccccc|ccccc}
\toprule
\multirow{2}{*}{n}& \multirow{2}{*}{T} &  \multicolumn{5}{c|}{$ N(0,1)$} &
\multicolumn{5}{c}{ $t(3)$} \\
&    & \bf{1} & \bf{2}    & \bf{3}    & \bf{4}    & \bf{$\ge$5}    & \bf{1}  &  \bf{2}    & \bf{3}    & \bf{4}    & \bf{$\ge$5}  \\
\midrule
\multicolumn{12}{c}{ CV-Spectral, $\tau=0.5$} \\
\midrule
30 & 15 & -- & 0.35 & \textbf{0.43} & 0.09 & 0.13 & -- & 0.45 & \textbf{0.31} & 0.09 & 0.15 \\
30 & 30 & -- & 0.04 & \textbf{0.92} & 0.03 & 0.01 & -- & 0.11 & \textbf{0.79} & 0.07 & 0.03 \\
30 & 60 & -- & 0.00 & \textbf{1.00} & 0.00 & 0.00 & -- & 0.01 & \textbf{0.99} & 0.00 & 0.00 \\  
60 & 15 & -- & 0.13 & \textbf{0.63} & 0.02 & 0.22 & -- & 0.20 & \textbf{0.44} & 0.01 & 0.35 \\
60 & 30 & -- & 0.00 & \textbf{1.00} & 0.00 & 0.00 & -- & 0.01 & \textbf{0.97} & 0.01 & 0.01 \\
60 & 60 & -- & 0.00 & \textbf{1.00} & 0.00 & 0.00 & -- & 0.00 & \textbf{1.00} & 0.00 & 0.00 \\
90 & 15 & -- & 0.09 & \textbf{0.79} & 0.00 & 0.12 & -- & 0.18 & \textbf{0.61} & 0.01 & 0.20 \\ 
90 & 30 & -- & 0.00 & \textbf{1.00} & 0.00 & 0.00 & -- & 0.00 & \textbf{1.00} & 0.00 & 0.00 \\
90 & 60 & -- & 0.00 & \textbf{1.00} & 0.00 & 0.00 & -- & 0.00 & \textbf{1.00} & 0.00 & 0.00 \\
\midrule
\multicolumn{12}{c}{ Heuristic, $\tau=0.5$} \\
\midrule
30 & 15 & 0.05 & 0.40 & \textbf{0.45} & 0.06 & 0.04 & 0.10 & 0.49 & \textbf{0.31} & 0.06 & 0.04 \\
30 & 30 & 0.01 & 0.04 & \textbf{0.92} & 0.02 & 0.01 & 0.00 & 0.15 & \textbf{0.79} & 0.03 & 0.03  \\
30 & 60 & 0.00 & 0.00 & \textbf{0.99} & 0.00 & 0.01 & 0.00 & 0.00 & \textbf{1.00} & 0.00 & 0.00 \\
60 & 15 & 0.18 & 0.37 & \textbf{0.44} & 0.00 & 0.01 & 0.50 & 0.28 & \textbf{0.22} & 0.00 & 0.00  \\
60 & 30 & 0.00 & 0.01 & \textbf{0.99} & 0.00 & 0.00 & 0.00 & 0.04 & \textbf{0.95} & 0.00 & 0.01\\
60 & 60 & 0.00 & 0.00 & \textbf{1.00} & 0.00 & 0.00 & 0.00 & 0.00 & \textbf{1.00} & 0.00 & 0.00  \\ 
90 & 15 & 0.02 & 0.31 & \textbf{0.64} & 0.01 & 0.02 & 0.12 & 0.40 & \textbf{0.46} & 0.00 & 0.02 \\  
90 & 30 & 0.00 & 0.00 & \textbf{1.00} & 0.00 & 0.00 & 0.00 & 0.02 & \textbf{0.98} & 0.00 & 0.00 \\
90 & 60 & 0.00 & 0.00 & \textbf{1.00} & 0.00 & 0.00 & 0.00 & 0.00 & \textbf{1.00} & 0.00 & 0.00 \\  
\midrule
\multicolumn{12}{c}{ GV, $\tau=0.5$} \\
\midrule
$30$&$15$&$0.00$&$0.06$&\textbf{0.51}&$0.34$&$0.09$&$0.00$&$0.16$&\textbf{0.50}&$0.27$&$0.08$\\
$30$&$30$&$0.00$&$0.00$&\textbf{0.81}&$0.16$&$0.03$&$0.00$&$0.01$&\textbf{0.76}&$0.20$&$0.04$\\
$30$&$60$&$0.00$&$0.00$&\textbf{0.98}&$0.02$&$0.00$&$0.00$&$0.00$&\textbf{0.96}&$0.03$&$0.00$\\
$60$&$15$&$0.00$&$0.04$&\textbf{0.52}&$0.32$&$0.12$&$0.00$&$0.11$&\textbf{0.44}&$0.33$&$0.12$\\
$60$&$30$&$0.00$&$0.00$&\textbf{0.86}&$0.13$&$0.02$&$0.00$&$0.00$&\textbf{0.78}&$0.18$&$0.04$\\
$60$&$60$&$0.00$&$0.00$&\textbf{0.99}&$0.01$&$0.00$&$0.00$&$0.00$&\textbf{0.98}&$0.02$&$0.00$\\
$90$&$15$&$0.00$&$0.03$&\textbf{0.51}&$0.32$&$0.14$&$0.00$&$0.11$&\textbf{0.37}&$0.33$&$0.19$\\
$90$&$30$&$0.00$&$0.00$&\textbf{0.88}&$0.11$&$0.01$&$0.00$&$0.00$&\textbf{0.79}&$0.17$&$0.03$\\
$90$&$60$&$0.00$&$0.00$&\textbf{0.99}&$0.01$&$0.00$&$0.00$&$0.00$&\textbf{0.98}&$0.02$&$0.00$\\
\bottomrule
\end{tabular}}
\caption{\label{tab:model4estG0.5}\small{Percentage of estimated number of groups based on CV-Spectral, Heuristic, and \cite{gu2019panel} (GV) methods for Model 3 with $\tau=0.5.$ }. The true $G$ is 3 (highlighted column).}
\end{table}
\endgroup

\section{Empirical Applications} \label{sec:data}
\subsection{Heterogeneity in environmental Kuznet curves}
We first apply our methodology to a panel data quantile regression analysis on the environmental Kuznet curves (EKC). The concept first emerged in the influential study of \cite{grossman1991environmental}. Various empirical studies have since then provided evidence in different countries that there exists an inverse-U relationship between economic development and the pollution level. As income per capita increases, we expect to see first deterioration of the environment, and then an improvement as income continues to rise. Understanding the relationship between pollution and per capita income is important for the design of the optimal environmental policy. Here we focus our analysis on using state-level panel data in the United States during the period of 1929 - 1994 and for brevity, we focus on the emission of $SO_2$. The dataset is available from the National Air Pollutant Emission Trends, 1900 - 1994, published by the US Environmental Protection Agency. Most early empirical work on EKC uses least squares methods pooling all the states together and utilizes either a quadratic or cubic specification to estimate the relationship between the emission level and per capita income. \cite{millimet2003environmental} discusses in detail some of the model specification issues and explores semi--parametric methods that provide a set of more flexible modelling tools. Given concerns that different states may take a different environmental transition path as income level arises, \cite{list1999environmental} estimates the EKC with both the quadratic and cubic specification state by state to account for potential state heterogeneity. They then group these states into three groups depending on whether the estimated peak of the state-specific EKC falls below, inside or above the 95\% confidence interval implied by a pooled model. This provides an interesting piece of evidence for some form of group heterogeneity, yet how group membership is constructed is ad hoc and does not account for the statistical uncertainty of the state-specific least square estimates. On the other hand, \cite{flores2014lessons} has criticized the least square approach and advocates the use of quantile regression methods. They document that quantile regression offers a more complete picture of the relationship between pollution and income. However, for a given quantile, they estimate the panel data quantile regression with state fixed effect without allowing the EKC coefficients to be state-dependent. Combining the insights of \cite{list1999environmental} and \cite{flores2014lessons}, we apply our methodology in a panel data quantile regression model {which allows individual fixed effects while estimating the group structure of the slope coefficients that determine the shape of the EKC curves across different states.}

For a given quantile level $\tau$, our model specification is: 
\[
q_{i,\tau}(Z_{it}) = \alpha_i(\tau) + \lambda_t(\tau) + Z_{it} \beta_{1, g_i}(\tau) + Z_{it}^2 \beta_{2, g_i}(\tau)\,,
\]
where $i$ corresponds to states, $t$ it the time index and $g_i$ records the group membership. We denote by $q_{i,\tau}(Z_{it})$ the conditional quantile function of $Y_{it}$ given $Z_{it}$ where the response $Y_{it}$  is the state-year per capita emission level of $SO_2$ and $Z_{it}$ is the per capita real income using 1987 dollar. We focus on the quadratic specification for better visualization of the estimation results. Cubic specification leads to similar grouping results. Other control variables can be added, for example, population density and the number of days with extreme temperature as considered in \cite{flores2014lessons}. However, \cite{flores2014lessons} report that these additional control variables do not change the estimates for the quadratics of the EKC. We first obtain state-specific estimates $\hat \beta_{1, i}(\tau), \hat \beta_{2,i}(\tau)$ as well as their associated covariance matrix. 

\begin{figure}[H]
\includegraphics[height=2in, width=\textwidth]{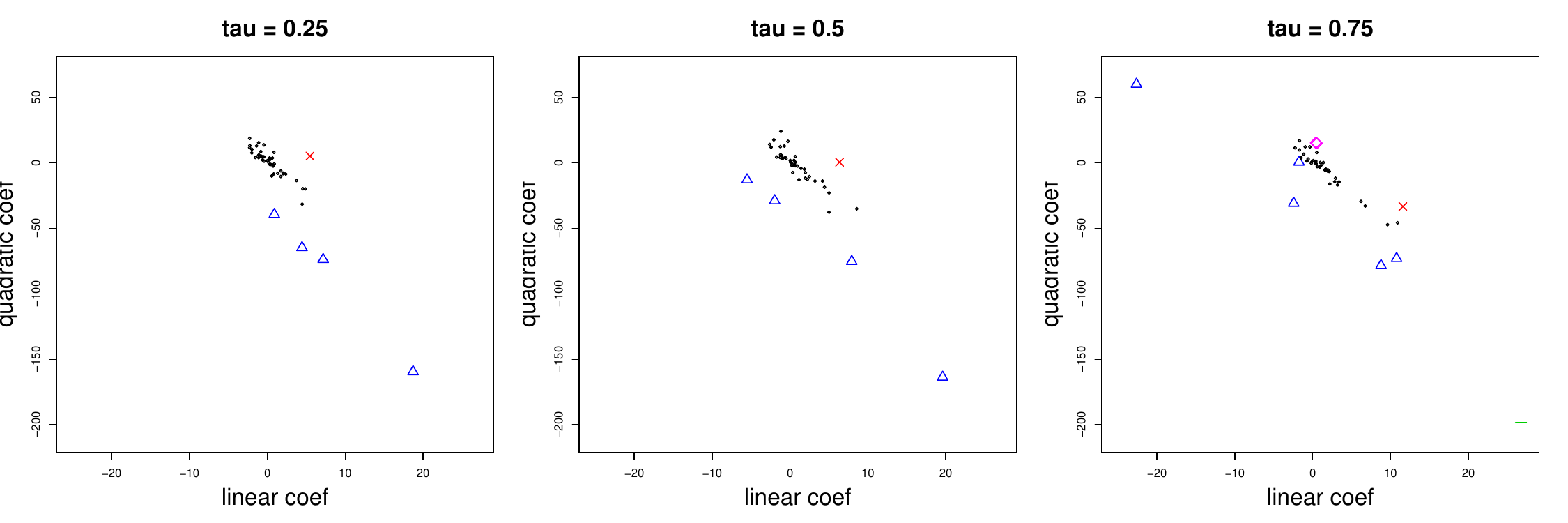}
\caption{\small{Raw state specific estimates for $\hat \beta_{1,i}(\tau), \hat \beta_{2,i}(\tau))$ and grouping of states for quantile levels $\tau = \{0.25, 0.5, 0.75\}$. Each symbol represents a different group. }}
\label{fig: group}
\end{figure}

To estimate the number of groups for different quantile levels $\tau = \{0.25, 0.5, 0.75\}$, we apply the heuristic in~\eqref{eq:selGstar}; see Figure \ref{fig: heuristic} in the online supplement for corresponding plots. For both 25 and 50th quantile, we find three groups and for 75th quantile, we find 5 groups. Given these estimates, we then apply the spectral clustering method on these raw estimates, accounting for the statistical uncertainty. Figure \ref{fig: group} shows the estimated group membership for $(\hat \beta_{1,i}(\tau), \hat \beta_{2,i}(\tau))$. Noticeably, for both 25th and 50th quantile, the grouping of the states are the same. The red cross in Figure \ref{fig: group} corresponds to West Virginia, while the blue triangles correspond to Arizona, Montana, Nevada, and Utah. A close inspection of the data suggests that the EKC for West Virginia looks to be closer to a linear trend within the range of years under consideration, while Arizona, Montana, Nevada, and Utah are states that have relatively higher emission level and a much more positive linear coefficient and a much negative quadratic coefficient when compared to all other states. Interestingly, these four states are also noted as the ``outlier" states in \cite{flores2014lessons} which documents that the residuals of these states are alarmingly high. Since their specification requires the EKC coefficients to be the same for all states, this provides some evidence that these states might have a different EKC. This is clearly confirmed by our analysis. For the 75th quantile, the State of Arizona has a more extreme estimate and now becomes a group by herself, as well as West Virginia. Two smaller groups consist of North Dakota and Wyoming as one group and Illinois, Montana, Nevada, New Mexico, and Utah as the other group. 

We note that some groups resulting from this empirical analysis are very small. The results should thus be interpreted with caution since this violates our theoretical assumptions which require proportional group sizes to be bounded from below. 

\subsection{Heterogeneity in intergenerational income mobility} 
The study on intergenerational income mobility across the United States by \cite{chetty2014land}, \cite{chetty2018opportunity} and \cite{chetty2018impacts} has been influential. Using tax records on the entire U.S. population, they document how children's expected incomes conditional on their parents' incomes differ across different geographical regions in the United States. Although the raw data used to obtain these estimates are not publicly available, they publish the region specific estimates at the commuting zone, country or census tract level, together with their associated standard errors. These estimates are used for policy purposes to encourage welfare improvements for children resides in the areas that have low mobility rates as for instance considered in \cite{bergman2019creating}. The categorization of a region having low mobility is often solely based on the point estimates without accounting for the associated statistical uncertainty. Our analysis focuses on the plausible hypothesis that although different geographical locations are likely to have heterogeneous mobility ratings, they may be divided into a few distinct groups and we let the data determine the number of groups utilizing both point estimates and their levels of precision. 	It is worth noting that, in contrast to most proposals in the existing literature, our method remains applicable even when raw individual-level data are not available due to privacy or other concerns and only estimated coefficients and their uncertainty estimates are given. 

We focus on the 100 most populous commuting zones. Let the point estimates of the income mobility to be $\hat \beta_i$ and the associated standard error to be $\hat \Sigma_i$. To apply the heuristic for the estimation of the number of groups, we also know $T_i$ which is the amount of data that leads to the estimates $(\hat \beta_i, \hat \Sigma_i)$.\footnote{All these information are publicly available from \url{https://opportunityinsights.org/data}. Note in this application, $T_i$ varies across individuals, we take the minimum $T_i$ when constructing the scaled dissimilarity measure for the estimation of $G$. Using the average value of $T_i$ leads to similar result. }  

We first use our method to select the number of groups. The left plot in Figure \ref{fig: eigen} shows that the number of groups is estimated to be nine and the right plot illustrates the gap of the adjacent eigen values. We then apply our algorithm to estimate the group membership, which is illustrated in Figure \ref{fig: effect}. Further details on the grouping of the hundred most populous commuting zones is provided in Table \ref{tab: effect}. Fayetteville and Memphis have the lowest point estimates for their income mobility among all the hundred commutting zones considered and they are grouped together. There are ten commuting zones grouped together as the top tier. The grouping provides a parsimonious description of the mobility heterogeneity. It also suggests that citizens in the commuting zones that belong to the same group, although having different point estimates, are likely to have similar true mobility ratings. 
\begin{figure}[H]\label{fig: eigen}
\centering
\includegraphics[scale=0.5]{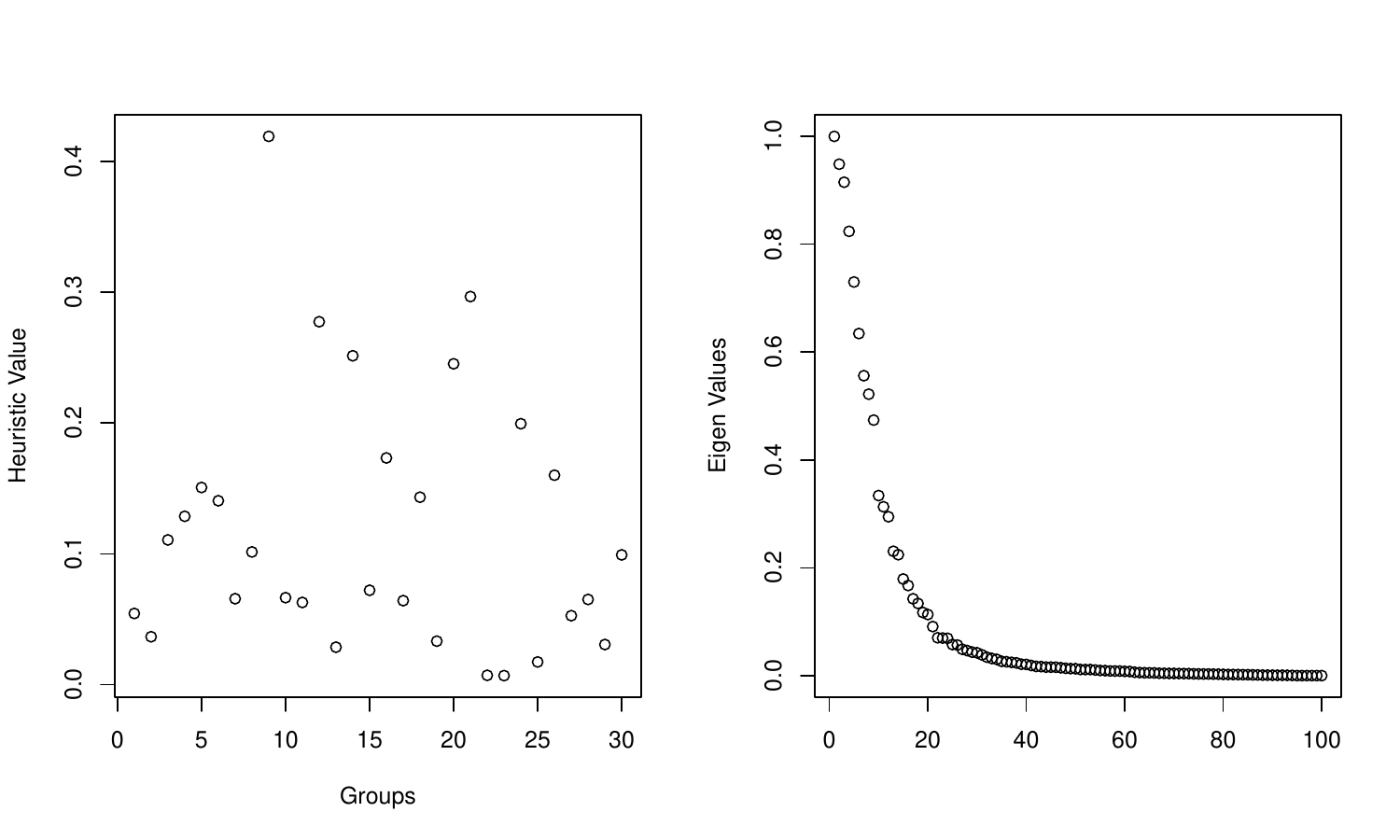}
\caption{The heuristic value for group selection and the associated eigen values for the 100 most populous commuting zones using publicly data in \cite{chetty2018impacts}.}
\end{figure}

\begin{figure}[H]\label{fig: effect} 
\centering
\includegraphics[scale=0.5]{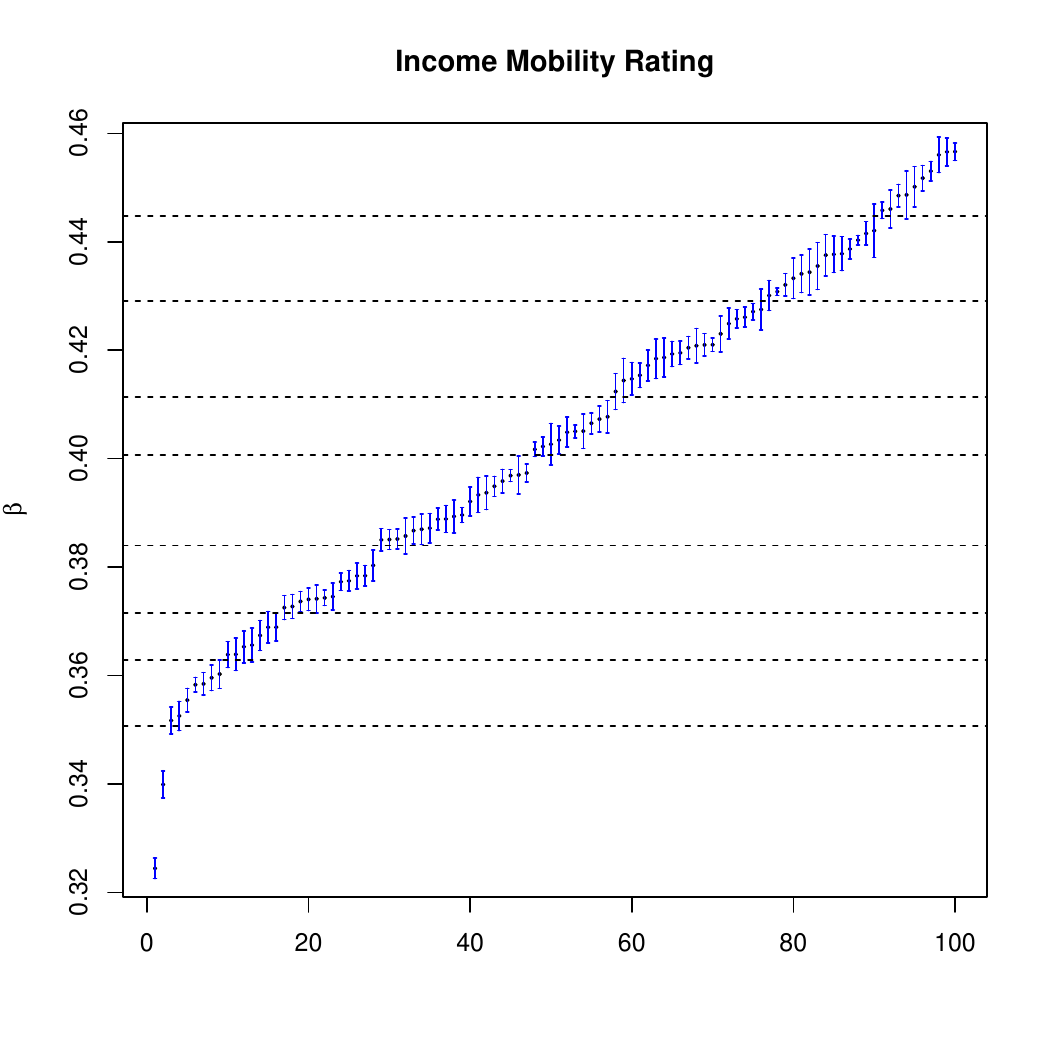}
\caption{Points in the figure are the sorted point estimates $\hat \beta_i$ for the 100 most populous commuting zones in the United States. The blue bars indicates the confidence set of each point estimates with $\pm$ 2 s.e. The dotted line are the division lines for the 9 groups based on the estimated group membership. } 
\end{figure}

\begin{table}[H]
\small
\centering
\resizebox{\textwidth}{!}{	\begin{tabular}{ll}
\toprule
& \textbf{Grouped Commuting Zones}  \\
\hline 
1	&  Boston, Des Moines, Honolulu, Minneapolis, Newark, Toms River, Salt Lake City\\
&  San Francisco, San Jose, Scranton \\
\hline 
2	&  Albany, Allentown, Brownsville, Los Angeles, Madison, Manchester, New York  \\
& Pittsburgh, Providence, Reading, Santa Barbara, Santa Rosa, Seattle, Spokane\\ 
\hline 
3	&  Bakersfield, Buffalo, Bridgeport, Canton, Denver, El Paso, Erie \\
&  Harrisburg, Houston,Modesto ,Omaha, Portland, Poughkeepsie, Sacramento\\
& San Diego, Springfield, Syracuse, Washington DC\\
\hline 
4	&  Austin, Eugene, Fort Worth, Miami, Oklahoma City, Philadelphia\\
& Rockford, San Antonio, Tulsa, Youngstown  \\
\hline 
5	&  Albuquerque, Baton Rouge, Cape Coral, Chicago, Cleveland, Dallas \\
&Fresno, Gary, Grand Rapids , Kansas City ,Las Vegas, Milwaukee, Orlando\\
&Port St. Lucie, Phoenix, Sarasota, South Bend, Toledo, Tucson \\
\hline 
6	&   Baltimore, Cincinnati, Columbus, Dayton, Detroit, Louisville \\
& Nashville, New Orleans, Pensacola, St. Louis, Tampa, Virginia Beach \\
\hline 
7	& Knoxville,  Indianapolis, Lakeland, Little Rock, Mobile, Raleigh, Richmond \\
\hline 
8	& Atlanta, Birmingham, Charlotte, Columbia, Greensboro, Greenville, Jacksonville  \\
\hline 
9	&  Fayetteville, Memphis\\
\bottomrule
\end{tabular}}
\caption{The 100 most populous commuting zones grouped using our method. First group are for those with the highest income mobility rating and the last group the lowest. } \label{tab: effect}
\end{table}

\section{Conclusion} \label{sec:conclusion}
In this paper, we propose a general methodology for studying group heterogeneity of effects in panel data models. We provide high-level conditions for the proposed method to achieve correct group identification and verify these conditions for several leading non-linear models often applied in empirical studies. We demonstrate that incorporating uncertainty information in individual-level estimates is useful for estimating group patterns. Although we focus on non-linear models, our methodology is naturally applicable to linear models, as well to situations where micro-level data is not available and only summary statistics are accessible to the researcher. We have proposed a method for selecting the number of groups, but left the theoretical validation of this method open to future research.

There are several important questions that merit further research. Our implementation of the dissimilarity measure and its theoretical analysis requires independence across individuals. In many settings, dependence across individuals is present. Properly modeling such dependence and accounting for it in our approach is an important question. 

Another substantial practical and theoretical challenge is dealing with short and highly unbalanced panels where individual level estimators are of very poor quality. In addition, as pointed out by the Associate Editor and a referee, there are cases when only some coefficients contain group structure. Using a two step procedure where coefficients containing group information are determined in a first step and grouping is only performed on those coefficients in a second step could lead to improvements in grouping accuracy. Implementing such an approach poses substantial theoretical challenges and is worthwhile investigating. The Associate Editor and referee also suggested In some cases, re-scaling the covariates first to bring all coefficients on the same scale before clustering could also be advantageous. This merits further exploration which we leave for future research. Finally, extending our approach to more complex models with time-varying group heterogeneity is another natural next step which we plan to address in future research.

\bibliographystyle{chicago}

\bibliography{References}

\newpage

\bigskip
\begin{center}
{\large\bf SUPPLEMENTARY MATERIAL}
\end{center}

In this supplementary material, we provide more simulation details, additional plots as well as proofs of the main results in Section~\ref{sec:theoryhighlevel} and Section~\ref{sec:theoryexamples}.

\section{Simulation Studies}

We provide more details about the simulation studies in Section~\ref{sec:sims}.

\subsection{Simulations in Section~\ref{sec:simlog}}

We implement the CLASSO estimator using CVX in Matlab with the mosek solver with version Mosek 8. 
The algorithm is initiated with $\vecb_i$ being the individual logistic regression estimates and $\eta_k$ being the origin for all $K_0$ groups. 
The algorithm is terminated when the objective function differs by a quantity less than 0.001 and when the $\ell_2$ norm of the estimated group centre $\eta_k$ changes by less than 0.1\%.

\subsection{Simulations in Section~\ref{sec:simqrslope}}
Simulations are done in the \textbf{quantreg} package in R. Covariance estimates are computed using the function \texttt{summary.rq()} with option \texttt{se="nid"} and default bandwidth choice \texttt{hs=true}.

\subsection{Simulations in Section~\ref{sec:simqrint} }
The bandwidth $d_T$ used in~\eqref{eq:hatsigma1} is based on the method implemented in the \textbf{quantreg} package in R (function \texttt{summary.rq()} with \texttt{se="nid"} and default choice \texttt{hs=true}).

\subsection{Simulations in Section~\ref{sec:simnumgr}  }
The maximum numbers of clusters to consider~$\op{Gmax}$ is set to $\op{Gmax}=5$ for $n<=30$ and $\op{Gmax}=10$ for $n>30$ cases for the CV methods, and we set $\op{Gmax}=10$ across all settings for the heuristic method. 
For the cross-validation method, we use 100 (4:4:2) random splits of the dataset into training and validation data (see \cite{wang2010} for details on the meaning of this splitting).

\section{Plots}\label{app:plots}

{
	\subsection{Plots in Section~\ref{sec:simquant}}
	\begin{figure}[H]
		\centering
\includegraphics[width=1\textwidth]{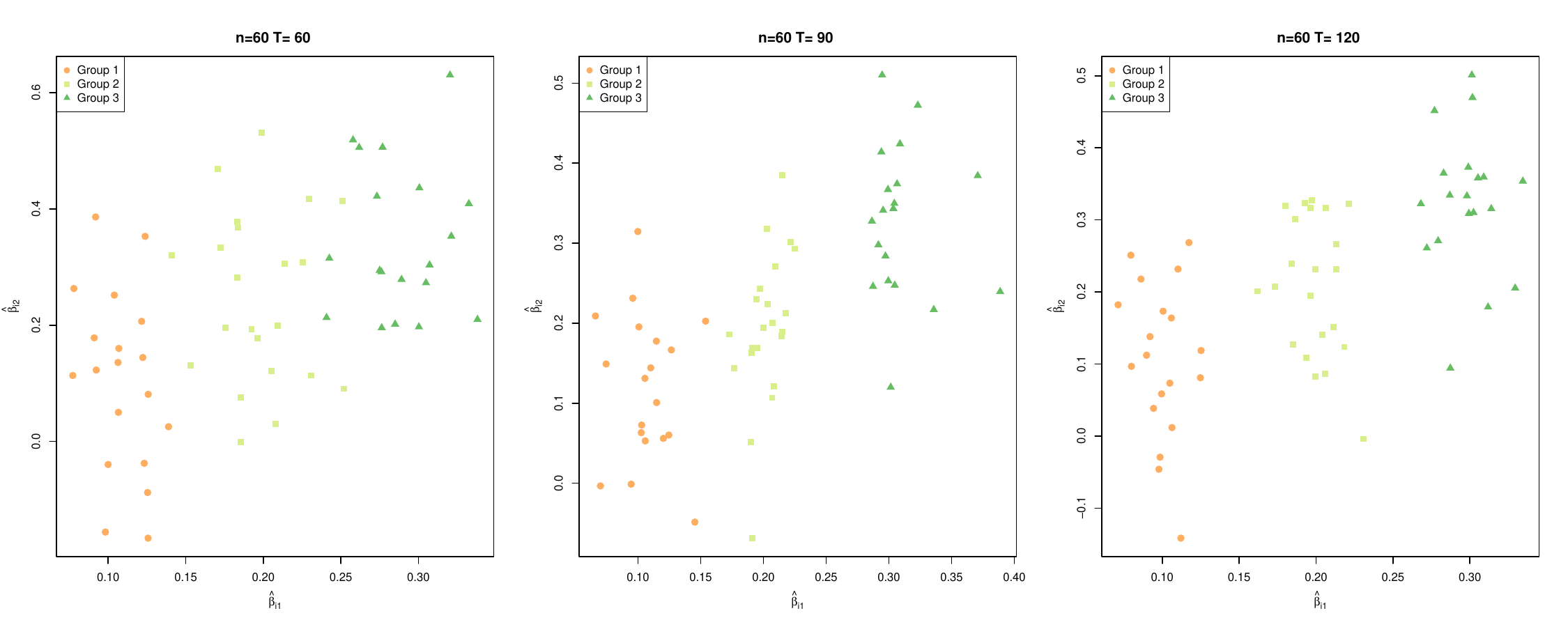}
\caption{\small{Scatter plots of $\{\hat{\vecb}_i\}_{i=1}^n$ for Model 1 with $t(3)$ error and $\tau=0.5.$}}
\label{fig:beta}
\end{figure} 
}

\subsection{Plots in Section~\ref{sec:simnumgr}}
\begin{figure}[H]
\centering
\subfigure[]{\includegraphics[width=0.4\textwidth]{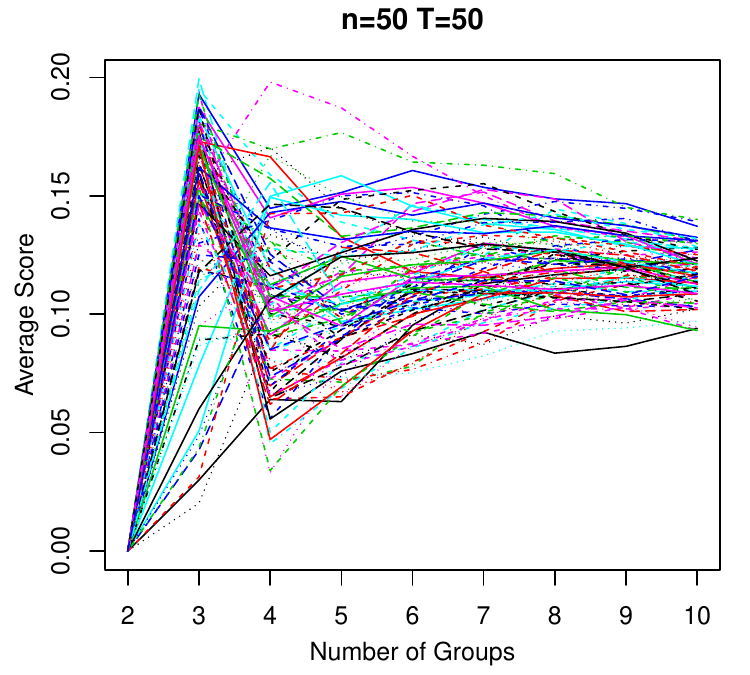}} 
\subfigure[]{\includegraphics[width=0.4\textwidth]{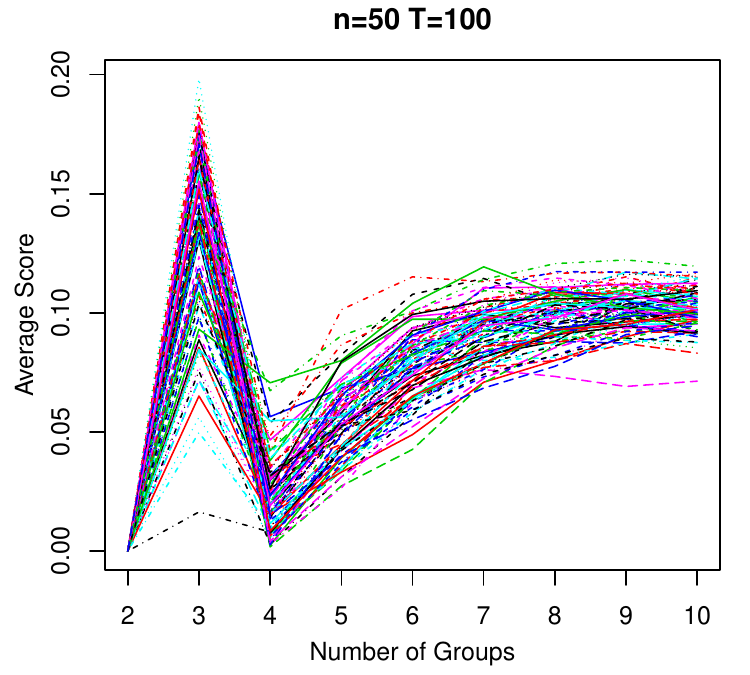}} 
\subfigure[]{\includegraphics[width=0.4\textwidth]{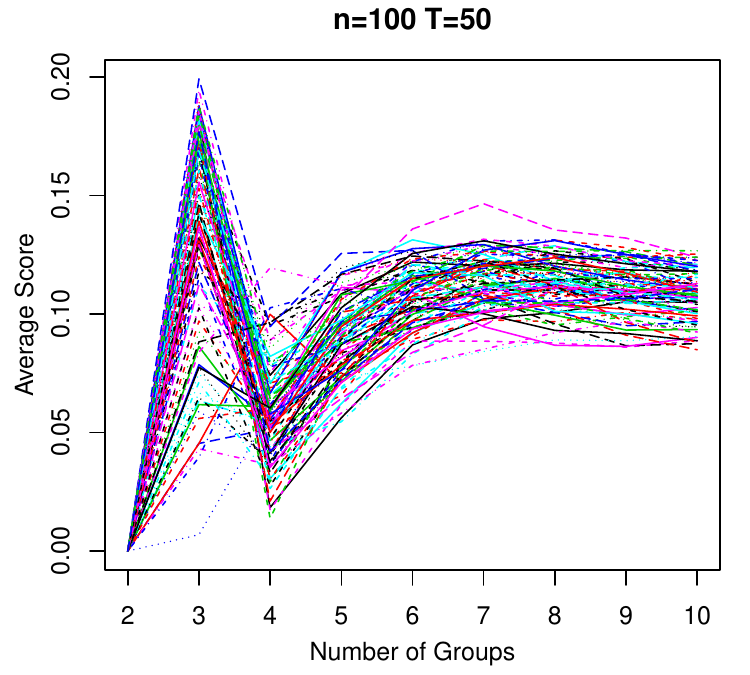}} 
\subfigure[]{\includegraphics[width=0.4\textwidth]{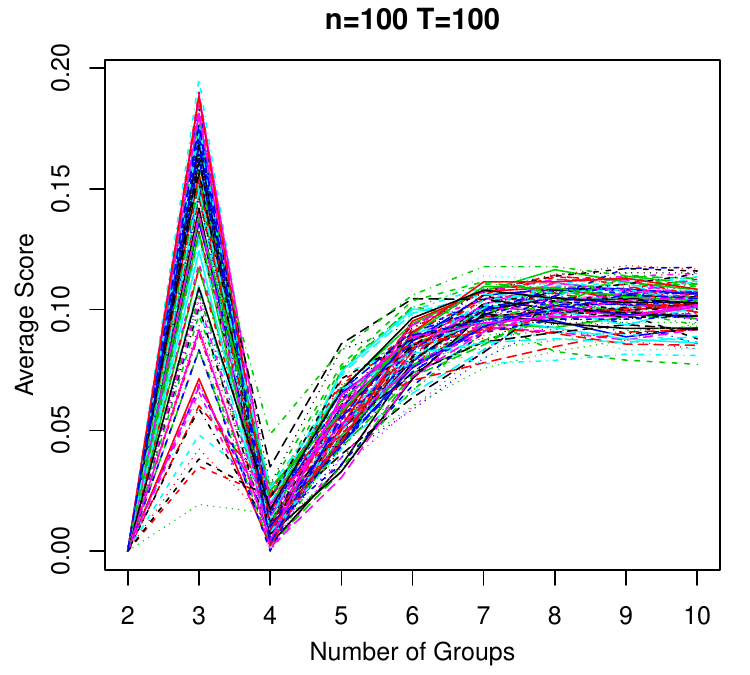}} 
\subfigure[]{\includegraphics[width=0.4\textwidth]{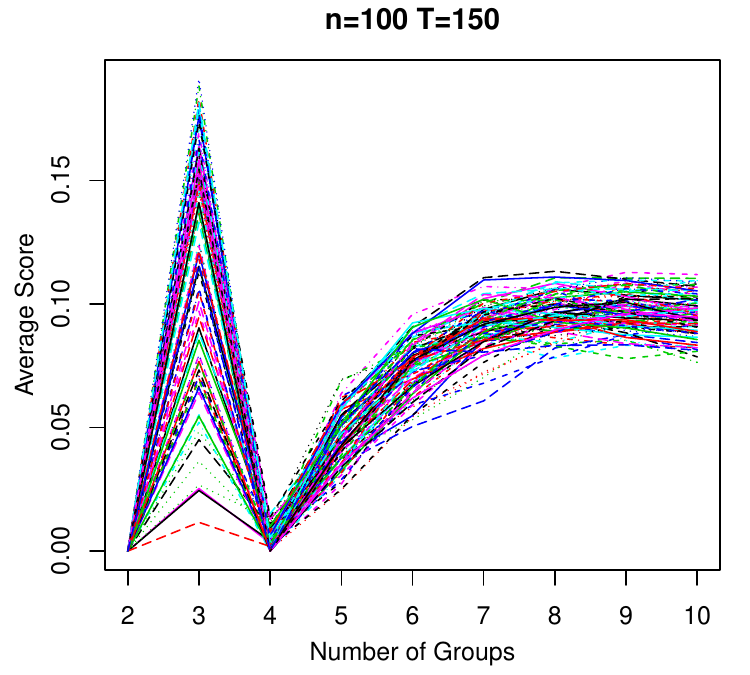}}
\caption{\small{Stability score for Model 3 with $t(3)$ error and $\tau=0.5.$}}
\label{fig:4}
\end{figure}

\subsection{Plots in Section~\ref{sec:data}}
\begin{figure}[H]
\includegraphics[scale=0.45]{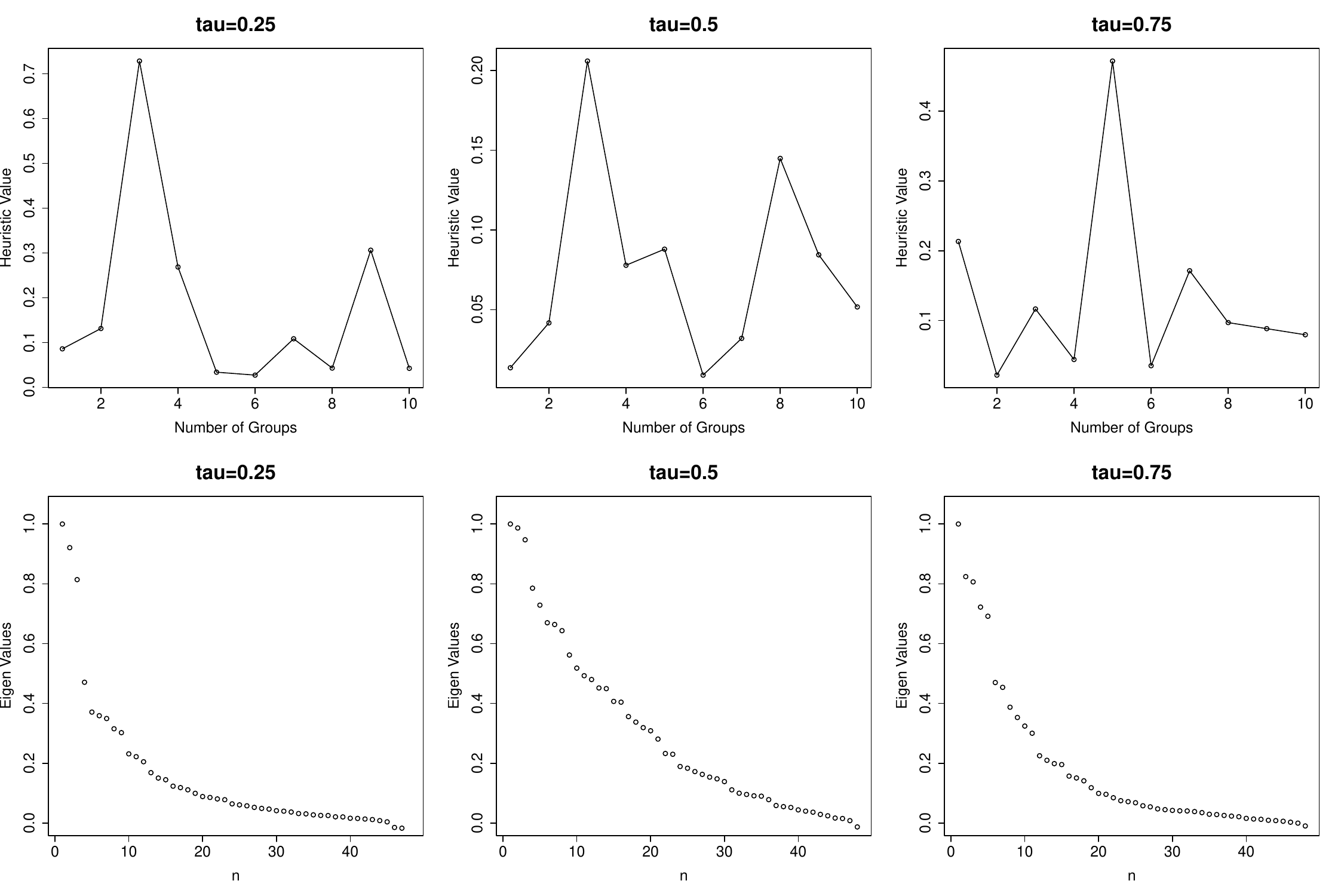}
\caption{\small{The heuristic values and the eigen-value plot for three different quantile levels $\tau = \{0.25, 0.5, 0.75\}$.}}
\label{fig: heuristic}
\end{figure}

\section{Proofs}

\subsection{Notation}
Let $a_n \lesssim_p b_n$ denotes that there exists
a non-random constant $C\in(0,\infty)$ that is independent of $n, T$, such that $\mpr(a_n\le Cb_n)\to 1.$ 
For a matrix~$A\in\rnp,$ we define the operator norm of $A$ as the maximum absolute column sum of the matrix
$
\normop{A}:=\max_{1\le i\le n}\sum_{j=1}^p |A_{ij}|\,,
$
define the Frobenius norm of $A$ as the square root of the sum of the absolute squares of all elements
$
\normf{A}:=\sqrt{\sum_{i=1}^n\sum_{j=1}^p |A_{ij}^2|}\,,
$
and define the spectral norm of $A$ as its largest singular value 
$
\normspec{A}:=\sigma_{\max} (A)\,.
$

To lighten notation we abbreviate the true number of groups as $G$ instead of $G^*$ whenever there is no risk of confusion.

\subsubsection{Proof of Theorem~\ref{thm:loc}}

The proof consists of two parts. First, we verify that under the assumptions made the following expansions and convergences hold. Second, we prove that~\eqref{eq:expansionL}-\eqref{eq:ratetildealpha} imply the statement of the Theorem.

For any (possibly random) sequence $\bm{\Delta}_{\vecg,T} := (\Delta_\alpha, \bm{\Delta}_\vecb) = \bigo_\mpr(1)$ we have an expansion of the form
\begin{multline}\label{eq:expansionL}
\sum_{t=1}^T  \Ll(\vecx_t,Y_t;\vecg^*+T^{-1/2}\bm{\Delta}_{\vecg,T}) - \Ll(\vecx_t,Y_t;\vecg^*)
\\
= \frac{1}{T^{1/2}} \sum_{t=1}^T \bm{\Delta}_{\vecg,T}^\top \grad_\vecg \Ll(\vecx_t,Y_t;\vecg^*) +  \frac{1}{2} \bm{\Delta}_{\vecg,T}^\top A \bm{\Delta}_{\vecg,T} + o_\mpr(1).
\end{multline}
Moreover, 
\begin{align}\label{eq:expansiongamma}
\sqrt{T}(\hat \vecg - \vecg^*) &= \frac{1}{T^{1/2}} \sum_{t=1}^T A^{-1}\grad_\vecg \Ll(\vecx_t,Y_t;\vecg^*) + o_\mpr(1),
\\	\label{eq:normalityBahadur}
\frac{1}{T^{1/2}} \sum_{t=1}^T \grad_\vecg \Ll(\vecx_t,Y_t;\vecg^*) &\Dkonv N(0,B)
\end{align}
for a non-degenerate covariance matrix $B$. Finally, letting 
\[
\tilde\alpha := \argmin_{\alpha} \sum_{t=1}^T \Ll(\vecx_t,Y_t;\alpha,\vecb^*+T^{-1/2}\Delta_\vecb)
\]
we have
\begin{equation}\label{eq:ratetildealpha}
\tilde \alpha - \alpha^* = \bigo_\mpr(T^{-1/2}).
\end{equation}

We now prove that~\eqref{eq:expansionL}--\eqref{eq:ratetildealpha} hold under Assumption~\ref{Ass:gen}. Consider the class of functions 
\[
\mathcal{F} := \Big\{ (\vecx,Y) \mapsto \Ll(\vecx,Y;\vecg): \vecg \in \tilde \Gamma \Big\}
\] 
where $\tilde \Gamma$ is the original parameter space if $\Gamma$ is bounded and a ball of Euclidean radius $1$ around $\vecg^*$ if $\Gamma$ is unbounded but the objective is convex. In both cases the bracketing numbers $N_{[\,]}(\eps,\mathcal{F},L_2(\mpr))$ are at most polynomial in $1/\eps$ (see Example~19.7 in \cite{vdV}). Thus Corollary 19.35 in \cite{vdV} implies 
\[
\sup_{\vecg \in \tilde\Gamma} \Big|\frac{1}{T}\sum_{t=1}^T \Ll(\vecx_t,Y_t;\vecg) - m(\vecg)\Big| = \bigo_\mpr(T^{-1/2}).
\]
Since the minimum is by assumption well-separated this implies $\hat \vecg - \vecg^* = o_\mpr(1)$ in the case where $\tilde\Gamma = \Gamma$. In the case of convexity we find that 
\begin{multline*}
\inf_{\vecg: \|\vecg-\vecg^*\| = 1} \Big|\frac{1}{T}\sum_{t=1}^T \Ll(\vecx_t,Y_t;\vecg) - 
\frac{1}{T}\sum_{t=1}^T \Ll(\vecx_t,Y_t;\vecg^*)\Big| 
\\
\geq \inf_{\vecg: \|\vecg-\vecg^*\| = 1} |m(\vecg) - m(\vecg^*)|+ \bigo_\mpr(T^{-1/2})  
\end{multline*}  
and hence by convexity of $\vecg \mapsto \frac{1}{T}\sum_{t=1}^T \Ll(\vecx_t,Y_t;\vecg)$ the minimizer of the latter must lie in ${\vecg: \|\vecg-\vecg^*\| = 1}$ with probability tending to one. This reduces the problem to the case of bounded parameter spaces $\Gamma$. In either case we have proved $\hat \vecg - \vecg^* = o_\mpr(1)$. Now~\eqref{eq:expansiongamma} follows from Theorem~5.23 in \cite{vdV} while the expansion in~\eqref{eq:expansionL} is established in the first line of the proof of the latter Theorem. Convergence in~\eqref{eq:normalityBahadur} follows from the CLT after observing that under Assumption~\ref{Ass:gen}(ii) $\Ll(\vecx_t,Y_t;\vecg^*)$ has a finite second moment. 

To establish~\eqref{eq:ratetildealpha}, note that the proof of Theorem 5.52 in \cite{vdV} yields the following more general result: assume that we have a sequence of functions $m_T: \Theta \to \R$ and estimated functions $\hat m_T$ such that for any sufficiently small $\delta > 0$
\begin{enumerate}
\item[(a)] $\sup_{\|\theta - \theta_T\|\leq \delta} m_T(\theta) - m_T(\theta_T) \leq C\delta^2$,
\item[(b)] $\E \Big[\sup_{\|\theta - \theta_T\|\leq \delta} \sqrt{T}|\hat m_T(\theta) - \hat m_T(\theta_T) - m_T(\theta) + m_T(\theta_T)|\Big] \leq C\delta$,
\item[(c)] $\hat m_T(\hat \theta) = \inf_{\theta} \hat m_T(\theta) + \bigo_\mpr(T^{-1})$.
\item[(d)] $\hat \theta = \theta_T + o_\mpr(1)$.
\end{enumerate}   
Then $\hat \theta - \theta_T = \bigo_\mpr(T^{-1/2})$. We will apply this with $ m_T(\theta) := \E[\Ll(\vecx,Y;\theta,\vecb^* + T^{-1/2}\bm\Delta_\vecb)]$, $\theta_T$ the well-separated global minimizer of $m_T(\theta)$ which exists by assumption for $T$ sufficiently large, and 
\[
\hat m_T(\theta) = T^{-1}\sum_{t=1}^T \Ll(\vecx_t,Y_t;\theta,\vecb^* + T^{-1/2}\bm\Delta_\vecb)
\]
Of those conditions, (a) follows by a Taylor expansion noting that the gradient of $m_T$ vanishes at $\theta_T$ and (c) follows by assuming the computed minimizer is sufficiently close to the global minimizer. Next observe that
\begin{multline*}
\sup_{\|\theta - \theta_T\|\leq \delta} \sqrt{T}|\hat m_T(\theta) - \hat m_T(\theta_T) - m_T(\theta) + m_T(\theta_T)| 
\\
\leq \sup_{\|\theta - \theta_T\|\leq \delta} \G_T(\Ll(\cdot;\theta,\vecb^* + T^{-1/2}\bm \Delta_\vecb)-\Ll(\cdot;\theta_T,\vecb^* + T^{-1/2}\bm \Delta_\vecb))
\end{multline*}
where $\G_T$ denotes the empirical process corresponding to the observations $(\vecx_t,Y_t)_{t=1,\dots,T}$. The class of functions 
\[
\mathcal{F}_T := \Big\{(\vecx,Y) \mapsto \Ll(\vecx,Y;\theta,\vecb^* + T^{-1/2}\bm \Delta_\vecb)-\Ll(\vecx,Y;\theta_T,\vecb^* + T^{-1/2}\bm \Delta_\vecb): |\theta - \theta_T| \leq \delta \Big\} 
\]
has envelope $\dot m(\cdot) \delta$ and bracketing numbers satisfying $N_{[\,]}(\eps,\mathcal{F}_T,L_2(\mpr)) \lesssim \frac{\delta}{\eps}$ (compare Example 19.7 in \cite{vdV}) and thus by Corollary 19.35 in \cite{vdV} we have
\[
\E\Big[\sup_{f \in \mathcal{F}_T} |\G_T(f)| \Big] \lesssim \int_0^{\|\dot m\|_{P,2}\delta} \sqrt{1+\log(\delta/\eps)}d\eps \lesssim \delta 
\] 
where the last equality follows after a change of variables. This implies~(b). The statement in (d) follows by similar arguments as the proof of consistency of $\hat \vecg$ given earlier since we assumed that each $m_T$ has a unique and well-separated global minimizer. This completes the proof of (a)--(d) and hence~\eqref{eq:ratetildealpha}.

From now on assume that~\eqref{eq:expansionL}--\eqref{eq:ratetildealpha} hold. We first analyze $\hat k^{PAM}$. Let
\begin{align*}
\hat G_\alpha = T^{-1/2} \sum_{t=1}^T \grad_\alpha \Ll(\vecx_t,Y_t;\alpha,\vecb^*)\Big|_{\alpha=\alpha^*}\,,
\\
\hat G_\vecb = T^{-1/2} \sum_{t=1}^T  \grad_\vecb \Ll(\vecx_t,Y_t;\alpha^*,\vecb)\Big|_{\vecb = \vecb^*}\,.
\end{align*}
By~\eqref{eq:normalityBahadur} we have 
\[
(\hat G_\alpha,\hat G_\vecb^\top)^\top \Dkonv (G_{\alpha},G_\vecb^\top)^\top \sim N(0,B)
\]
and by~\eqref{eq:expansiongamma} and~\eqref{eq:normalityBahadur} 
\begin{equation}\label{eq:Bahadur}
T^{1/2}(\hat \alpha - \alpha^*, (\hat \vecb - \vecb^*)^\top)^\top = - A^{-1} (\hat G_\alpha,\hat G_\vecb^\top)^\top + o_\mpr(1) \Dkonv -A^{-1} (G_{\alpha},G_\vecb^\top)^\top.
\end{equation}
In what follows, for squared matrices $M$ of dimension $p+1$ consider the following block structures 
\[
M = \Big[ 
\begin{array}{cc}
M_{11} & M_{12}
\\
M_{21} & M_{22}
\end{array}
\Big]
\]
with $M_{11} \in \R$. With this notation we find that
\begin{align*}
\sqrt{T}(\hat \vecb - \vecb^*) &\Dkonv \Big[ [A^{-1}]_{21} ~~ [A^{-1}]_{22} \Big] ( G_\alpha,G_\vecb^\top)^\top = [A^{-1}]_{21} G_\alpha +  [A^{-1}]_{22} G_\vecb
\\
& = [A^{-1}]_{22} G_\vecb - \frac{1}{A_{11}}[A^{-1}]_{22}A_{21} G_\alpha = [A^{-1}]_{22} \Big(  G_\vecb - \frac{A_{21}}{A_{11}} G_\alpha \Big)
\end{align*}
where we used block matrix inversion combined with the fact that $A_{11}$ is a scalar. Denoting by $C$ the covariance matrix of $G_\vecb - \frac{A_{21}}{A_{11}} G_\alpha$, we have 
\[
\sqrt{T}(\hat \vecb - \vecb^*) \Dkonv N(0, [A^{-1}]_{22} C[A^{-1}]_{22}).
\] 
In this notation, $\Sigma_\vecb = [A^{-1}]_{22} C[A^{-1}]_{22}$. Since $\hat \Sigma_\vecb$ is a consistent estimator for $\Sigma_\vecb$ by assumption and since $\vecb^* = \vecb_1, \vecb_2 = \vecb^* + T^{-1/2}\Delta$ by assumption, we obtain by the definition of $\hat k^{PAM}$
\[
P(\hat k^{PAM} = 1) \to P\Big( \|Z + \Sigma_{\vecb}^{-1/2}\Delta\|_2 > \|Z\|_2\Big).
\] 
where $Z \sim N(0,I_{p\times p})$ with $p$ denoting he dimension of $\vecb$. This can be further simplified as follows
\begin{align*}
& \|Z + \Sigma_\vecb^{-1/2}\Delta\|_2 > \|Z\|_2 
\\
\iff & \|Z + \Sigma_\vecb^{-1/2}\Delta\|_2^2 > \|Z\|_2^2
\\
\iff & \|Z\|_2^2 + \|\Sigma_\vecb^{-1/2}\Delta\|_2^2 + 2Z^\top \Sigma_\vecb^{-1/2}\Delta > \|Z\|_2^2
\\
\iff & \|\Sigma_\vecb^{-1/2}\Delta\|_2^2 > - 2Z^\top \Sigma_\vecb^{-1/2}\Delta
\\
\iff & \frac{1}{2}\|\Sigma_\vecb^{-1/2}\Delta\|_2^2 > \|\Sigma_\vecb^{-1/2}\Delta\|_2 N(0,1).
\end{align*}
Thus
\begin{equation}\label{eq:limpobPAM}
P(\hat k^{PAM} = 1) \to \Phi(\|\Sigma_\vecb^{-1/2}\Delta\|_2/2).
\end{equation}
Next we derive the corresponding limit for $\hat k^{BM}$. Let
\begin{align*}
\tilde \Delta_\alpha &:= \sqrt{T}(\tilde\alpha - \alpha^*)
\\
\Delta_\alpha^* &:= -\frac{\hat G_\alpha + A_{12}\Delta}{A_{11}} .
\end{align*}
Apply the expansion in~\eqref{eq:expansionL} with $\bm\Delta_\vecg = (\tilde \Delta_\alpha, \Delta)$ and with $\bm\Delta_\vecg = (\Delta_\alpha^*, \Delta)$ and subtract those expansions to obtain
\begin{align*}
0 \geq& \sum_{t=1}^T  \Ll(\vecx_t,Y_t;\tilde\alpha,\vecb^*+T^{-1/2}\Delta) - \sum_{t=1}^T  \Ll(\vecx_t,Y_t;\alpha^* + T^{-1/2}\Delta_\alpha^* ,\vecb^*+T^{-1/2}\Delta)
\\
=& 	\tilde \Delta_\alpha \hat G_\alpha + \Delta^\top \hat G_\vecb +  \frac{1}{2} A_{11} \tilde \Delta_\alpha^2 + \frac{1}{2}  \Delta^\top A_{22} \Delta + \tilde \Delta_\alpha A_{12}\Delta
\\
& - \Big\{\Delta_\alpha^* \hat G_\alpha + \Delta^\top \hat G_\vecb +  \frac{1}{2} A_{11} (\Delta_\alpha^*)^2 + \frac{1}{2}  \Delta^\top A_{22} \Delta + \Delta_\alpha^* A_{12}\Delta \Big\} + o_\mpr(1)
\\
=& \frac{A_{11}}{2} (\Delta_\alpha^* - \tilde\Delta_\alpha)^2 + o_\mpr(1). 
\end{align*} 
where the inequality in the first line follows because $\tilde\alpha$ is defined as minimizer and the last line is obtained by plugging in the definition of $\Delta_\alpha^*$. This implies $\Delta_\alpha^* = \tilde \Delta_\alpha + o_\mpr(1)$. Next observe that
\begin{align*}
& \inf_\alpha \sum_{t=1}^T \Ll(\vecx_t,Y_t;\alpha,\vecb^*+T^{-1/2}\Delta) - \sum_{t=1}^T \Ll(\vecx_t,Y_t;\alpha^*,\vecb^*)
\\
& = \sum_{t=1}^T \Ll(\vecx_t,Y_t;\tilde \alpha,\vecb^*+T^{-1/2}\Delta) - \sum_{t=1}^T \Ll(\vecx_t,Y_t;\alpha^*,\vecb^*)
\\
& = \tilde \Delta_{\alpha}\hat G_\alpha + \Delta^\top \hat G_\vecb + \frac{1}{2} (\tilde \Delta_\alpha,\Delta^\top) A (\tilde \Delta_\alpha,\Delta^\top)^\top + o_\mpr(1)
\\
& = \Delta^\top \hat G_\vecb + \frac{1}{2} \Delta^\top A_{22} \Delta  - \frac{1}{2A_{11}} (\hat G_\alpha + A_{12}\Delta)^2 + o_\mpr(1)
\end{align*}
where we used the definition of $\tilde \alpha$ and the expansion $\Delta_\alpha^* = \tilde \Delta_\alpha + o_\mpr(1)$ in the last line and~\eqref{eq:expansionL} in the second to last line. Similarly, setting $\Delta = 0$ in the above expansion we find 
\[
\inf_\alpha \sum_{t=1}^T \Ll(\vecx_t,Y_t;\alpha,\vecb^*) - \sum_{t=1}^T \Ll(\vecx_t,Y_t;\alpha^*,\vecb^*) = - \frac{1}{2A_{11}} \hat G_\alpha^2 + o_\mpr(1).
\] 
Subtracting those two expansions and expanding the square in $(\hat G_\alpha + A_{12}\Delta)^2$ we find
\begin{align*}
\hat k^{BM} = 1 &\iff \Delta^\top \hat G_\vecb + \frac{1}{2} \Delta^\top A_{22} \Delta  - \frac{\hat G_\alpha A_{12}\Delta}{A_{11}} - \frac{(A_{12}\Delta)^2}{2A_{11}} + o_\mpr(1) > 0
\\
&\iff \Delta^\top \Big(\hat G_\vecb - \frac{\hat G_\alpha A_{21}}{A_{11}}\Big) + \frac{1}{2} \Delta^\top \Big(A_{22} - \frac{A_{21}A_{12}}{A_{11}} \Big)\Delta + o_\mpr(1) > 0 
\\
&\iff \Delta^\top \Big(\hat G_\vecb - \frac{\hat G_\alpha A_{21}}{A_{11}}\Big) + \frac{1}{2} \Delta^\top \Big[ [A^{-1}]_{22}\Big]^{-1}\Delta + o_\mpr(1) > 0 
\end{align*}
where the last line follows by block inversion for matrices since $A_{11}$ is a scalar. Thus
\begin{align}
P(\hat k^{BM} = 1) &\to P\Big( \frac{1}{2} \Delta\Big[ [A^{-1}]_{22}\Big]^{-1}\Delta > N(0,\Delta^\top C \Delta)\Big) \nonumber
\\
&= \Phi\Big( \frac{\Delta\Big[ [A^{-1}]_{22}\Big]^{-1}\Delta}{2(\Delta^\top C \Delta)^{1/2}} \Big). \label{eq:limpobBM}
\end{align}
To lighten notation, let $D := \Big[ [A^{-1}]_{22}\Big]^{-1}$. Note that by the Cauchy-Schwarz inequality and the definition of $\Sigma_\vecb$ 
\[
\frac{\Delta^\top D \Delta}{(\Delta^\top C \Delta)^{1/2}} 
= 
\frac{\Delta^\top C^{1/2} C^{-1/2} D \Delta}{(\Delta^\top C\Delta)^{1/2}} 
\leq 
\frac{\|\Delta^\top C^{1/2}\|_2 \|C^{-1/2} D \Delta\|_2}{(\Delta^\top C \Delta)^{1/2}} 
= 
(\Delta^\top \Sigma_\vecb^{-1} \Delta)^{1/2}.
\]
This inequality is strict unless $C^{1/2}\Delta$ is a scalar multiple of $C^{-1/2} D \Delta$. Thus 
\[
\lim_{T\to\infty} P(\hat k^{PAM} = 1) \geq \lim_{T\to\infty} P(\hat k^{BM} = 1) 
\]  
with strict inequality unless $C^{1/2}\Delta$ is a scalar multiple of $C^{-1/2} D \Delta$. 

For a proof that 
\[
\lim_{T\to \infty} \mpr\Big( \hat k^{PAM} = 1 \Big) \geq \lim_{T\to \infty} \mpr\Big( \hat k^{PAM,K_T} = 1 \Big)\,,
\]
we obtain by similar computations as above that 
\begin{align*}
&\lim_{T\to \infty} \mpr\Big( \hat k^{PAM,K_T} = 1 \Big) 
= \Phi\Big( \frac{\Delta^\top K^\top K\Delta}{2(\Delta^\top K^\top K\Sigma_{\vecb} K^\top K\Delta)^{1/2}} \Big) 
= \Phi\Big( \frac{\Delta K^\top K\Sigma_{\vecb}^{1/2} \Sigma_{\vecb}^{-1/2}\Delta}{2(\Delta^\top K^\top K\Sigma_{\vecb} K^\top K\Delta)^{1/2}} \Big)
\\
&\le \Phi\Big( \frac{(\Delta K^\top K\Sigma_{\vecb}K^\top K \Delta)^{1/2} (\Delta \Sigma_{\vecb}^{-1} \Delta)^{1/2}}{2(\Delta^\top K^\top K\Sigma_{\vecb}K^\top K \Delta)^{1/2}} \Big) 
= \Phi\Big( \frac{(\Delta \Sigma_{\vecb}^{-1} \Delta)^{1/2}}{2} \Big) = \lim_{T\to \infty} \mpr\Big( \hat k^{PAM} = 1 \Big)\,.
\end{align*}
\hfill $\Box$

\subsection{Proof of the generic spectral clustering results (Theorems~\ref{th:specnonasy},~\ref{thm:clusternew})} \label{sec:proofclust}

Since the result is trivial when $G^*=1$, we will without loss of generality assume that $G^* \geq 2$. We will further write $G$ instead of $G^*$ since there is no risk of confusion in this subsection. 

To simplify notation, we will without loss of generality assume that the units are ordered according to their true grouping, i.e. unit $1,...,|I_1^*|$ belong to group $1$, unit $|I_1^*|+1,...,|I_2^*|$ belong to group $2$, etc. This is to shorten notation only, all arguments will work with more complex notation if this assumption is dropped. 

To proceed to the proof, we first consider the decomposition $\hat A= \hat A_{\op{diag}}+\hat A_{\op{off-diag}},$ where 
\begin{equation*}
\hat A_{\op{diag}} := \begin{pmatrix}
\hat{A}^{(11)} & \mathbf{0}  & \cdots&\mathbf{0} \\ 
\mathbf{0}&  \hat A^{(22)} & \cdots&\mathbf{0}\\
\cdots & \cdots& \cdots&\cdots\\ 
\mathbf{0} &\mathbf{0}&\cdots &\hat A^{(GG)}
\end{pmatrix}\,
\end{equation*}
and 
\begin{equation*}
\hat A_{\op{off-diag}} := \begin{pmatrix}
\mathbf{0}  & \hat A^{(12)}  & \cdots& \hat A^{(1G)} \\ 
\hat A^{(21)}& \mathbf{0} & \cdots&\hat A^{(2G)}\\
\cdots & \cdots& \cdots&\cdots\\ 
\hat A^{(G1)} &\hat A^{(G2)}&\cdots &\mathbf{0} 
\end{pmatrix}\,,
\end{equation*}
{with} $\hat A^{(ij)}\in\R^{|I^*_i|\times|I^*_j|},i,j=1,\dots,G.$ Define the degree matrix~$\hat D_{\op{diag}}$ corresponding to $\hat A_{\op{diag}}$ as 
\begin{equation*}
\hat D_{\op{diag}}:=\op{diag}\bigl((\hat D_{\op{diag}})_1,\dots,(\hat D_{\op{diag}})_n\bigr)\,,    
\end{equation*}
with the elements $(\hat D_{\op{diag}})_i:=\sum_{j=1}^n (\hat A_{\op{diag}})_{ij}.$
Define the {corresponding} graph Laplacian~$\hat L_{\op{diag}}$ as
\begin{equation*}
\hat L_{\op{diag}}:= I- \hat D_{\op{diag}}^{-1/2}\hat A_{\op{diag}}D_{\op{diag}}^{-1/2}\,.
\end{equation*}

The remaining proof proceeds as follows: in step 1, we show that $\hat L_{\op{diag}}$ has non-negative eigenvalues and that the eigenvalue zero has multiplicity $G$. Moreover, the eigenspace corresponding to that eigenvalue is spanned by the vectors $\hat D_{\op{diag}}\1_{I_j^*} \in \R^n$ with entries
\begin{equation} \label{eq:groupvect}
(\1_{I_j^*})_k = \left\{ 
\begin{array}{l}
1, \quad k \in I_j^*
\\
0, \quad k \notin I_j^*
\end{array} \right. ,
\end{equation} 
see Lemma~\ref{eigenspace}. In step 2, we bound the distance in operator norm between $\hat L$ and $\hat L_{\op{diag}}$ (Lemma~\ref{boundL}). In step 3, we quantify the gap between the $G$-th and $(G+1)$-th smallest eigenvalues of $\hat L_{\op{diag}}$ (Lemma~\ref{eigengap}). In step 4, we use the results from step 2 and step 3 to show that the matrix $\hat U$ defined in step 4 of the spectral clustering algorithm is close to a rotation of the matrix~$U{\in\R^{n\times G}}$ defined via 
\[
U := \big (\1_{I_1^*},...,\1_{I_G^*}\big )
\]
in Frobenius norm (Lemma~\ref{boundU}), i.e. the Frobenius norm of the difference between those matrices converges to zero. This convergence together with a simple analysis of the $k$-means algorithm yields our main result in step 5.

\medskip

\noindent \textbf{Step 1:} \textit{Eigenstructure of $\hat L_{\op{diag}}$.} \\
The following result is essentially a reformulation of Proposition 4 from \cite{von2007} in our setting.
The proof follows by exactly the same type of arguments as in the latter paper, for the sake of completeness and for the reader's convenience we provide a short proof in our specific setting.

\begin{lemma}\label{eigenspace}
The multiplicity of the eigenvalue 0 of $\hat L_{\op{diag}}$ equals $G$. The eigenspace of the eigenvalue 0 of $\hat L_{\op{diag}}$ is spanned by the vectors $\hat D_{\op{diag}}^{1/2} \1_{I_j^*}$ where $\1_{I_j^*}$ are defined in~\eqref{eq:groupvect}.
\end{lemma}
\begin{proof}[Proof of Lemma~\ref{eigenspace}]

Begin by observing that $\hat L_{\op{diag}}$ is block-diagonal with $G$ blocks, say $\hat L^{(11)},...,\hat L^{(GG)}$, of size $|I_1^*|\times|I_1^*|,...,|I_G^*|\times|I_G^*|$. 
It thus suffices to show that the eigenvalues of each block are non-negative and that the multiplicity of the eigenvalue $0$ for each block equals $1$. 
Since all blocks share a similar structure we will focus on the first block. Assume that $v=(v_1,\dots,{v_{|I^*_1|}})^\top$ is an eigenvector of $\hat L^{(11)}_{\op{diag}}$ with norm $1$ {corresponding to }eigenvalue $\lambda$. 
Then, we have
\begin{align*}
\lambda = &v^\top \hat L^{(11)}_{\op{diag}} v\\
=&\sum_{i \in I_1} v_i^2 -\sum_{i,j \in I_1} v_i \frac{\hat A_{ij}}{\sqrt{(\hat D_{\op{diag}})_i}\sqrt{(\hat D_{\op{diag}})_j}}v_j
\\
=&\frac{1}{2}\Bigg(\sum_{i \in I_1} v_i^2-2\sum_{i,j \in I_1} \hat A_{ij}\frac{v_i}{\sqrt{ (\hat D_{\op{diag}})_i}}\frac{v_j}{\sqrt{ (\hat D_{\op{diag}})_j}}+\sum_{j\in I_1} v_j^2\Bigg)
\\
=& \frac{1}{2}\sum_{i,j \in I_1} \hat A_{ij}\Bigg(\frac{v_i}{\sqrt{ (\hat D_{\op{diag}})_i}}-\frac{v_j}{\sqrt{ (\hat D_{\op{diag}})_j}}\Bigg)^2
\\
\geq& 0\,,
\end{align*}
where $(\hat D_{\op{diag}})_i$ denotes the $i$-th diagonal elements of $\hat D_{\op{diag}},$ and  the last line follows since by construction $\hat A_{ij} > 0$. The latter also implies that $\lambda = 0$ if and only if $v_i\big/\sqrt{(\hat D_{\op{diag}})_i} = v_j \big/\sqrt{ (\hat D_{\op{diag}})_j }$ for all $i,j$, which is only possible if $v_i = C\sqrt{ (\hat D_{\op{diag}})_i }$ for a constant $C$ independent of $i$. This completes the proof.

\end{proof}

\medskip

\noindent\textbf{Step 2:} \textit{Bound on operator norm distance between $\hat L$ and $\hat L_{\op{diag}}$.} \\

Now we consider the distance between $\hat L$ and $\hat L_{\op{diag}}$ in operator norm.
\begin{lemma}\label{boundL}
On the event $\frac{n A_{1,max}}{A_{0,min}\min_k |I_k^*|} \le 1$ it holds  that
\begin{align*}
\normop{\hat L-\hat L_{\op{diag}}} %
&\le \frac{4 \sqrt{2}n A_{1,max}}{A_{0,min}\min_k |I_k^*|} \sqrt{\frac{A_{0,max}\max_k |I_k^*|}{ A_{0,min}\min_k |I_k^*|}}\,.
\end{align*}
\end{lemma}
\begin{proof}[Proof of Lemma~\ref{boundL}]
The proof follows a similar strategy as in \cite{chung2011spectra}, and Lemma 3.1 of \cite{vandelft2021} but modified to account for the fact that $n$ can diverge while it is fixed in the latter paper. 
Decompose the difference $\hat L-\hat L_{\op{diag}}$ as follows
\begin{align*}
&\hat L-\hat L_{\op{diag}}\\
=& (\hat{D}^{-1/2}-\hat D_{\op{diag}}^{-1/2})\hat{A}\hat{D}^{-1/2}
+\hat D_{\op{diag}}^{-1/2}\hat{A}(\hat{D}^{-1/2}-\hat D_{\op{diag}}^{-1/2})
+\hat D_{\op{diag}}^{-1/2}(\hat{A}-\hat A_{\op{diag}})\hat D_{\op{diag}}^{-1/2}\\
=& (I-\hat D_{\op{diag}}^{-1/2}\hat{D}^{1/2})\hat{D}^{-1/2}\hat{A}\hat{D}^{-1/2} + (\hat D_{\op{diag}}^{-1/2}\hat{D}^{1/2})\hat{D}^{-1/2}\hat{A}\hat{D}^{-1/2}(I-\hat{D}^{1/2}\hat D_{\op{diag}}^{-1/2})\\
&+\hat D_{\op{diag}}^{-1/2}(\hat{A}-\hat A_{\op{diag}})\hat D_{\op{diag}}^{-1/2}\,.
\end{align*}
Now, we bound the terms on the right hand side separately. 
Define the $i$-th diagonal elements of the diagonal matrix $D$ by $(D)_i.$
By definition of the norm~$\normop{\cdot}$, we have
\begin{align*}
\normop{I-\hat D_{\op{diag}}^{-1/2}\hat{D}^{1/2}}
&= \max_i\Biggl|1-\sqrt{\frac{\hat D_i}{(\hat D_{\op{diag}})_i}} \Biggr|
\\
&\le  \max_i \Biggl|1-\frac{\hat D_i}{(\hat D_{\op{diag}})_i} \Biggr|\\
&\le  \frac{\max_i \Bigl|(\hat D_{\op{diag}})_i-\hat{D}_i\Bigr|}{\min_i (\hat D_{\op{diag}})_i}\,,
\end{align*}
where we used the fact that {$|1-x| = |1-\sqrt{x}||1+\sqrt{x}| \geq |1-\sqrt{x}|,\forall x>0$}. 
We also have
\begin{align*}
\normop{\hat D_{\op{diag}}^{-1/2}\hat{D}^{-1/2}}
=\normop{I-(I-\hat D_{\op{diag}}^{-1/2}\hat{D}^{-1/2})}
\le 1+ \frac{\max_i |(\hat D_{\op{diag}})_i-\hat{D}_i|}{\min_i (\hat D_{\op{diag}})_i}\,,
\end{align*}
and
\begin{align*}
\normop{\hat D^{-1/2}\hat A\hat D^{-1/2}}
=&\max_i \Biggl\{\sum_{j=1}^n \frac{\hat A_{ij}}{\sqrt{\hat D_i}\sqrt{\hat D_j}} \Biggr\}\\
\le&  \max_i \Biggl\{\frac{1}{\sqrt{\hat D_i}}\frac{1}{\min_k\sqrt{\hat D_k}}\sum_{j=1}^n \hat A_{ij} \Biggr\}\\
= &\frac{\max_i\sqrt{\hat D_i}}{\min_j\sqrt{\hat D_j}}\,.
\end{align*}
Moreover, by the sub-multiplicativity of the norm~$\normop{\cdot}$, it holds that 
\begin{equation*}
\normop{\hat D_{\op{diag}}^{-1/2}(\hat A-\hat A_{\op{diag}})\hat D_{\op{diag}}^{-1/2}}
\le \frac{1}{\min_i (\hat D_{\op{diag}})_i}\normop{\hat A-\hat A_{\op{diag}}}\,.
\end{equation*}
Collecting pieces gives
\begin{align*}
&\normop{\hat L-\hat L_{\op{diag}}} \\
\le& \frac{\max_i|(\hat D_{\op{diag}})_i-\hat D_i| }{\min_i (\hat D_{\op{diag}})_i} \frac{\max_i\sqrt{\hat D_i}}{\min_j\sqrt{\hat D_j}}
\Biggl(2+\frac{\max_i|(\hat D_{\op{diag}})_i-\hat D_i| }{\min_i (\hat D_{\op{diag}})_i}\Biggr)+\frac{1}{\min_i  (\hat D_{\op{diag}})_i}\normop{\hat A_{\op{diag}}-\hat A }\,.
\end{align*}

Define
\begin{align*}
	S_{0,i} &:= \sum_{j: i,j \mbox{ in same group}} \hat A_{ij}\,,
	\\
	S_{1,i} &:= \sum_{j: i,j \mbox{ in different groups}} \hat A_{ij}	\,.
\end{align*}
With this notation we have 
\[
\big |(\hat D_{\op{diag}})_i-\hat D_i \big| = S_{0,i}
\]
and 
\[
(\hat D_{\op{diag}})_i = S_{1,i}
\]
as well as $\hat D_i = S_{0,i} + S_{1,i}$. Moreover, by definition,
\[
\normop{\hat A_{\op{diag}}-\hat A } = \max_i S_{1,i}.
\]
Collecting pieces yields
\begin{align*}
\normop{\hat L-\hat L_{\op{diag}}} &\le \frac{\max_i S_{1,i}}{\min_i S_{0,i}}\Bigg(1 + \sqrt{\frac{\max_i S_{0,i} + S_{1,i}}{\min_i S_{0,i} + S_{1,i}}}\Big(2 +  \frac{\max_i S_{1,i}}{\min_i S_{0,i}}\Big) \Bigg)
\end{align*}
Recall the definition of $A_{1,max}, A_{0,\min}$. We have
\begin{align*}
S_{1,i}	\le nA_{1,max}\,,
\end{align*}
and
\begin{equation}\label{re:minDdiag}
A_{0,max} \max_k|I_k^*| \geq S_{0,i} \geq \min_k|I_k^*| A_{0,min}\,.
\end{equation}
This further yields
\begin{align*}
	\normop{\hat L-\hat L_{\op{diag}}} \le \frac{n A_{1,max}}{A_{0,min}\min_k |I_k^*|}\Bigg(1 + \sqrt{\frac{\max_k |I_k^*| + n A_{1,max}}{ A_{0,min}\min_k |I_k^*|}}\Big(2 + \frac{n A_{1,max}}{A_{0,min}\min_k |I_k^*|}\Big) \Bigg)\,.
\end{align*}
Assuming $\frac{n A_{1,max}}{A_{0,min}\min_k |I_k^*|} \le 1$ and noting $\frac{\max_k |I_k^*|}{A_{0,min} \min_k |I_k^*|} \ge 1$ this can be further bounded by 
\[
\normop{\hat L-\hat L_{\op{diag}}} \le \frac{4 \sqrt{2}n A_{1,max}}{A_{0,min}\min_k |I_k^*|} \sqrt{\frac{\max_k |I_k^*|}{ A_{0,min}\min_k |I_k^*|}}. 
\]
This completes the proof.
\end{proof}

\noindent \textbf{Step 3:} \textit{Bounding the $G+1$'st smallest eigenvalue of $\hat L_{\op{diag}}.$}\\
Denote the $i$-th smallest eigenvalue of~$\hat L_{\op{diag}}$ by $\lambda_{i}.$
By Lemma~\ref{eigenspace}, we know that $\lambda_1=\cdots=\lambda_G=0.$
Thus, we need to find a lower bound on {the $G+1$'st smallest eigenvalue}~$\lambda_{G+1}$. This is done in the following Lemma.
\begin{lemma}\label{eigengap}
We have
\[
\lambda_{G+1}\geq \frac{A_{0,min}}{8A_{0,max}}\,. %
\] 

\end{lemma}
\begin{proof}[Proof of Lemma~\ref{eigengap}]

Recall the \textit{Cheeger constant} (see for instance equation (2.2) in \cite{chung1997}) of a undirected graph $\graph = (V,E)$ ($V$ denotes the set of vertices and $E$ denotes the set of edges) with weights $w_{i,j}$ on the vertices $(i,j) \in E$:
\[
\mathfrak{H} := \min _{\mathcal{I}{\subset V}} \frac{\sum_{j\in\mathcal{I},k\notin\mathcal{I}}{w_{j,k}}}{ \min\big\{  \sum_{j\in\mathcal{I}}d_j,\sum_{k \in V\backslash \mathcal{I}} d_k\big\} }\,,
\]
where 
\[
d_k := \sum_{(i,j) \in E: i \in \mathcal{I}} w_{i,j}\,.
\]
Then, Theorem 2.2 in~\cite{chung1997} implies that the eigengap of the normalized graph Laplacian is bounded below by $\mathfrak{H}^2/2$. 
To translate this result to our setting consider the fully connected graph with vertices given by $V = I_k^*$ and edge weights $w_{i,j} := \hat A_{ij}, i,j \in V$. 
Hence, the Cheeger constant corresponding to block $\hat L^{(mm)}$ on the diagonal of $\hat L_{\op{diag}}$ is {defined as} %
\[
\mathfrak{H}_m := \min _{\mathcal{I}\subset I_m^*} \frac{\sum_{j\in\mathcal{I},i\in I_m^* \backslash\mathcal{I}}{\hat A_{ij}}}{ \min\bigl\{  \sum_{j\in\mathcal{I}} \hat d_j(m),\sum_{k \in I_m^* \backslash \mathcal{I}} \hat d_k(m)\bigr\} }\,,
\]
where $\hat d_j(m) := \sum_{i \in I_m} \hat A_{ij}$. Since the non-zero eigenvalues of $\hat L_{\op{diag}}$ are exactly the eigenvalues of the corresponding block diagonal pieces, it follows that 
\[
\lambda_{G+1} \geq \frac{\min_{m=1,...,G} \mathfrak{H}_m^2}{2}\,.
\] 
Hence, it suffices to prove that
\begin{equation*}
\min_{m=1,...,G} \mathfrak{H}_m \geq A_{0,min}/2A_{0,max} \,.
\end{equation*}
Observe that
\begin{align*}
\sum_{j\in\mathcal{I}}\hat d_j(m) \leq |\mathcal{I}| |I_m^*| A_{0,max}
\end{align*}
and 
\[
\sum_{j\in\mathcal{I},i\in I_m^* \backslash\mathcal{I}}{\hat A_{ij}} \ge |\mathcal{I}| |I_m^*\backslash \mathcal{I}|A_{0,min} \,.
\]
Let $ \mathcal{\bar I} := I_m^* \backslash \mathcal{I}$. 
It then holds that
\[
\mathfrak{H}_m \geq \frac{A_{0,min}}{A_{0,max}} \min _{\mathcal{I}\subset I_m^*} \frac{|\mathcal{I}||\mathcal{\bar I}|}{ |I_m^*|\min\{  |\mathcal{I}|  ,|\mathcal{\bar I}| \} } = \frac{A_{0,min}}{A_{0,max}} \min_{\mathcal{I}\subset I_m^*} \frac{|\mathcal{I}| \vee |\mathcal{\bar I}|}{|I_m^*|} \geq \frac{A_{0,min}}{2A_{0,max}},
~~m = 1,...,G\,.
\]
This completes the proof.
\end{proof}

\medskip

\noindent \textbf{Step 4:} \textit{Frobenius norm convergence of~\,$\hat U$ to a transformation of \,$U$.} 

\begin{lemma}\label{boundU}
There exists a orthogonal matrix $O_{n,T}{\in\R^{G\times G}}$ such that on the event $\frac{n A_{1,max}}{A_{0,min}\min_k |I_k^*|} \le 1$ we have
\begin{align*}
\normf{\hat U - U O_{n,T}}^2
&\le \frac{2^{16} n^2 G\max_k|I_k^*|^2 A_{1,max}^2 A_{0,max}^3}{ A_{0,min}^5\min_k |I_k^*|^3} \,. 
\end{align*}
\end{lemma}

\begin{proof}[Proof of Lemma~\ref{boundU}]
In the first step we apply Theorem 2 from \cite{yu2015}. 
In the notation of the latter paper let $d=G, s = n, r = n-G+1$, $\hat \Sigma = \hat L, \Sigma = \hat L_{\op{diag}}$. 
Let $\hat Z, Z$ denote the matrices which contain the eigenvectors corresponding to the $G$ smallest eigenvalues of $\hat L_{\op{diag}}$ and $\hat L$, respectively (in the notation of \cite{yu2015} we have $\hat V = \hat Z, V = Z$). 
Note that by Lemma~\ref{eigenspace} we can choose $Z$ to have columns $\hat D_{\op{diag}} \1_{I_j^*}, j=1,...,G$. By equation (3) in Theorem 2 from \cite{yu2015} there exists an orthonormal matrix $\hat O \in \R^{G\times G}$ with 
\begin{equation}\label{re:DKthm}
\normf{\hat{Z} \hat O -Z}\le\frac{2^{3/2}\sqrt{G}\normop{\hat L-\hat L_{\op{diag}}}}{\lambda_{G+1}} \,.
\end{equation}
Here we note that for symmetric matrices the operator norm $\|\cdot\|_{\op{op}}$ used in \cite{yu2015} coincides with our $\normspec{\cdot }$ and the latter satisfies $\normspec{A} \leq \normop{A}$ for symmetric matrices $A$. 
Let ${O_{n,T}} := \hat O^\top$ and note that by orthogonality of $\hat O$ we have $\normf{\hat{Z} \hat O -Z} = \normf{\hat{Z} - Z O_{n,T}}$. 
In what follows write $O$ for $O_{n,T}$ to simplify notation.
Note that $\hat U_{i,\cdot}=\frac{\hat Z_{i,\cdot}}{\normtwo{\hat Z_{i,\cdot}}},$ and $(U O) _{i,\cdot}=\frac{(ZO)_{i,\cdot}}{\normtwo{ Z_{i,\cdot}}}.$ 
Similarly to Lemma 3.2 in \cite{vandelft2021}, it follows that
\begin{align*}
\normf{\hat U-UO}^2&=\sum_{i=1}^n \Bigg\|\frac{\hat Z_{i,\cdot}}{\normtwo{\hat Z_{i,\cdot}}} - \frac{ (ZO)_{i,\cdot}}{\normtwo{Z_{i,\cdot}}}\Bigg\|_2^2 \\
&=    \sum_{i=1}^n \Bigg\| \frac{\hat Z_{i,\cdot} \,\normtwo{Z_{i,\cdot}} - \hat Z_{i,\cdot}\, \normtwo{\hat Z_{i,\cdot}} + \hat Z_{i,\cdot}\,\normtwo{\hat Z_{i,\cdot}} -  (ZO)_{i,\cdot}\,\normtwo{\hat Z_{i,\cdot}} }{\normtwo{\hat Z_{i,\cdot}}\,\normtwo{Z_{i,\cdot}}}  \Bigg\|_2^2 \\
&\le 2\sum_{i=1}^n \Bigg\| \frac{ \hat Z_{i,\cdot} (\,\normtwo{Z_{i,\cdot}}-\normtwo{\hat Z_{i,\cdot}})}{ \normtwo{\hat Z_{i,\cdot}}\,\normtwo{Z_{i,\cdot}}} \Bigg\|_2^2 +\Bigg\| \frac{\hat Z_{i,\cdot} - (ZO)_{i,\cdot} }{\normtwo{Z_{i,\cdot}}} \Bigg\|_2^2 \\
& = 2\sum_{i=1}^n \frac{(\,\normtwo{Z_{i,\cdot}}-\normtwo{\hat Z_{i,\cdot}})^2}{\normtwo{Z_{i,\cdot}}^2} +\frac{\normtwo{\hat Z_{i,\cdot}-(ZO)_{i,\cdot}}^2}{\normtwo{Z_{i,\cdot}}^2}\\
&\le 4 \sum_{i=1}^n \frac{\normtwo{\hat Z_{i,\cdot}-(ZO)_{i,\cdot}}^2}{\normtwo{Z_{i,\cdot}}^2}\\
&\le \frac{4}{\min_i\normtwo{Z_{i,\cdot}}^2}\, \normf{\hat Z-(ZO)}^2\,.
\end{align*}
Combining this with~\eqref{re:DKthm} yields
\begin{equation}\label{re:boundU}
\normf{\hat U - U O}^2
\le \frac{32G}{(\lambda_{G+1})^2\min_i\normtwo{Z_{i,\cdot}}^2 }\, \normop{\hat L-\hat L_{\op{diag}}}^2\,.
\end{equation}
Recalling the definition of $Z$, we obtain
\begin{align*}
\normtwo{ Z_{i,\cdot}}^2 
=& \frac{(\hat D_{\op{diag}})_i}{\sum_{j\in I^*_k}(\hat D_{\op{diag}})_j} = 1/|I_k^*|,~~\forall i\in I^*_k\,,
\end{align*}
where the last line follows since $(\hat D_{\op{diag}})_i$ is the same for all $i$ from the same group. This yields 
\[
1/\min_i \normtwo{ Z_{i,\cdot}}^2 = \max_k |I_k^*|
\]
and thus 
\[
\normf{\hat U - U O}^2
\le \frac{32G\max_k|I_k^*|}{(\lambda_{G+1})^2 }\, \normop{\hat L-\hat L_{\op{diag}}}^2\,.
\]
Combining this with the bounds in Lemma~\ref{boundL} we find
\begin{align*}
\normf{\hat U - U O}^2
& \le \frac{2^{16} n^2 G\max_k|I_k^*|^2 A_{1,max}^2 A_{0,max}^3}{ A_{0,min}^5\min_k |I_k^*|^3}  
\end{align*}

\end{proof}

\noindent\textbf{Step 5:} \textit{Completing the argument}

Recall that the last step of the algorithm consists of applying $k$-means clustering to the $n$ embedded points $\hat U_{1,\cdot},\dots,\hat U_{n,\cdot}.$ In other words, this step determines group centers $\hat c_1,...,\hat c_G$ through
\[
\{\hat{c}_1,\dots,\hat{c}_G\}\in\argmin_{c_1,\dots,c_G\in\R^G}\Biggl\{\sum_{i=1}^n \min_{j\in\{1,\dots,G\}}\normtwo{\hat U_{i,\cdot}-c_j}^2\Biggr\}\,.
\] 
The data points $\hat U_{i,\cdot}$ and $\hat U_{j,\cdot}$ are grouped together if and only if
\[
\argmin_k \|\hat c_k - \hat U_{i,\cdot}\|_2 = \argmin_k \|\hat c_k - \hat U_{j,\cdot}\|_2\,.
\] 
We will prove that as soon as $\normf{\hat U - UO_{n,T}} < 1/2$ all individuals are clustered correctly. Combined with Lemma~\ref{boundU} and noting that under the assumptions of the theorem we have $\frac{n A_{1,max}}{A_{0,min}\min_k |I_k^*|} \le 1$,  this will complete the proof. 
By orthogonality of $O_{n,T}$ and the definition of $U$ we have for $i,j$ in different groups
\[
\|(UO_{n,T})_{i,\cdot} - (UO_{n,T})_{j,\cdot}\|_2 = \|U_{i,\cdot} - U_{j,\cdot} \|_2 = \sqrt{2}\,. 
\] 
Note that by definition of the Frobenius norm
\[
\max_{i \neq j} \Big\{ \| \hat U_{i,\cdot} - (UO_{n,T})_{i,\cdot}\|_2 + \|\hat U_{j,\cdot} - (UO_{n,T})_{j,\cdot} \|_2 \Big\} \leq \sqrt{2} \, \normf{\hat U - UO_{n,T}} < 1/\sqrt{2}\,.
\]
Combining the above inequality with the reverse triangle inequality we have for $i,j$ in different groups
\[
\min_{i,j \text{ in different groups}} \|\hat U_{i,\cdot} - \hat U_{j,\cdot}\|_2 \geq \|U_{i,\cdot} - U_{j,\cdot} \|_2- \sqrt{2}\, \normf{\hat U - UO_{n,T}} > 1/\sqrt{2}\,.  
\]
Similarly, we have
\[
\max_{i,j \text{ in the same group}} \|\hat U_{i,\cdot} - \hat U_{j,\cdot}\|_2  \leq \sqrt{2}\, \normf{\hat U - UO_{n,T}} < 1/\sqrt{2}\,.
\]
Hence any two points in the same group are closer to each other than to any point outside of that group. 
This implies that $\hat c_j$ are just the group means of group $I_j^*$ (modulo permutation of group labels) and that individuals $i,j$ are grouped together if and only if $i,j \in I_k^*$ for some $k$. 
This completes step 5 and thus the proof of Theorem~\ref{th:specnonasy}. 
\hfill $\Box$

\subsubsection{Proof of Theorem~\ref{thm:clusternew} } Note that $G^* \leq n, \min_k |I_k^*| \geq 1, \max_k |I_k^*| \leq n$. Thus the bound in~\eqref{eq:nonasyperfectA} holds with probability approaching one if $A_{1,max}^2A_{0,max}^3/A_{0,min}^5 = o_\mpr(n^{-5})$. Now letting $\hat \eta_{max} := \lambda_{max}(b_T^{-1/2}\hat \Sigma_{i,j}^{-1/2})$ and $\hat \eta_{min} := \lambda_{min}(b_T^{-1/2}\hat \Sigma_{i,j}^{-1/2})$ we find that under Assumption~\ref{asm:est_cov} $\hat \eta_{min}$ is bounded away from zero and $\hat \eta_{max}$ is bounded from above by a fixed constant, both with probability tending to one. Moreover, by definition of $\hat A_{ij}$, that $A_{0,max} \leq 1$ and
\begin{align*}
	A_{1,max} &\leq \exp(-b_T^{1/2}\hat \eta_{min}(\min_{k \neq \ell} \|\vecb_k^* - \vecb_\ell^*\|_2 - 2\max_{i}\|\hat \vecb_i - \vecb_i\|_2))
	\\
	A_{0,min} &\geq \exp(-2 b_T^{1/2}\hat \eta_{max} \max_{i}\|\hat \vecb_i - \vecb_i\|_2).
\end{align*}
Thus
\[
A_{1,max}^2A_{0,max}^3/A_{0,min}^5 \leq \exp\Big(- 2b_T^{1/2} \{\hat \eta_{min}\Delta_{min} - (2\hat \eta_{min} + 5\hat \eta_{max})a_{n,T} \}\Big)\,.
\]
The assumption $a_{n,T} = o_\mpr(\Delta_{min})$ ensures that the exponent is bounded from below by a (positive) constant multiple of $\Delta_{min}b_T^{1/2}$ which grows faster than $\log n$ by assumption. This completes the proof. 
\hfill $\Box$

\subsection{Proofs for examples section}

Throughout this section, we will use the following empirical process notation: let $\mpr_{T,i}$ denote the empirical measure of the sample $(\vecz_{it},Y_{it})_{t=1,\dots,T}$ and let $\mpr_i$ denote the measure corresponding to the distribution of $(\vecz_{i1},\response_{i1})$ and let $\G_{T,i} := \sqrt{T}(\mpr_{T,i} - \mpr_i)$ denote the corresponding empirical process. For a function $f: (\vecz,y) \mapsto f(\vecz,y)$ and a signed measure $\mpr$ let $\mpr f$ stand for $\int f d\mpr$. For a class of functions $\mathcal{G}$ define $\|\G_{T,i}\|_\mathcal{G} := \sup_{g \in \mathcal{G}} |\G_{T,i} f|$. 
Given the function~$f:\rd\times\R\to \R,$ define $\sigma_{q,i}(f):=\Var\Big(\frac{1}{\sqrt{q}}\sum_{i=1}^q f(\vecz_{it},Y_{it})\Big).$

\subsubsection{Proofs for logistic regression in the independent case (Theorem~\ref{mle_est})}

Throughout this section we will use the following additional notation. Let
\[
f(y;\vecg,\vecz)=\exp\Bigl\{y\vecz^\top\vecg-g(\vecz^\top\vecg)\Bigr\}
\]
denote the pmf of $y\in\R$ conditional on $\vecz\in\R^{p+1}$; here the function $g$ is defined via
\begin{align*}
g:\R&\to\R\\
z&\mapsto\log(1+e^z)\,.
\end{align*}
We abbreviate the corresponding log-likelihood as $\ell(\vecz,y;\vecg):=y\vecz^\top\vecg-g(\vecz^\top\vecg).$ Define
\[
\mathbb{M}_{i,T}(\vecg):=\frac{1}{T}\sum_{t}[Y_{it}\vecz_{it}^\top\vecg-g(\vecz_{it}^\top\vecg)]
\] 
and 
\[
\mathbb{M}_{i}(\vecg):=\E_{}[\mathbb{M}_{i,T}(\vecg)], i=1,\dots,n\,.
\]

\begin{lemma}\label{bracketing_integral}
Given $p\in\mathbb{Z}^+,$ we have $\int_0^1 \sqrt{1+\log (\epsilon^{-p}) }d\epsilon\le 1+\sqrt{2\pi p}e^{1/p}.$
\end{lemma}
\begin{proof}[Proof of Lemma~\ref{bracketing_integral}]
Set $t=\sqrt{1+\log(\epsilon^{-p})},$ we then have
\begin{align*}
&\int_0^1 \sqrt{1+\log (\epsilon^{-p}) }d\epsilon\\
\le & \frac{2}{p}e^{\frac{1}{p}}\int_1^\infty t^2e^{\frac{-t^2}{p}}dt\\
=& -e^{\frac{1}{p}}\int_1^\infty td(e^{-\frac{t^2}{p}})\\
=& -e^{\frac{1}{p}}\Bigl(te^{-\frac{t^2}{p}}\Big|_1^\infty-\int_1^\infty e^{-\frac{t^2}{p}}dt\Bigr)\\
=&1 +e^{\frac{1}{p}}\int_1^\infty e^{-\frac{t^2}{p}}dt\\
\le& 1+\sqrt{2\pi p}e^{1/p}\,.
\end{align*}
\end{proof}

\begin{proof}[Proof of Theorem~\ref{mle_est} (i)]
Define set $\Gamma_i(\delta):=\{\vecg\in\R^{p+1}: \normtwo{\vecg-\vecg^*_i}\le \delta\}.$ 
By the concavity of function $\mathbb{M}_{i,T}$ and definition of~$\est_i,$ when all the directional derivatives on the boundary of  the set $\Gamma_i(\delta)$ is negative, that is,
\[
\sup_{\vecg:\normtwo{\vecg-\vecg^*_i}=\delta} (\vecg-\vecg^*_i)^\top \grad \mathbb{M}_{i,T}(\vecg)<0\,,
\]
it follows that $\est_i\in\Gamma_i(\delta).$
This implies 
\[
\mpr\Bigg(\sup_i\sup_{\vecg:\,\normtwo{\vecg-\vecg^*_i}=\tilde C\sqrt{\frac{\log n}{T}}} (\vecg-\vecg^*_i)^\top \grad \mathbb{M}_{i,T}(\vecg)<  0\Bigg)
\le \mpr\Bigg (\sup_i\normtwo{\est_i-\vecg^*_i}\le \tilde C\sqrt{\frac{\log n}{T}} \Bigg)\,.
\]
Hence it suffices to show that under the stated assumptions it holds that
\begin{equation}\label{eq:prlog1step0}
\mpr\Bigg(\sup_i\sup_{\vecg:\,\normtwo{\vecg-\vecg^*_i}=\tilde C\sqrt{\frac{\log n}{T}}} (\vecg-\vecg^*_i)^\top \grad \mathbb{M}_{i,T}(\vecg)< 0\Bigg)\to 1
\end{equation}
provided that $\tilde C$ is picked sufficiently large. Note that 
\begin{align} \label{eq:prlog1h0}
(\vecg-\vecg^*_i)^\top \grad \mathbb{M}_{i,T}(\vecg) %
&= (\vecg-\vecg^*_i)^\top \Big(  \grad \mathbb{M}_{i,T}(\vecg) - \grad \mathbb{M}_{i}(\vecg) \Big) +  (\vecg-\vecg^*_i)^\top \grad \mathbb{M}_{i}(\vecg)\,.
\end{align}
We now handle the last two terms on the right hand side of the last equality separately. More precisely, we will show that for any $\tilde C > 0$ there exists $\delta > 0$ such that for $\log n/T < \delta$ we have
\begin{align}
\label{eq:prlog1step1}
&\sup_{\vecg:\normtwo{\vecg-\vecg^*_i}=\tilde C\sqrt{\frac{\log n}{T}}}(\vecg-\vecg^*_i)^\top \grad \mathbb{M}_{i}(\vecg)\le -C_1 \frac{\log n}{T}, \quad i=1,\dots,n\,,
\end{align}
where $C_1=\tilde C^2 \kappa_2\inf_i\{\lambda_{\min}( \E\bigl[ \vecz_{it}\vecz_{it}^\top \bigr] )\}$ with $\kappa_1 := \max_i \{\normtwo{\vecg^*_i}\}$ and $\kappa_2 := \frac{e^{ \kappa(1+\kappa_1)}}{(1+e^{ \kappa(1+\kappa_1)})^2}$.

Additionally, we will prove that for $\tilde C$ sufficiently large (where ``sufficiently large" does not depend on $n,T$), it holds that
\begin{equation} \label{eq:prlog1step2}
\mpr\Bigg( \sup_i\sup_{\vecg:\normtwo{\vecg-\vecg^*_i}=\tilde C\sqrt{\frac{\log n}{T}}} (\vecg-\vecg^*_i)^\top \Big(  \grad \mathbb{M}_{i,T}(\vecg) - \grad \mathbb{M}_{i}(\vecg) \Big)  > \frac{C_1}{2}\frac{\log n}{T}   \Bigg)\to 0\,.
\end{equation}
Combining the above statements with the decomposition in~\eqref{eq:prlog1h0} yields~\eqref{eq:prlog1step0}.

\bigskip

\noindent\textbf{Proof of display~\eqref{eq:prlog1step1}}.
In what follows assume that the vector~$\vecg\in\R^{p+1}$ satisfies $\normtwo{\vecg-\vecg^*_i}=\tilde C\sqrt{\frac{\log n}{T}}.$
Using Taylor expansion, we have
\[
\grad \mathbb{M}_{i}(\vecg) =  \grad \mathbb{M}_{i}(\vecg^*_i) + \grad^2 \mathbb{M}_{i}(\tilde\vecg) (\vecg -\vecg^*_i) \,,
\]
where $\tilde\vecg\in\R^{p+1}$ is on the line connecting $\vecg$ and $\vecg^*_i.$ Note that $\grad \mathbb{M}_{i}(\vecg^*_i)=\mathbf{0}.$
Multiplying both sides of the equation by $(\vecg-\vecg^*_i)$ gives
\[
(\vecg-\vecg^*_i)^\top  \grad \mathbb{M}_{i}(\vecg) =   (\vecg-\vecg^*_i)^\top  \grad^2 \mathbb{M}_{i}(\tilde\vecg) (\vecg -\vecg^*_i) \,.
\]
Note that
$\grad^2\mathbb{M}_{i}(\tilde\vecg)=-\E\bigl[ g^{''}(\vecz_{it}^\top\tilde\vecg)\vecz_{it}\vecz_{it}^\top \bigr].$
It then follows that 
\begin{align*}
(\vecg-\vecg^*_i)^\top  \grad \mathbb{M}_{i}(\vecg) %
=& - (\vecg-\vecg_i^*)^\top  \E\bigl[ g^{''}(\vecz_{it}^\top\tilde\vecg)\vecz_{it}\vecz_{it}^\top \bigr] (\vecg-\vecg_i^*)\,.
\end{align*}
This implies
\begin{equation}
\label{logistic1}
(\vecg-\vecg^*_i)^\top  \grad \mathbb{M}_{i}(\vecg)\le  -\normtwo{\vecg-\vecg_i^*}^2 \lambda_{\min}\Bigl( \E\bigl[ g^{''}(\vecz_{it}^\top\tilde\vecg)\vecz_{it}\vecz_{it}^\top \bigr] \Bigr)\,.
\end{equation}
For any $\tilde C$ we have $\tilde C\sqrt{\frac{\log n}{T}}<1$ provided that $\log n/T$ is small enough. 
Note that $\tilde\vecg$ is in between $\vecg$ and $\vecg^*_i$ and $\normtwo{\vecg-\vecg^*_i}=\tilde C\sqrt{\frac{\log n}{T}}.$
We then find $\normtwo{\tilde \vecg}\le 1+\kappa_1\,.$
Combining this with Assumption~\ref{asm:logistic_bounded} yields
\[
\normtwo{\vecz_{it}^\top\tilde\vecg}\le \sup_{i,t}\normtwo{\vecz_{it}}\normtwo{\tilde\vecg}\le \kappa(1+\kappa_1)\,.
\]
Note that the function
$z\mapsto g^{''}(z)=\frac{e^{z}}{(1+e^{z})^2}$ is positive and decreasing on $\R_+$.
It then follows that 
\begin{equation*}
g^{''}(\vecz_{it}^\top\tilde\vecg) \ge \frac{e^{ \kappa(1+\kappa_1)}}{(1+e^{ \kappa(1+\kappa_1)})^2}\,.
\end{equation*}
Define $\kappa_2:=\frac{e^{ \kappa(1+\kappa_1)}}{(1+e^{ \kappa(1+\kappa_1)})^2}.$
Plugging the last display into the inequality~\eqref{logistic1} and using $\normtwo{\vecg-\vecg_i^*}=\tilde C\sqrt{\frac{\log n}{T}}$ yields
\[
(\vecg-\vecg^*_i)^\top  \grad \mathbb{M}_{i}(\vecg)\le  -\kappa_2\tilde C^2 \frac{\log n}{T} \lambda_{\min}\Bigl( \E\bigl[ \vecz_{it}\vecz_{it}^\top \bigr] \Bigr)\,.
\]
By Assumption~\ref{asm:logistic_eigenvalue}  that $\inf_{i}\{\lambda_{\min}\bigl( \E\bigl[ \vecz_{it}\vecz_{it}^\top \bigr]\bigr)\}$ is bounded from zero, we obtain~\eqref{eq:prlog1step1} by setting 
$C_1:=\kappa_2\tilde C^2\inf_i\{\lambda_{\min}\bigl( \E\bigl[ \vecz_{it}\vecz_{it}^\top \bigr] \bigr)\}.$

\bigskip

\noindent\textbf{Proof of display~\eqref{eq:prlog1step2}} We will show  
\[
\mpr\Bigg( \sup_i\sup_{\vecg:\normtwo{\vecg-\vecg^*_i}=\tilde C\sqrt{\frac{\log n}{T}}} (\vecg-\vecg^*_i)^\top \Big(  \grad \mathbb{M}_{i,T}(\vecg) - \grad \mathbb{M}_{i}(\vecg) \Big)  > \frac{C_1}{2}\frac{\log n}{T}   \Bigg)\to 0\,.
\]
By Cauchy-Schwaz inequality, it holds that 
\[
(\vecg-\vecg^*_i)^\top \Big(  \grad \mathbb{M}_{i,T}(\vecg) - \grad \mathbb{M}_{i}(\vecg) \Big)\le \normtwo{\vecg-\vecg^*_i}\normtwo{ \grad \mathbb{M}_{i,T}(\vecg) - \grad \mathbb{M}_{i}(\vecg)}\,,
\]
which implies
\begin{align*}
&\mpr\Bigg( \sup_i\sup_{\vecg:\normtwo{\vecg-\vecg^*_i}=\tilde C\sqrt{\frac{\log n}{T}}} (\vecg-\vecg^*_i)^\top \Big(  \grad \mathbb{M}_{i,T}(\vecg) - \grad \mathbb{M}_{i}(\vecg) \Big)  >  \frac{C_1}{2}\frac{\log n}{T}   \Bigg)\\
&\le \mpr\Bigg(  \sup_i\sup_{\vecg:\normtwo{\vecg-\vecg^*_i}=\tilde C\sqrt{\frac{\log n}{T}}} \normtwo{\vecg-\vecg^*_i}\normtwo{ \grad \mathbb{M}_{i,T}(\vecg) - \grad \mathbb{M}_{i}(\vecg)} >  \frac{C_1}{2}\frac{\log n}{T}   \Bigg)
\\
&= \mpr\Bigg(  \sup_i\sup_{\vecg:\normtwo{\vecg-\vecg^*_i}=\tilde C\sqrt{\frac{\log n}{T}}} \normtwo{ \grad \mathbb{M}_{i,T}(\vecg) - \grad \mathbb{M}_{i}(\vecg)} >  \frac{C_1}{2 \tilde C} \sqrt{\frac{\log n}{T}}   \Bigg)
\\
&\le \mpr\Bigg(  \sup_i\sup_{\vecg:\normtwo{\vecg-\vecg^*_i}\leq 1} \normtwo{\vecg-\vecg^*_i}\normtwo{ \grad \mathbb{M}_{i,T}(\vecg) - \grad \mathbb{M}_{i}(\vecg)} >  C_2 \sqrt{\frac{\log n}{T}}   \Bigg)\,,
\end{align*}
where $C_2 := \frac{C_1}{2\tilde C} = \tilde C \kappa_2\inf_i\{\lambda_{\min}( \E\bigl[ \vecz_{it}\vecz_{it}^\top \bigr] )\}/2$. We will show that the last line in the display above converges to zero provided that $\tilde C$ is large enough. Define the vector 
\[
M_{it}(\vecg):=Y_{it}\vecz_{it}-\frac{\vecz_{it} e^{\vecz_{it}^\top \vecg}}{1+e^{\vecz_{it}^\top\vecg}}-\E\Bigg[Y_{it}\vecz_{it}-\frac{\vecz_{it} e^{\vecz_{it}^\top \vecg}}{1+e^{\vecz_{it}^\top\vecg}}\Bigg]\,.
\]
Denote the $j$-th entry of the vector~$M_{it}(\vecg)$ by $M_{it,j}(\vecg).$
We now show that 
\[
\mpr\Bigg( \sup_i\sup_{\vecg:\normtwo{\vecg-\vecg^*_i}\le 1} \frac{1}{T}\sum_{t=1}^T\Big|M_{it,j}(\vecg)\Big|  >C_2\sqrt{\frac{\log n}{T}}  \Bigg)\to 0\,,
\]
Define the function $h^{j}_{\vecg}(\vecz,y)$ via
\begin{align*}
h^{j}_{\vecg}: \R^{p+1} \times\{0,1\}& \to \R \\
(\vecz,y)&\mapsto h^{j}_{\vecg}(\vecz,y):= z_j \Big(y - \frac{e^{\vecz^\top \vecg }}{1+e^{\vecz^\top \vecg} }\Big)\1\{\normtwo{\vecz} \leq \kappa\}\,,
\end{align*}
where $z_j$ denotes the $j$-th element of the vector~$\vecz$.
Consider the function class~$\mathcal{H}_{i,j}(\delta):=\bigl\{ h^{j}_{\vecg}(\vecz,y): \normtwo{\vecg-\vecg_i^*}\le \delta\bigr\}.$ 
Set $\mathcal{H}_{i,j}:=\mathcal{H}_{i,j}(1)$.
It follows that 
\begin{equation*}
\mpr\Bigg( \sup_{\vecg:\,\normtwo{\vecg-\vecg^*}\le 1} \Bigl |\frac{1}{T}\sum_{t=1}^T M_{it,j}(\vecg)\Bigr |>C_2\sqrt{\frac{\log n}{T}}\Bigg)
= \mpr\Bigl(\norm{ \mathbb{G}_{T,i}}_{\mathcal{H}_{i,j}}>C_2\sqrt{\log n}\Bigr) \,.
\end{equation*}
Now, we study the probability
\[
\mpr\Bigl(\norm{ \mathbb{G}_{T,i}}_{\mathcal{H}_{i,j}}>C_2\sqrt{\log n}\Bigr) \,.
\]
Note that for any $\vecg\in\R^{p+1}$ satisfying $\normtwo{\vecg-\vecg_i^*}<1,$ it holds that 
\begin{equation} \label{eq:fbound}
\Biggl|  yz_j- \frac{e^{\vecz^\top \vecg }z_j}{1+e^{\vecz^\top \vecg} } \Biggr| = |z_j| \Biggl|  y- \frac{e^{\vecz^\top \vecg }}{1+e^{\vecz^\top \vecg} } \Biggr|
\le  \|\vecz\|_2\,,
\end{equation}
and thus an envelope for the class $\mathcal{H}_{i,j}$ is given by $\kappa.$
Moreover, for any functions $h_{\vecg}^j(\vecz,y), h_{\tilde\vecg}^j(\vecz,y)\in \mathcal{H}_{i,j},$ it holds that
\begin{align}
&\Bigg|   yz_j- \frac{e^{\vecz^\top \vecg }z_j}{1+e^{\vecz^\top \vecg} } - yz_j + \frac{e^{\vecz^\top \tilde\vecg }z_j}{1+e^{\vecz^\top \tilde\vecg} }\Bigg|\\
\le&\normtwo{\vecz}\Bigg|\frac{e^{\vecz^\top \vecg }z_j}{1+e^{\vecz^\top \vecg} }-\frac{e^{\vecz^\top\tilde \vecg }z_j}{1+e^{\vecz^\top\tilde \vecg} } \Bigg|\\
\le & \normtwo{\vecz}^2\normtwo{\vecg-\tilde\vecg}\,, \label{eq:diffbound}
\end{align}
where the last inequality follows from the mean value theorem and the bound $e^z/(1+e^z)^2 \leq 1$.
Thus, the $\epsilon$-bracketing number of the function class~$\mathcal{H}_{i,j}$ satisfies
\begin{equation}\label{eq:brk}
\sup_{i,j} N_{[\,]}(\epsilon, \mathcal{H}_{i,j} , \|\cdot\|_{2}) \leq C_4\epsilon^{-p-1}
\end{equation}
for a constant $C_4$ independent of $n$. By Theorem 2.14.2 in \cite{van1996} and Assumption~\ref{asm:logistic_bounded}, we have
\begin{align*}
\E \Bigg[\sup_{\vecg: \,\normtwo{\vecg-\vecg_i^*}<1} \Biggl(\sqrt{T}   \Bigl |\frac{1}{T}\sum_{t=1}^T M_{it,j}(\vecg)\Bigr | \Biggr) \Bigg]\lesssim 2\kappa J_{[\,]}(1,  \mathcal{H}_{i,j} )\,,
\end{align*}
where
\begin{equation*}
J_{[\,]}(1, \mathcal{H}_{i,j} ):=\int_0^1 \sqrt{1+\log N_{[\,]}(\epsilon, \mathcal{H}_{i,j} , \normtwo{\cdot}) }d\epsilon \leq \int_0^1 \sqrt{1+\log(C_4 \epsilon^{-p-1})} d\epsilon < \infty
\end{equation*}
by Lemma~\ref{bracketing_integral}. This implies
\begin{equation}\label{eq:Emp_expectation}
\mu := \sup_{i,j}\E \Bigg[\sup_{\vecg: \,\normtwo{\vecg-\vecg_i^*}<1} \Bigl(\sqrt{T}   \Bigl |\frac{1}{T}\sum_{t=1}^T M_{it,j}(\vecg)\Bigr | \Bigr)\Bigg] \leq C_5\kappa
\end{equation}
for a constant $C_5$ independent of $n$. For a function class $\mathcal{H}$ define $\mu_i(\mathcal{H}):=\E \bigl[\,\norm{ \mathbb{G}_{T,i}}_{\mathcal{H}} \bigr],$ and  $\sigma^2_i(\mathcal{H}):=\norm{\mpr_i[(h-\mpr_i h)^2]}_{\mathcal{H}}.$ By~\eqref{eq:Emp_expectation} we have $\mu_i\big(\mathcal{H}_{i,j}  \big) \leq \mu$.
Since the envelope for the class $\mathcal{H}_{i,j}$ is $\kappa,$ it holds by Assumption~\ref{asm:logistic_bounded} that
\begin{align*}
\sigma^2:= \sup_{i,j} \sigma^2_i\bigl(\mathcal{H}_{i,j} \bigr) \leq \kappa^2\,.
\end{align*}
Define $\tilde C_2 := C_2C^*$ where $C^*$ denotes the universal constant $C$ from Theorem 2.14.25 in \cite{van1996}. 
Set $t = C_2\sqrt{\log n}-\mu$ to obtain for $\log n > \mu^2/\tilde C_2^2,$ it holds that
\begin{align*}
\sup_{i,j}\mpr\Bigl(\norm{ \mathbb{G}_{T,i}}_{\mathcal{H}_{i,j}} > \tilde C_2\sqrt{\log n}\Bigr) 
\le & \sup_{i,j} \mpr \Bigl(\norm{ \mathbb{G}_{T,i}}_{ \mathcal{H}_{i,j} }> C^*\{\mu_i\bigl( \mathcal{H}_{i,j} \bigr) +t\} \Bigr)\,.
\end{align*}
Invoking Theorem 2.14.25 in \cite{van1996} yields
\begin{align*}
&\sup_{i,j} \mpr \Bigl(\norm{ \mathbb{G}_{T,i}}_{ \mathcal{H}_{i,j}  }> C^*\{\mu_i\bigl( \mathcal{H}_{i,j} \bigr) +t\} \Bigr)\\
\le & \exp \Biggl(-D\Bigl( \frac{(C_2\sqrt{\log n}-\mu)^2}{\sigma^2} \bigwedge \frac{(C_2\sqrt{\log n}-\mu)\sqrt{T}}{\kappa} \Bigr)\Biggr)\,,
\end{align*}
where $D$ is a universal constant independent of $n,T,C_2$. Collecting pieces gives
\begin{align*} 
&\mpr\Bigg( \sup_i  \sup_{\vecg:\,\normtwo{\vecg-\vecg^*}\le 1} \Bigl |\frac{1}{T}\sum_{t=1}^T M_{it,j}(\vecg)\Bigr |>C_2\sqrt{\frac{\log n}{T}} \Bigg)
\\
\le & \sum_{i=1}^n \mpr\Bigg( \sup_{\vecg:\,\normtwo{\vecg-\vecg^*}\le1}  \bigl |\frac{1}{T}\sum_{t=1}^T M_{it,j}(\vecg) \bigr|>C_2\sqrt{\frac{\log n}{T}}\Bigg)
\\
\le &  \exp \Biggl(\log n -D\Bigl( \frac{(C_2\sqrt{\log n}-\mu)^2}{\sigma^2} \bigwedge \frac{(C_2\sqrt{\log n}-\mu)\sqrt{T}}{\kappa} \Bigr) \Biggr)\,.
\end{align*}
By assumption that $\log n/T \to 0$ and $n, T\to\infty $, we can pick a $\tilde C$ sufficiently large such that $C_2$ is large enough to obtain 
\begin{equation*}
\mpr\Bigg( \sup_i \sup_{\vecg:\,\normtwo{\vecg-\vecg^*}\le1}\Bigl |\frac{1}{T}\sum_{t=1}^T M_{it,j}(\vecg) \Bigr |>C_2\sqrt{\frac{\log n}{T}}\Bigg)\to 0 \,.
\end{equation*}
This completes the proof. \end{proof}

\begin{proof}[Proof of Theorem~\ref{mle_est} (ii)]
Define the functions $h_i: \R^{p+1} \to \R^{(p+1) \times (p+1)}$ through
\[
h_i(\vecg) := \E\Bigg[ \frac{e^{\vecz_{it}^\top\vecg}}{(1+e^{\vecz_{it}^\top\vecg})^2}\vecz_{it}\vecz_{it}^\top \Bigg].
\]
Note that 
\begin{equation*}
\sup_i\normspec{\widetilde{\Sigma}_i^{-1}-h_i(\hat \vecg_i)} \le \kappa\sup_i\lambda_{\max}(\E[\vecz_{it}\vecz_{it}^\top])\sup_{i}\normtwo{\hat\vecg_i-\vecg_i^*}\,.
\end{equation*}
By Theorem~\ref{mle_est} and Assumption~\ref{asm:logistic_eigenvalue}, we obtain
\begin{equation*}
\sup_i\normspec{\widetilde{\Sigma}_i^{-1}-h_i(\hat \vecg_i)} = \bigo_\mpr \Bigg(\sqrt{\frac{\log n}{T}} \Bigg)\,.
\end{equation*}
Since
\begin{equation*}
\sup_i\normspec{\hat{\widetilde{\Sigma}}_i^{-1}- \widetilde{\Sigma}_i^{-1}}\le \sup_i\normspec{\hat{\widetilde{\Sigma}}_i^{-1}- h_i(\hat \vecg_i)}+\sup_i\normspec{h_i(\hat \vecg_i)-\widetilde{\Sigma}_i^{-1}}\,,
\end{equation*}
it remains to show that 
\begin{equation}\label{eq:loghelp1}
\sup_i  \normspec{\hat{\widetilde{\Sigma}}_i^{-1}- h_i(\hat \vecg_i) }=\bigo_\mpr \Bigg(\sqrt{\frac{\log n }{T}} \Bigg).
\end{equation}
For ease of notation, we define the matrix~$N_{it}(\vecg)\in\R^{(p+1)\times(p+1)}$ via
\begin{equation*}
N_{it}(\vecg):= \frac{e^{\vecz_{it}^\top\vecg}}{(1+e^{\vecz_{it}^\top\vecg})^2}\vecz_{it}\vecz_{it}^\top - \E\Bigl[ \frac{e^{\vecz_{it}^\top\vecg}}{(1+e^{\vecz_{it}^\top\vecg})^2} \vecz_{it}\vecz_{it}^\top\Bigr]\,.
\end{equation*}
Define the $(j,\ell)$-th entry of matrix $N_{it}(\vecg)$ by $N_{it,j,\ell}(\vecg).$ 
Given $\delta>0,$ it holds that 
\begin{align*}
&\mpr\Bigg( \sup_i \Big|\frac{1}{T}\sum_{t=1}^T N_{it,j,\ell}(\hat\vecg_i)\Big|>C\sqrt{\frac{\log n}{T}}\Bigg)\\
\le&\mpr\Bigg( \sup_i \sup_{\vecg:\,\normtwo{\vecg-\vecg_i^*}\le\delta} \Big|\frac{1}{T}\sum_{t=1}^T N_{it,j,\ell}(\vecg)\Big|>C\sqrt{\frac{\log n}{T}}\Bigg) + \mpr\Bigl( \sup_i \normtwo{\hat\vecg_i-\vecg_i^*}> \delta\Bigr)\,.
\end{align*}
By  equation~\eqref{eq:log-mle1} of Theorem~\ref{mle_est}, it holds that
\begin{align*}
\mpr\Bigl( \sup_i \normtwo{\hat\vecg_i-\vecg_i^*}> \delta\Bigr)\to 0\,.
\end{align*}
It remains to bound the probability 
\begin{equation*}
\mpr\Bigg( \sup_i \sup_{\vecg:\,\normtwo{\vecg-\vecg_i^*}\le \delta} \Big|\frac{1}{T}\sum_{t=1}^T N_{it,j,\ell}(\vecg)\Big|>C\sqrt{\frac{\log n}{T}}\Bigg)\,.
\end{equation*}
Define the function $h^{j,\ell}_{\vecg}(\vecz)$ via
\begin{align*}
h^{j,\ell}_{\vecg}: \R^{p+1} & \to \R \\
\vecz&\mapsto h^{j,\ell}_{\vecg}(\vecz):= \frac{e^{\vecz^\top \vecg }}{(1+e^{\vecz^\top \vecg} )^2}z_jz_l \1\{\normtwo{\vecz} \leq \kappa\}\,,
\end{align*}
where $z_j$ denotes the $j$-th element of the vector~$\vecz.$
Consider the function class~$\mathcal{H}_i^{j,\ell}(\delta):=\bigl\{ h^{j,\ell}_{\vecg}(\vecz): \normtwo{\vecg-\vecg_i^*}\le \delta\bigr\}.$ It follows that 
\begin{equation*}
\mpr\Bigg( \sup_{\vecg:\,\normtwo{\vecg-\vecg^*}\le \delta} \Bigl |\frac{1}{T}\sum_{t=1}^T N_{it,j,\ell}(\vecg)\Bigr |>C\sqrt{\frac{\log n}{T}}\Bigg)
= \mpr\Bigl(\norm{ \mathbb{G}_{T,i}}_{\mathcal{H}_i^{j,\ell}(\delta)}>C\sqrt{\log n}\Bigr) \,.
\end{equation*}
Moreover, by Assumption~\ref{asm:logistic_bounded}, it holds for any function~$ h^{j,\ell}_{\vecg}(\vecz)\in\mathcal{H}_i^{j,\ell}(\delta)$ that $|h^{j,\ell}_{\vecg}(\vecz)| \leq \kappa^2/4$.
Employing the similar entropy method as in the proof of equation~\eqref{eq:log-mle1} of Theorem~\ref{mle_est}, we obtain 
\begin{equation*}
\E \sup_{\vecg:\,\normtwo{\vecg-\vecg^*}\le \delta} \Biggl (\sqrt{T} \Bigl |\frac{1}{T}\sum_{t=1}^T N_{it,j,\ell}(\vecg)\Bigr | \Biggr) \leq C_6\kappa^2
\end{equation*}
for a constant $C_6$ independent of $n,T$. 
Defining $\mu_i(\mathcal{H}):=\E \bigl[\,\norm{ \mathbb{G}_{T,i}}_{\mathcal{H}} \bigr],$ and  $\sigma^2_i(\mathcal{H}):=\norm{\mpr_i[(h-\mpr_i h)^2]}_{\mathcal{H}}$ we have
\begin{align*}
\mu &:= \sup_i \mu_i\Big(\mathcal{H}_i^{j,\ell}(\delta)\Big)\leq C_6\kappa^2\,,
\\
\sigma^2 &:= \sup_i\sigma_i^2\Big(\mathcal{H}_i^{j,\ell}(\delta)\Big)\le  \kappa^4\,.
\end{align*}
Denote the universal constant $C$ from Theorem 2.14.25 in \cite{van1996} by $C^*$ and set $t = \tilde C\sqrt{\log n}-\mu$ with $\tilde C := C/C^*$ for $\log n > \mu_1^2/\tilde C^2$ to obtain
\begin{align*}
\mpr\Bigl(\norm{ \mathbb{G}_{T,i}}_{\mathcal{H}_i^{j,\ell}(\delta)}>C\sqrt{\log n}\Bigr)
\le & \mpr \Bigl(\norm{ \mathbb{G}_{T,i}}_{\mathcal{H}_i^{j,\ell}(\delta)}>C^*\{\mu_i\bigl(\mathcal{H}_i^{j,\ell}(\delta)\bigr) +t\} \Bigr)\,.
\end{align*}
Invoking Theorem 2.14.25 in \cite{van1996} yields
\begin{align*}
&\mpr \Bigl(\norm{ \mathbb{G}_{T,i}}_{\mathcal{H}_i^{j,\ell}(\delta)}>C^*\{\mu_i\bigl(\mathcal{H}_i^{j,\ell}(\delta)\bigr) +t\} \Bigr)\\
\le & \exp \Biggl(-D\Bigl( \frac{(\tilde C\sqrt{\log n}-\mu)^2}{\sigma^2} \bigwedge \frac{(\tilde C\sqrt{\log n}-\mu)\sqrt{T}}{\kappa^2} \Bigr)\Biggr)\,,
\end{align*}
where $D$ is an universal constant independent of $n,T,C$. Collecting pieces gives
\begin{align*} 
&\mpr\Bigg( \sup_i \sup_{\vecg:\,\normtwo{\vecg-\vecg^*}\le\delta} \Bigl |\frac{1}{T}\sum_{t=1}^T N_{it,j,\ell}(\vecg) \Bigr|>C\sqrt{\frac{\log n}{T}}\Bigg)
\\
\le & \sum_{i=1}^n \mpr\Bigg( \sup_{\vecg:\,\normtwo{\vecg-\vecg^*}\le\delta}  \Bigl |\frac{1}{T}\sum_{t=1}^T N_{it,j,\ell}(\vecg) \Bigr|>C\sqrt{\frac{\log n}{T}}\Bigg)
\\
\le &  \exp \Biggl(\log n -D\Bigl( \frac{(\tilde C\sqrt{\log n}-\mu)^2}{\sigma^2} \bigwedge \frac{(\tilde C\sqrt{\log n}-\mu)\sqrt{T}}{\kappa^2} \Bigr) \Biggr)\,.
\end{align*}
By assumption $\log n/T \to 0$, and hence we can pick $C$ sufficiently large to obtain as $\min(n,T) \to \infty$
\begin{equation*}
\mpr\Bigg( \sup_i \sup_{\vecg:\,\normtwo{\vecg-\vecg^*}\le\delta}\Bigl |\frac{1}{T}\sum_{t=1}^T N_{it,j,\ell}(\vecg) \Bigr |>C\sqrt{\frac{\log n}{T}}\Bigg)\to 0 \,.
\end{equation*}
This establishes~\eqref{eq:loghelp1}. Note that the eigenvalues of $\E\Bigl[ \frac{e^{\vecz_{it}^\top\vecg_i^*}}{(1+e^{\vecz_{it}^\top\vecg_i^*})^2}\vecz_{it}\vecz_{it}^\top  \Bigr]$ are bounded uniformly away from zero and from above -- indeed, boundedness from above follows since $\normtwo{\vecz_{it}} \leq \kappa$ by assumption, for boundedness from below recall that $z \mapsto e^z/(1+e^z)^2$ is decreasing and non-negative on $\R_+$ so that for $\kappa_1 := \max_i \{\|\vecg_i^*\|_2\} < \infty$, it holds that
\[
\frac{e^{\vecz_{it}^\top\vecg_i^*}}{(1+e^{\vecz_{it}^\top\vecg_i^*})^2} \geq \frac{e^{\kappa\kappa_1}}{(1+e^{\kappa\kappa_1})^2} > 0\,.
\] 
A Taylor expansion of the map $A \mapsto A^{-1}$ together completes the proof. \end{proof}

\newpage

\subsubsection{Proofs for logistic regression under dependence:  Theorem~\ref{mle_est1}}

\begin{proof}[Proof of Theorem~\ref{mle_est1} (i)]
The proof of Theorem~\ref{mle_est1} (i) is similar to the proof of Theorem~\ref{mle_est} (i), the only difference is that we employ Proposition C.2 in~\cite{kato2012asymptotics} instead of Talagrand’s inequality for i.i.d. random variables used in the previous proof.
We use the same notation as in the proof of Theorem~\ref{mle_est} (i).
To establish the desired result, we need to to derive the bounds~\eqref{eq:prlog1step1} and~\eqref{eq:prlog1step2}.
Note that the proof of display~\eqref{eq:prlog1step1} remains unchanged under the dependent setting, so we omit the proof for the sake of brevity.
We aim to show the bound~\eqref{eq:prlog1step2}, i.e. 
\begin{equation} \label{eq:prlog1step2_mix}
\sup_i\sup_{\vecg:\normtwo{\vecg-\vecg^*_i}=\tilde C\sqrt{\frac{\log n}{T}}} (\vecg-\vecg^*_i)^\top \Big(  \grad \mathbb{M}_{i,T}(\vecg) - \grad \mathbb{M}_{i}(\vecg) \Big) =\bigo_{\mpr}\Big( \frac{\log n}{T}  \Big) \,.
\end{equation}
By the proof of Theorem~\ref{mle_est} (i), it suffices to show that
\[
\mpr\Bigg( \sup_i\sup_{\vecg:\normtwo{\vecg-\vecg^*_i}\le 1} \frac{1}{T}\sum_{t=1}^T\Big|M_{it,j}(\vecg)\Big|  >C_2\sqrt{\frac{\log n}{T}}  \Bigg)\to 0\,,
\] 
where $C_2$ is the constant defined in the proof of Theorem~\ref{mle_est} (i), 
\[
M_{it}(\vecg)=Y_{it}\vecz_{it}-\frac{\vecz_{it} e^{\vecz_{it}^\top \vecg}}{1+e^{\vecz_{it}^\top\vecg}}-\E\Bigg[Y_{it}\vecz_{it}-\frac{\vecz_{it} e^{\vecz_{it}^\top \vecg}}{1+e^{\vecz_{it}^\top\vecg}}\Bigg]\,.
\]
and $M_{it,j}(\vecg)$ denotes the $j$-th entry of the vector~$M_{it}(\vecg).$ 
Define the function class
\begin{align*}
\mathcal{H}_{i,j}:=\Bigg\{  (\vecz,y)\mapsto & \Bigg[\Big( y\vecz-\frac{\vecz e^{\vecz^\top \vecg}}{1+e^{\vecz^\top\vecg}}\Big)_{j}-\E\Big[  \Big( y\vecz-\frac{\vecz e^{\vecz^\top \vecg}}{1+e^{\vecz^\top\vecg}}\Big)_{j} \Big] \Bigg] \1\{\normtwo{\vecz}\le\kappa\}: \\
&y\in\{0,1\}, \vecz\in\R^{p+1},\vecg\in\R^{p+1},\normtwo{\vecg-\vecg_i^*}\le 1\Bigg\}\,.
\end{align*}
It then follows that
\begin{equation*}
\mpr\Bigg( \sup_{\vecg:\,\normtwo{\vecg-\vecg^*}\le 1} \Bigl |\frac{1}{T}\sum_{t=1}^T M_{it,j}(\vecg)\Bigr |>C_2\sqrt{\frac{\log n}{T}}\Bigg)
= \mpr\Bigg(\norm{    \mpr_{T,i}-\mpr_i }_{\mathcal{H}_{i,j}}>C_2\sqrt{\frac{\log n}{T}} \Bigg) \,.
\end{equation*}
We will show that 
\[
\norm{    \mpr_{T,i}-\mpr_i }_{\mathcal{H}_{i,j}} =\bigo_{\mpr} \Bigg(\sqrt{\frac{\log n}{T}} \Bigg) \,.
\]
By displays~\eqref{eq:fbound} and Assumption~\ref{asm:logistic_bounded}, it holds for any $i,j$ and $h\in \mathcal{H}_{i,j}$ that
\[
\normsup{h}\le U_1\,,\text{ and }\Var(h)\le U_2\,,
\]
with some universal constants~$U_1,U_2>0.$
Applying Lemma 4 in~\cite{galvao2018} to the function  $h/U_2$ gives
\[
\sup_{i,j}\sup_{h\in \mathcal{H}_{i,j}}\sup_{1\le q\le T} \Var\Big(\frac{1}{q^{1/2}}\sum_{t=1}^q h(\vecz_{it},Y_{it}) \Big)\le U_3\,,
\] 
with some positive universal constant~$U_3<\infty$.
Note that the envelope for function class~$\mathcal{H}_{i,j}$ is $2\kappa$ and the upper bound for the $\epsilon$-bracketing number in~\eqref{eq:brk} holds for any $L_p$-norm and any probability measure $Q.$
Then, we obtain the following bounds of the $\epsilon$-covering number for any probability measure $Q$ and any $0<\epsilon<1$ that
\[
N(\mathcal{H}_{i,j}, L_1(Q),\epsilon)
\le N_{[~]}(\mathcal{H}_{i,j}, L_1(Q),\epsilon/2)
\le (2A/\epsilon)^\nu\,,
\]
with some constants $A,\nu<\infty.$ 
By Proposition C.2 of~\cite{kato2012asymptotics}, it holds for any $q_{n,T}\ge 1$ satisfying $q^2_{n,T}\log (q_{n,T})=o(T)$, any $i$, and any $s_{n,T}>0$ that 
\begin{equation}\label{eq:mix0}
\mpr\Bigg(  \norm{\mpr_{i,T}-\mpr_i}_{\mathcal{H}_{i,j} }  \ge C \Big( \sqrt{\frac{\log (q_{n,T})}{T} }+ \sqrt{\frac{s_{n,T}}{T}} + \frac{s_{n,T}q_{n,T}}{T}\Big) \Bigg)
\le 2e^{-s_{n,T}}+ 2T\beta(q_{n,T})\,,
\end{equation}
where $C>0$ is a constant independent of $T,n,i,j.$
Let $q_{n,T}:=C_1(\log n+\log T)$ with the constant $C_1>1$ satisfying $b_{\beta}^{C_1}\le e^{-2}.$ 
With the assumption that $T$ grows at most polynomially in $n$ and $(\log n)^3 = o(T)$, one can verify that $q^2_{n,T}\log (q_{n,T})=o(T).$
Let $s_{n,T}:=2\log n,$ it then holds for large $n,T$ that
\[
\sqrt{\frac{\log (q_{n,T})}{T} }+ \sqrt{\frac{s_{n,T}}{T}} + \frac{s_{n,T}q_{n,T}}{T} 
\lesssim \sqrt{\frac{\log n}{T}} 
\]
and
\[
2e^{-s_{n,T}}+2T\beta(q_{n,T}) \lesssim \frac{1}{n^2}+\frac{1}{n^2T}\,.
\]
Taking the union bound for~\eqref{eq:mix0} over $i=1,\dots,n$ gives the desired result.

\end{proof}

\begin{proof}[Proof of Theorem~\ref{mle_est1} (ii)]
The proof of $\sup_{i}\normspec{\widehat B_{iT}^{-1}- B_i^{-1}}=o_{\mpr}(1)$ is similar to the proof of Theorem~\ref{mle_est} (ii), which boils down to show that
\begin{equation*}
\mpr\Bigg( \sup_i \sup_{\vecg:\,\normtwo{\vecg-\vecg_i^*}\le \delta} \Big|\frac{1}{T}\sum_{t=1}^T N_{it,j,\ell}(\vecg)\Big|>C\sqrt{\frac{\log n}{T}}\Bigg)\to 0\,,
\end{equation*}
where 
\begin{equation*}
N_{it}(\vecg):= \frac{e^{\vecz_{it}^\top\vecg}}{(1+e^{\vecz_{it}^\top\vecg})^2}\vecz_{it}\vecz_{it}^\top - \E\Bigg[ \frac{e^{\vecz_{it}^\top\vecg}}{(1+e^{\vecz_{it}^\top\vecg})^2} \vecz_{it}\vecz_{it}^\top\Bigg]\,,
\end{equation*}
and $N_{it,j,\ell}(\vecg)$ denotes the $(j,\ell)$-th entry of the matrix~$N_{it}(\vecg).$
The desired result follows by an application of the union bound and Proposition C.2 of~\cite{kato2012asymptotics} with similar arguments as in the proof of Theorem~\ref{mle_est1} (i) after noting that by Lemma 4 in~\cite{galvao2018}, we have
\[
\sup_{j,\ell}\sup_{i} \Var  \Bigg(\frac{1}{\sqrt{q}}\sum_{t=1}^q  N_{it,j,\ell}(\vecg)  \Bigg)=\bigo(1)\,.
\]
Thus, it remains to show that $\sup_{i}\normspec{\widehat H_{iT}-H_i}=o_{\mpr}(1).$
Similar to the proof of the convergence of $\widehat B_{iT}^{-1}$, one can verify that $\sup_{i}\normspec{  \frac{1}{T}\sum_{t=1}^T  \widehat \vecw_{it}  \widehat\vecw_{it}^\top -  \E\big[\vecw_{i1}\vecw_{i1}^\top \big] }=o_{\mpr}(1).$
We now aim to show that 
\[
\sup_{i}\normspec{  \sum_{1\le j\le m_T} \Big(1-\frac{j}{T}\Big) \Bigg(\frac{1}{T}\sum_{t=1}^{T-j}\big( \widehat \vecw_{it} \widehat \vecw_{i,t+j}^\top + \widehat \vecw_{i,t+j} \widehat \vecw_{it}^\top \big) \Bigg)
-\sum_{j=1}^{\infty}\E \big[\vecw_{i1}\vecw_{i,1+j}^\top+\vecw_{i,1+j}\vecw_{i1}^\top \big] }=o_{\mpr}(1)\,.
\]
To this end, we introduce an intermediate term 
\[
\widetilde A_{iT}:=  \sum_{1\le j\le m_T} \Big(1-\frac{j}{T}\Big) \Bigg(\frac{1}{T}\sum_{t=1}^{T-j}\big(  \vecw_{it}  \vecw_{i,t+j}^\top +  \vecw_{i,t+j}  \vecw_{it}^\top \big) \Bigg)\,,
\]
and we define
\[
\widehat A_{iT}:=\sum_{1\le j\le m_T} \Big(1-\frac{j}{T}\Big) \Bigg(\frac{1}{T}\sum_{t=1}^{T-j}\big( \widehat \vecw_{it} \widehat \vecw_{i,t+j}^\top + \widehat \vecw_{i,t+j} \widehat \vecw_{it}^\top \big) \Bigg) ,
\quad 
A_i:=\sum_{j=1}^{\infty}\E \big[\vecw_{i1}\vecw_{i,1+j}^\top+\vecw_{i,1+j}\vecw_{i1}^\top \big]\,.
\]
So, we aim to show that $\sup_i\normspec{\widehat A_{iT} - A_i}=o_{\mpr}(1).$ Consider the decomposition
\[
\sup_i \normspec{\widehat A_{iT} - A_i}
\le \sup_i  \normspec{\E[\widetilde A_{iT}]-A_i}+
\sup_i  \normspec{\E[\widetilde A_{iT}]-\widehat A_{iT}}\,.
\]
We note that $\sup_i\normspec{\E[\widetilde A_{iT}]-A_i}=o_{\mpr}(1)$ follows by similar arguments as the proof of the last display in the proof of Lemma 12 in~\cite{galvao2018}.

It remains to show that $\sup_i\normspec{\E[\widetilde A_{iT}]-\widehat A_{iT}}=o_{\mpr}(1).$
Invoking the triangle inequality again, we have
\[
\normspec{\E[\widetilde A_{iT}]-\widehat A_{iT}}
\le
\normspec{\widetilde A_{iT}-\E[\widetilde A_{iT}]}
+\normspec{\widetilde A_{iT}-\widehat A_{iT}}\,.
\]
To bound $\sup_i \normspec{\widetilde A_{iT}-\E[\widetilde A_{iT}]}$, observe that
\[
\normspec{\widetilde A_{iT}-\E[\widetilde A_{iT}]} \leq m_T \max_{j=1,\dots,m_T} \normspec{\frac{1}{T}\sum_{t=1}^{T-j} \vecw_{it}\vecw_{i,t+j}^\top+\vecw_{i,t+j}\vecw_{it}^\top - \E \big[\vecw_{i1}\vecw_{i,1+j}^\top+\vecw_{i,1+j}\vecw_{i1}^\top \big]}.
\]
By similar computations as in the proof of of display~(53) in~\cite{galvao2018} one can show that
\[
\sup_{j=1,\dots m_T}\sup_{i=1,\dots,n} \sup_{q \geq 1} \sup_{k,\ell} \Var\Big(\frac{1}{\sqrt{q}}\sum_{t=1}^q (\vecw_{it}\vecw_{i,t+j}^\top+\vecw_{i,t+j}\vecw_{it}^\top)_{k,\ell}\Big) = O(m_T)
\]
By applying Corollary~C.1 in \cite{kato2012asymptotics} with $q = C\log(nm_T), s = C\log(nm_T)$ for a suitable constant $C$ it follows that
\[
\sup_i \normspec{\widetilde A_{iT}-\E[\widetilde A_{iT}]}=\bigo_{\mpr}\Big( m_T\sqrt{\frac{m_T\log (n m_T)}{T}}\Big)
\]
Finally, note that for all $t=1,\dots,T$, by a Taylor expansion and Assumption~\ref{asm:logistic_bounded} 
\[
\|\hat w_{it} - w_{it}\|_2 \leq \kappa^2\|\hat \vecg_i - \vecg_i\|_2. 
\]
Thus by elementary computations
\begin{align*}
\max_i \normspec{\widetilde A_{iT}-\widehat A_{iT}} \lesssim m_T \max_i \|\hat \vecg_i - \vecg_i\|_2 = \bigo_{\mpr}\Big( m_T\sqrt{\frac{\log n}{T}}\Big) = o_{\mpr}(1).
\end{align*}
Combining all bounds obtained so far we have
\[
\sup_i\normspec{\widehat A_{iT} - A_i} = \bigo_{\mpr}\Big( m_T\sqrt{\frac{\log n}{T}} + m_T\sqrt{\frac{m_T\log (nm_T)}{T}} \Big) + o_{\mpr}(1) = o_{\mpr}(1)
\] 
by the assumptions on $m_T$. \end{proof}

\newpage

\subsection{Proofs for quantile regression in the independent case (Theorem~\ref{quantreg_est} and Theorem~\ref{qrint_est})}

\begin{proof}[Proof of Theorem~\ref{quantreg_est}(i)]

\noindent 
Define $\vecg_{n,T,i}:= \est_i-\vecg^*_i.$
The Theorem 5.1 in \cite{chao2017} can be used in our framework by setting $n=T, m=p+1, \xi_m= \kappa, g_n=0,$ and $c_n=0.$
We then find
\begin{equation}\label{re:error}
\vecg_{n,T,i}=-\frac{1}{T}B_i^{-1}\sum_{t=1}^T \psi_{i,\tau}(\vecz_{it},\response_{it})+\vecg_{n,T,i,1}+\vecg_{n,T,i,2}+\vecg_{n,T,i,3}\,,
\end{equation}
where $B_i := \E[f_{\response\mid\vecz}(q_{i,\tau}(\vecz_{i1})\mid \vecz_{i1})\vecz_{i1}\vecz_{i1}^\top]$ and $\psi_{i, \tau} (\vecz, \response) := \vecz (\1(\response \le q_{i,\tau}(\vecz))-\tau).$
Define $\tilde\vecg_{n,T,i}:=\vecg_{n,T,i,1}+\vecg_{n,T,i,2}+\vecg_{n,T,i,3}.$ 
We now prove that 
\begin{align}\label{eq:prqrh1}
\sup_{i }\normtwo{\tilde\vecg_{n,T,i}}=o_\mpr\Bigg(\sqrt{\frac{\log n}{T}}\Bigg)\,.
\end{align}
For this, we show that 
\begin{align*}
\sup_{i}\normtwo{\vecg_{n,T,i,k}}=o_\mpr\Bigg(\sqrt{\frac{\log n}{T}}\Bigg), ~~k=1,2,3\,.
\end{align*}
Now, we handle the three remainder terms~$\vecg_{n,T,i,1}, \vecg_{n,T,i,2}, \vecg_{n,T,i,3}$ separately.
By equation (5.1) in Theorem 5.1 of \cite{chao2017}, we have almost surely
\begin{equation*}
\sup_i \normtwo{\vecg_{n,T,i,1}} \leq C/T
\end{equation*}
for a constant $C$ independent of $n,T,i$. Since $1/T = o(\sqrt{\log n/T})$ it follows that 
\begin{equation}\label{eq:term1}
\sup_{i}\normtwo{\vecg_{n,T,i,1}}=o_\mpr\Bigg(\sqrt{\frac{\log n}{T}}\Bigg)\,.
\end{equation}
By equation (5.2) in Theorem 5.1 of \cite{chao2017} applied with $\kappa_n = 2\log n \ll T$, there exists a constant $C_1$ independent of $n,T$ (and bounded uniformly in $i$ as seen by a close inspection of the corresponding proof in \cite{chao2017}) such that for all sufficiently large $T$
\begin{equation}\label{re:rem2}
\mpr\Bigg(\,\normtwo{\vecg_{n,T,i,2}}>C_1\Big(\sqrt{\frac{\log T}{T}}+\sqrt{\frac{2\log n}{T}}\Big)^2 \Bigg) \leq 2\exp(-\kappa_n) = 2/n^2\,.
\end{equation}
Since 
\[
\Bigg(\sqrt{\frac{\log T}{T}}+\sqrt{\frac{2 \log n}{T}}\Bigg)^2 \leq 2\frac{2\log n + \log T}{T} = o\Bigg(\sqrt{\frac{\log n}{T}} \Bigg)\,,
\]
an application of the union bound shows that 
\[
\sup_{i }\normtwo{\vecg_{n,T,i,2}} = o_\mpr\Bigg( \sqrt{\frac{\log n}{T}}\Bigg)\,.
\]
Next apply (5.2) in Theorem 5.1 from \cite{chao2017} with $\kappa_n = 2\log n \ll T$ to obtain the existence of a constant $C_2$ independent of $T$ (and bounded uniformly in $i$ as seen by a close inspection of the corresponding proof in \cite{chao2017}) such that for all sufficiently large $T$
\begin{equation}\label{re:rem3}
\mpr\Biggl(\,\normtwo{\vecg_{n,T,i,3}}>C_2\Bigl(\sqrt{\frac{\log T}{T}}+\sqrt{\frac{2\log n}{T}}\Bigr)^{3/2}\Biggr)< 2/n^2\,.
\end{equation}
Note that
\begin{align*}
\Bigg(\sqrt{\frac{\log T}{T}}+\sqrt{\frac{2 \log n}{T}}\Bigg)^{3}
\leq 8\frac{ (2\log n)^{3/2} + (\log T)^{3/2}}{T^{3/2}} = o\Bigg(\frac{\log n}{T}\Bigg)
\end{align*}
by the assumption that $\log n = o(T)$. Combining this with the union bound and~\eqref{re:rem3} shows that 
\[
\sup_{i }\normtwo{\vecg_{n,T,i,3}} = o_\mpr\Big( \sqrt{\frac{\log n}{T}}\Big)\,,
\]
and collecting pieces yields~\eqref{eq:prqrh1}.

To complete the proof, define the classes of functions
\[
\mathcal{G}_i := \Big\{(\vecz,y) \mapsto \mathbf{a}^\top\vecz (\1\{y\le \vecz^\top \mathbf{b}\}-\tau) \1\{\|\vecz\|_2\le \kappa \}:\mathbf{b}\in\R^{p+1}, \mathbf{a}\in\R^{p+1},\normtwo{\mathbf{a}}=1 \Big\}
\] 
and note that 
\begin{equation}\label{eq:decomp}
\sup_i \normtwo{\frac{1}{T}B_i^{-1}\sum_{t=1}^T \psi_{i,\tau}(\vecz_{it},\response_{it})} \leq \sup_i \normspec{B_i^{-1}} \sup_i \|\mpr_{T,i} - \mpr_i\|_{\mathcal{G}_i}
\end{equation}
for $\mpr_i$ denoting the measure of and $\mpr_{T,i}$ corresponding to the empirical measure of $\{(\vecz_{it},Y_{it}), t=1,\dots,T\}$. Under the assumptions made we have $\sup_i \normspec{B_i^{-1}} = O(1)$, and Lemma C.3 from~\cite{chao2017} applied with $\kappa_n = 2\log n \ll T$ shows that there exists a constant $C_3$, independent of $n,T$ (and bounded uniformly in $i$ as revealed by a close look at the corresponding proof) such that 
\[
\mpr\Bigg( \norm{\mpr_{T,i} - \mpr_i}_{\mathcal{G}_i} > C_3 \sqrt{\frac{\log n}{T}} \Bigg) \leq n^{-2}\,.
\] 
Applying the union bound shows that
\[
\sup_i \norm{\mpr_{T,i} - \mpr_i}_{\mathcal{G}_i} = \bigo_\mpr\Bigg(\sqrt{\frac{\log n}{T}} \Bigg)\,.
\] 
Combining this with~\eqref{eq:prqrh1} completes the proof.
\end{proof}

\bigskip

Next we proceed to the proof of Theorem~\ref{quantreg_est}(ii). The proof will make use of the following additional notation
\begin{align*}
\psi_{i, \tau} (\vecz, \response) &= \vecz (\1(\response \le q_{i,\tau}(\vecz))-\tau)\\
f_{it} &:= \frac{2d_T}{q_{i,\tau+d_T}(\vecz_{it})-q_{i,\tau-d_T}(\vecz_{it})}\\
\quad e_{it} &:= 1/f_{it}\\
B_{iT} &= \frac{1}{T}\sum_{t=1}^{T} f_{it} \vecz_{it} \vecz_{it}^\top\\
\widetilde\Sigma_{iT}^{-1} &= \E[B_{iT}] H_i^{-1}\E[B_{iT}]\,.
\end{align*}
We begin by stating and proving an intermediate technical result.

\bigskip

\begin{lemma}\label{lemma2} 
Let Assumptions \ref{A1}-\ref{A4} hold and assume $\log n = o(T), \frac{\log n}{Td_T} =o(1)$. Let $\widehat{e}_{it} := \widehat{f}_{it}^{-1}$, then $\sup_{i,t}\left |\widehat{e}_{it} - e_{it}\right | = \bigo_\mpr(b_{n,T})$ with $b_{n,T} =\sqrt{\frac{\log n}{T{d_T^2}}}.$ 
\end{lemma}

\begin{proof}[Proof of Lemma~\ref{lemma2}]
The proof essentially follows from the arguments in the proof of Lemma 9 of \cite{galvao2018}, but modifications are needed to take into account that $n (\log T)^2/T = o(1)$ made in that paper is replaced by $\log n = o(T)$ and that the rate changes accordingly. 
\noindent
By definitions of $\widehat{e}_{is}$ and $e_{is} $, it holds that
\begin{equation*}
\widehat{e}_{is} - e_{is} = \vecz_{is}^\top \Big(\big(\est_{i}(\tau + d_T) - \vecg^*_i(\tau + d_T) \big)- \big(\est_{i}(\tau - d_T) - \vecg^*_i(\tau - d_T) \big)\Big)/2d_T\,.
\end{equation*}
We know from the display~\eqref{re:error} and Theorem~\ref{quantreg_est} that 
\begin{multline*}
\vecz_{is}^\top (\est_{i}(\tau \pm d_T) - \vecg^*_i(\tau \pm d_T))
\\ = -\frac{1}{T} \vecz_{is}^\top B_i^{-1}\sum_{t=1}^{T}  \vecz_{it} \Big(\1 \{\response_{it} \leq q_{i, \tau \pm d_T}(\vecz_{it})\} - (\tau \pm d_T)\Big) + \bigo_\mpr\Bigg(\sqrt{\frac{\log n}{T}}\Bigg)\,.
\end{multline*}
Hence with $U_{it} := F_{\response|\vecz}(\response_{it}|\vecz_{it}) \sim U[0,1]$ independent of $\vecz_{it}$, it holds that
\begin{equation}\label{eq:eit1}
\widehat{e}_{is} - e_{is} = -\frac{1}{2Td_T} \vecz_{is}^\top B_i^{-1}\sum_{t=1}^{T} \vecz_{it} \Big( \1 \{U_{it} \leq \tau + d_T\} - \1 \{U_{it} \leq \tau - d_T\} - 2d_T\Big)+ \bigo_\mpr\Bigg({\frac{1}{d_T}}\sqrt{\frac{\log n}{T}}\Bigg)\,.
\end{equation}
Define the vectors~$M_{it}\in\R^{p+1}$ via
\begin{equation*}
M_{it} :=  \vecz_{it} \Big( \1 \{U_{it} \leq \tau + d_T\} - \1 \{U_{it} \leq \tau - d_T\} - 2d_T\Big) \Big/2d_T\,.
\end{equation*}
Fix an arbitrary $k \in \{1,...,p+1\}$ and let $M_{it,k}$ denote the $k$-th entry of the vector $M_{it}$. 
It then follows that $\E[M_{it,k}] = 0$ and $\sup_i \Var[M_{it,k}] \leq \frac{{C}_1}{d_T}$ for some constant ${C}_1$  under Assumption~\ref{A1}. 
Under Assumption~\ref{A1}, we also have $\sup_{i,t,k} |M_{it,k}| \leq {C}_2/d_T$ for some constant ${C}_2>0$. 
Invoking the Bernstein inequality yields 
\begin{align*}
\mpr\Bigg(\Big|\sum_{t=1}^{T} M_{it,k}\Big| > T\epsilon\Bigg) 
\leq& 2\exp \Bigg(-\frac{\frac{1}{2}T^2\epsilon^2}{\sum_{t=1}^{T} \E[M_{it,k}^2] + \frac{1}{3}{C}_2 d_T^{-1} T\epsilon}\Bigg)
\\
= &2\exp\Bigg(-\frac{\frac{1}{2}T^2\epsilon^2}{{C}_1Td_T^{-1}+ \frac{1}{3}{C}_2 d_T^{-1} T\epsilon}\Bigg)\,.
\end{align*}
Take $\epsilon = {C}_3 T^{-1/2}d_T^{-1/2}(\log n)^{1/2}$ for a constant ${C}_3$ which will be determined later. 
Under the assumption $\frac{\log n}{Td_T} \to 0$, it follows that $\epsilon\to 0$ and the right hand side of the inequality becomes
\begin{equation*}
2\exp\Bigg(-\frac{1}{2}\frac{({C}_3)^2d_T^{-1}\log n}{{C}_1d_T^{-1} + \frac{1}{3}{C}_2{C}_3 d_T^{-1}T^{-1/2}d_T^{-1/2}(\log n)^{1/2}}\Bigg)\leq 2\exp\Big(-\frac{1}{4} ({C}_3)^2 \log n/{C}_1\Big)\,,
\end{equation*}
where the last inequality holds for $\log n/(T d_T)$ sufficiently small.
Then, we have
\begin{align*}
\mpr\Bigg(\sup_k\sup_i \Big|\frac{1}{T}\sum_{t=1}^{T} M_{it,k}\Big|>\epsilon \Bigg) 
&\leq \sum_k \sum_{i} \mpr\Bigg( \Big|\frac{1}{T}\sum_{t=1}^{T} M_{it,k}\Big|>\epsilon \Bigg)
\\
&\leq 2np\exp\Bigg( -\frac{({C}_3)^2}{4{C}_1}  \log n\Bigg) \to 0
\end{align*}
by taking $({C}_3)^2 > 4{C}_1$. 
Hence, we obtain
\begin{equation*}
\sup_i \normtwo{\frac{1}{T}\sum_{t=1}^{T} M_{it}}= \bigo_\mpr\Biggl( \sqrt{\frac{\log n}{Td_T}}\Biggr)\,.
\end{equation*}
Combining this with~\eqref{eq:eit1}, the fact that $\sup_i \normspec{B_i^{-1}} = \bigo(1)$, and Assumption~\ref{A4} gives
\begin{equation*}
\sup_{i,s}|\widehat{e}_{is} - e_{is}| = \bigo_\mpr\Bigg(\sqrt{\frac{\log n}{Td_T}}+ \sqrt{\frac{\log n}{T{d_T^2}}}\Bigg) = \bigo_\mpr\Biggl(\sqrt{\frac{\log n}{T{d_T^2}}} \Biggr)
\end{equation*}
as desired. 
\end{proof}

\begin{proof}[Proof of Theorem~\ref{quantreg_est}(ii)]
The following bound follows by the same arguments as Lemma 8 of \cite{galvao2018} (note that the condition $n (\log T)^2/T = o(1)$ made in that paper is not used in their proof of Lemma 8):
\begin{equation} \label{Bit00}
\sup_i \normspec{\E[B_{iT}] - B_i } = o(1)\,.
\end{equation}
In addition, we will prove the following bounds
\begin{equation} \label{Bit0}
\sup_i\normspec{\widehat{B}_{iT} - B_{iT} }= \bigo_\mpr(b_{n,T})\,
\end{equation}
with $b_{n,T} =\sqrt{\frac{\log n}{T{d_T^2}}},$
\begin{equation} \label{Bit1}
\sup_i \normspec{B_{iT}-\E[B_{iT}]} = \bigo_\mpr\Bigg(\sqrt{\frac{\log n}{T}}\Bigg),
\end{equation}
and 
\begin{equation} \label{Ait1}
\sup_i \normspec{\widehat{H}_{iT}^{-1}-H_i^{-1}} = \bigo_\mpr\Bigg(\sqrt{\frac{\log n}{T}}\Bigg).
\end{equation}		
The remaining proof follows from similar arguments as the proof of Lemma 10 of \cite{galvao2018}, 
but modifications are needed to take into account that $n (\log T)^2/T = o(1)$ made in that paper is replaced by $\log n = o(T)$ and that the rate changes accordingly. We note that 
\begin{align*}
&\widehat{B}_{iT} - B_{iT} = \frac{1}{T}\sum_{t=1}^{T} (\widehat{f}_{it} -f_{it})\vecz_{it}\vecz_{it}^\top \,.
\end{align*}
Using Taylor expansion, we have
\begin{equation}\label{eq:taylor}
\widehat{f}_{it} - f_{it} = \widehat{e}_{it}^{\,-1} -e_{it}^{-1} = \frac{e_{it} - \widehat{e}_{it}}{e_{it}^2} + \bigo\Big(|\widehat{e}_{it} - e_{it}|^2\Big),
\end{equation}
where the remainder term is uniform in $i,t$ since under Assumption~\ref{A2}, it holds that
\begin{align}
\inf_{i,t} e_{i,t} &=  \nonumber
\inf_{i,t} \frac{q_{i,\tau + d_T}(\vecz_{it}) - q_{i,\tau - d_T}(\vecz_{it})}{2d_T} 
\geq \inf_{i,t}\inf_{|\eta - \tau| 
\leq d_T} \frac{1}{f_{\response\mid \vecz}(q_{i,\eta}(\vecz_{it})\mid \vecz_{it})}
\\
&= \frac{1}{\sup_{i,t}\sup_{\eta, \vecz}f_{\response\mid\vecz}(q_{i,\eta}(\vecz_{it})\mid \vecz_{it})} \geq 1/f_{max}\,, \label{eitlowb}
\end{align}
almost surely. By Assumption~\ref{A1} and Lemma~\ref{lemma2}, the bound in~\eqref{Bit0} follows. 

\noindent
Next we prove~\eqref{Bit1}. Define the matrix~$N_{it}\in \R^{(p+1)\times (p+1)}$ via
\begin{equation*}
N_{it} := f_{it}\vecz_{it}\vecz_{it}^\top -\E[f_{it}\vecz_{it}\vecz_{it}^\top ]\,.
\end{equation*}
It then follows that $\E[N_{it}] = \mathbf{0}.$ 
We denote by $N_{it,j,\ell}$ the $(j,\ell)$-th entry of the matrix~$N_{it}$.
By Assumption~\ref{A1} and the inequality~\eqref{eitlowb}, we have $\sup_{i,t,j,\ell}|N_{it,j,\ell}| \leq {C}_5$ and $\sup_{i,t,j,\ell}\Var[N_{it,j,\ell}] \leq {C}_6$ for some constants ${C}_5$, ${C}_6 > 0$. 
Applying the Bernstein inequality gives, for any $\epsilon_2 > 0$,  
\begin{align*}
\mpr\Bigg(\Big|\sum_{t=1}^{T} N_{it,j,\ell}\Big| > T \epsilon_2\Bigg) 
\leq &2\exp \Bigg(-\frac{\frac{1}{2}T^2 \epsilon_2^2}{\sum_{t=1}^{T} \E[N_{it,j,\ell}^2] + \frac{1}{3}{C}_5 T \epsilon_2}\Bigg)\\
\leq &2\exp\Bigg(-\frac{\frac{1}{2}T^2 \epsilon_2^2}{T{C}_6 + \frac{1}{3}{C}_5 T\epsilon_2}\Bigg)\,.
\end{align*}
Take $\epsilon_2 = {C}_7 T^{-1/2}(\log n)^{1/2}$ for some constant ${C}_7>0$ to be determined later, 
and the right hand side of the inequality becomes, for $\log n/T$ sufficiently small,
\begin{equation*}
2\exp\Bigg(-\frac{1}{2} \frac{({C}_7)^2\log n}{{C}_6 +\frac{1}{3}{C}_5{C}_7 T^{-1/2}(\log n)^{1/2}}\Bigg)\leq 2\exp\Bigg(-\frac{({C}_7)^2\log n}{4{C}_6} \Bigg)\,.
\end{equation*}
Choosing $({C}_7)^2 > 4{C}_6,$ then for every $j,\ell$, it holds that
\begin{align*}
\mpr\Bigg(\sup_i \Big|\frac{1}{T}\sum_{t=1}^{T} N_{it,j,\ell}\Big|> \epsilon_2\Bigg)
\leq \sum_{i=1}^{n} \mpr\Bigg(\Big|\frac{1}{T}\sum_{t=1}^{T} N_{it,j,\ell}\Big|>\epsilon_2\Bigg) 
=& 2n \exp\Big(-\frac{({C}_7)^2\log n}{4{C}_6} \Big)
\to 0\,.
\end{align*}
Thus, $\sup_{i} \normspec{ \frac{1}{T}\sum_{t=1}^{T} N_{it}}= \bigo_\mpr\Big(\sqrt{\frac{\log n}{T}}\Big)$. 
This implies~\eqref{Bit1}.
Finally, we prove the bound~\eqref{Ait1}.
By Assumption~\ref{A1}, it holds that~$\sup_i \normspec{H_i^{-1}} < \infty.$
Moreover, we have
\begin{multline}
\widehat{H}_{iT}^{-1} - H_i^{-1} 
= H_i^{-1}(H_i \widehat{H}_{iT}^{-1} - I) 
= H_i^{-1}(H_i - \widehat{H}_{iT})\widehat{H}_{iT}^{-1}\\
=H_i^{-1}(H_i - \widehat{H}_{iT}) H_{i}^{-1} + \bigo\Big (\normspec{H_i^{-1}}^2\,\normspec{\widehat{H}_{iT} - H_i}^2\Big)\label{eq:h_decomp}\,,
\end{multline} 
where 
\begin{equation*}
\sup_i \normspec{\widehat{H}_{iT} - H_i} = \bigo_\mpr\Bigg(\sqrt{\frac{\log n}{T}}\Bigg)\,
\end{equation*}
holds by an application of the Bernstein inequality which is similar to the one given above. 
This completes the proof.
\end{proof}

\textit{Proof sketch of Theorem~\ref{qrint_est}} Both parts follow by simple computations provided that we can establish the bound
\[
\sup_{|\eta - \tau| \leq \varepsilon} \sup_{i\in\{1,\dots,n\} }|\hat\alpha_i(\eta)-\alpha^*_i(\eta)| =\bigo_\mpr\biggl(\sqrt{\frac{\log n}{T}}\biggr)\,.
\]
for some $\varepsilon > 0$. This can be established by following the arguments given in Step 1--Step 3 in the proof of Theorem 3.2 in~\cite{kato2012asymptotics}. Note that all empirical processes appearing in those steps retain the same complexity (in terms of VC dimension and envelope functions). Note also that the assumption that $T$ grows at most polynomially in $n$ made in their Theorem 3.2 can be dropped at the cost of replacing $\log n$ by $\log(T \vee n)$, see also the discussion in the latter paper following Theorem 3.2. \hfill $\Box$

\subsubsection{Proof for quantile regression in the dependent case (Theorem~\ref{quantreg_est1})}

Before proving Theorem~\ref{quantreg_est1} we collect some preliminary technical results. Let $\mathcal{S}^{p+1} := \{\mathbf{a}\in\R^{p+1},\normtwo{\mathbf{a}}=1\}$. Let $\tilde{\mathcal{T}} = [\tau-\eps,\tau + \eps]$ where $\eps > 0$ is such that $\tilde{\mathcal{T}} \subset \mathcal{T}$ for $\mathcal{T}$ from Assumption~\ref{A3}. Define the function classes
{
\begin{align}
\mathcal{G}_1 &:=\Big\{(\vecz,y) \mapsto \mathbf{a}^\top\vecz (\1\{y\le \vecz^\top \mathbf{b}\}-\tau) \1\{\|\vecz\|_2\le \kappa \}:\mathbf{b}\in\R^{p+1},\tau\in\tilde{\mathcal{T}}, \mathbf{a} \in \mathcal{S}^{p+1}  \Big\}\,. \label{def:G1}
\\
\mathcal{G}_2(\delta) &:= \Big\{ (y,\vecz) \mapsto \mathbf{a}^\top \vecz (\1\{y\leq \mathbf{b}_1^\top \vecz\} - \1\{y\leq \mathbf{b}_2^\top \vecz\})\1\{\| \vecz \|_2 \leq \kappa \} \Big| \|\mathbf{b}_1-\mathbf{b}_2\|_2 \leq \delta, \mathbf{a}\in \mathcal{S}^{p+1} \Big\}. \label{def:G2}
\end{align}
}
Further, define the functions
{
\[
g_{b,k,\ell}(\vecz_1,\vecz_2,y_1,y_2) := \vecz_{1,k}\vecz_{2,\ell}(\1\{y_1\le \vecz_1^\top \mathbf{b}\}-\tau)(\1\{y_2\le \vecz_2^\top \mathbf{b}\}-\tau), \quad k,\ell = 1,\dots,d, \mathbf{b} \in \R^d. 
\]
}
With this notation, let 
\begin{align*}
\mathcal{G}_{3,k,\ell} &:= \Big\{ (\vecz_1,\vecz_2,y_1,y_2) \mapsto g_{\bm{b},k,\ell}(\vecz_1,\vecz_2,y_1,y_2) : \bm{b} \in \R^{p+1}  \Big\},
\end{align*}
and 
\[
\mu_{3,k,\ell}(\bm{b},i,j) := \E[g_{\bm{b},k,\ell}(\vecz_{i1},\vecz_{i,1+j},y_{i1},y_{i,1+j})].
\] 
Consider the empirical measures $\mpr_{i,j,T}$ corresponding to $\big\{(\vecz_{i,t},\vecz_{i,t+j},y_{i,t},y_{i,t+j})\big\}_{t=1,\dots,T}$ and denote by $\tilde \mpr_{i,j}$ the distribution of $(\vecz_{i,1},\vecz_{i,1+j},y_{i,1},y_{i,1+j})$. Note that for $j\neq 0$ this includes "observations" outside of the observable sample. This quantity only appears in the proofs and is not used to compute any of the estimators.  With this notation we have the following technical result.  

\begin{lemma}\label{lem:tailboundsdep} Assume the conditions of Theorem~\ref{quantreg_est1}(i). For $n$ sufficiently large we have for all $s>0$ and all $ 1 \ll q_{n,T}$ with $q_{n,T}^2\log q_{n,T} = o(T)$ for constants $C_{\mathcal{G}_1}, \tilde C_{\mathcal{G}_1}$  
\begin{align}\label{eq:mix1}
\mpr\Bigg(  \norm{\mpr_{i,T}-\mpr_i}_{\mathcal{G}_1}  \ge C_{\mathcal{G}_1} \Bigg( \sqrt{\frac{\log q_{n,T}}{T} }+ \sqrt{\frac{s}{T}} + \frac{s q_{n,T}}{T}\Bigg) \Bigg)
&\le 2e^{-s }+2T\beta(q_{n,T})\,,
\\ \label{eq:mix1_2}
\mpr\Big( \max_{i=1,\dots,n} \norm{\mpr_{i,T}-\mpr_i}_{\mathcal{G}_1} \leq  \tilde C_{\mathcal{G}_1}\sqrt{\frac{\log(nT)}{T}} \Big) &\geq 1 - \frac{1}{nT}.
\end{align}
Further, we also have for all $s>0$ and all $ 1 \ll q_{n,T}$ with $q_{n,T}^2 \log (m_T\vee q_{n,T}) = o(T)$ and a constant $C_{\mathcal{G}_3}$ for $n$ sufficiently large
\begin{align} \label{eq:mixG3}
\mpr\Bigg(  \norm{\mpr_{i,j,T}-\tilde \mpr_{i,j}}_{\mathcal{G}_{3,k,\ell}}  \ge C_{\mathcal{G}_3} \Bigg(\sqrt{\frac{m_T \log q_{n,T}}{T} } + \sqrt{\frac{s m_T}{T}} + \frac{s q_{n,T}}{T} \Bigg) \Bigg)
&\le 2e^{-s }+2T\beta(q_{n,T})\,
\\ \label{eq:mixG3_2}
\max_{i,j,k,\ell} \norm{\mpr_{i,j,T}-\mpr_i}_{\mathcal{G}_{3,k,\ell}} &= O_\mpr\Big(\sqrt{\frac{m_T\log(nT)}{T}}\Big).
\end{align}
Next, let 
\[
\sigma^2_{q,i}(g) := \Var\Big(\frac{1}{\sqrt{q}}\sum_{t=1}^q g(\vecz_{it},Y_{it})\Big)
\]
and assume that 
{
\begin{equation}\label{eq:sigma0}
\sigma^2_{n,T}(\delta)\ge \sup_i\sup_{g\in\mathcal{G}_2(\delta)} \sigma^2_{q,i}(g)\,.
\end{equation}
}
Then for any $s >0$ and any $q_{n,T}$ satisfying 
\begin{equation}\label{eq:sigma}
q_{n,T}^2 \log \Big(\frac{q_{n,T}}{\sigma_{n,T}^2(\delta)}\Big)\le \tilde C T\sigma_{n,T}^2(\delta)\,,
\end{equation}
for a certain constant $\tilde C$ depending only on $\kappa$ and the dimension $p$ of $\vecz_{it}$ we have 
\begin{multline}\label{eq:mixG2}
\mpr\Bigg( \norm{\mpr_{i,T}-\mpr_i}_{\mathcal{G}_2(\delta)} \ge C \Bigg(\sqrt{ \frac{\sigma^2_{n,T}(\delta)}{T}   \log \Big( \frac{q_{n,T}}{\sigma^2_{n,T}(\delta)}\Big) } + \sqrt{\frac{\sigma^2_{n,T}(\delta) s}{T}} + \frac{sq_{n,T}}{T}\Bigg) \Bigg)
\\
\le 2e^{-s}+2T\beta(q_{n,T}).
\end{multline}
In particular, for $\delta = \delta_{n,T} := (C T^{-1}\log(nT))^{1/2} $ with $C > 0$ arbitrary but fixed we obtain
\begin{equation} \label{eq:mixG2_2}
\max_{i=1,\dots,n} \norm{\mpr_{i,T}-\mpr_i}_{\mathcal{G}_2(\delta_{n,T})} = O_\mpr\Big( \frac{(\log(nT))^{5/4}}{T^{3/4}}\Big).
\end{equation}
Finally, letting $U_{it} := F_{Y_{it}|X_{it}}(Y_{it}|X_{it})$,
\begin{equation}
\label{eq:norm_phi}
\sup_{i} \normtwo{ \frac{1}{2Td_T}  \sum_{t=1}^{T} \vecz_{it} \Big( \1 \{U_{it} \leq \tau + d_T\} - \1 \{U_{it} \leq \tau - d_T\} - 2d_T\Big)  } =\bigo_{\mpr} \Bigg(\frac{\log n}{\sqrt{T d_T}} \Bigg)\,.
\end{equation}
\end{lemma}

\textit{Proof of Lemma~\ref{lem:tailboundsdep}}
The proof strategy for many parts is similar to that in the proof of Lemma~5 in~\cite{galvao2018} and we will only point out the relevant differences. We will repeatedly apply Proposition~C.2 from \cite{kato2012asymptotics}. That result requires the corresponding function classes to be centered. Assume that $\mathcal{F}$ is a class of functions that are not centered and such that $\sup_{Q} N(\mathcal{F}, L_1(Q),\epsilon)\le (A/\epsilon)^{\nu}$ for some constants $A, \nu$ and let $\tilde{\mathcal{F}} := \{f - \mpr f: f \in \mathcal{F}\}$. The it is easy to see that $N(\tilde{\mathcal{F}}, L_1(Q),\epsilon) \le N(\mathcal{F}, L_1(Q),\epsilon/2)$ and $\|f-\mpr f\|_\infty \le 2 \|f\|_\infty$ so that Proposition~C.2 from \cite{kato2012asymptotics} can be applied to non-centered function classes to obtain bounds on $\|\mpr_T f - \mpr f\|_{\mathcal{F}}$. This fact will be repeatedly used throughout the proofs that follow.  

\medskip

\textbf{Proof of~\eqref{eq:mix1} and~\eqref{eq:mix1_2}} By the proof of Lemma 5 in~\cite{galvao2018}, it holds for any $g\in\mathcal{G}_1$ that
\[
\normsup{g}\le U_1,\quad\text{and} \quad \sup_i\sup_{g\in\mathcal{G}_1}\Var\big(g(\vecz_{i1},Y_{i1})\big)\le U_2\,
\]
with some positive universal constants $U_1$ and $U_2.$ 
Moreover, it holds for any probability measure~$Q$ and any $0<\epsilon<1$ that
\[
N(\mathcal{G}_1, L_1(Q),\epsilon)\le (A/\epsilon)^{\nu}\,
\]
with some positive constants $A,\nu<\infty.$ The claim in~\eqref{eq:mix1} follows by Proposition C.2 of~\cite{kato2012asymptotics}. For~\eqref{eq:mix1_2}, let $q_{n,T} := C_1 \log(nT)$ with the constant $C_1\ge 1$ satisfying $b_{\beta}^{C_1}\le e^{-2}$ and $s = 2\log (nT)$. Clearly $q_{n,T} \gg 1$, $q_{n,T}^2 \log(q_{n,T}) = o(T)$, so~\eqref{eq:mix1_2} follows from the union bound and simple calculations. 

\medskip

\textbf{Proof of~\eqref{eq:mixG3} and~\eqref{eq:mixG3_2}} Observe that any function in $\mathcal{G}_{3,k,\ell}$ can be expressed as through sums and products of functions from the classes 
{
$\mathcal{H}_1 := \{(y_1,\vecz_1,y_2,\vecz_2) \mapsto \vecz_{1,k}\vecz_{2,\ell}| 1 \leq k,\ell \leq p+1\}, \mathcal{H}_2 := \{(y_1,\vecz_1,y_2,\vecz_2) \mapsto \tau - \1\{y_1 \leq \vecz_1^\top \mathbf{b} \} | \mathbf{b} \in \mathbb{R}^{p+1}\}$, $\mathcal{H}_3 := \{(y_1,\vecz_1,y_2,\vecz_2) \mapsto \tau - \1\{y_2 \leq \vecz_2^\top \mathbf{b} \} | \mathbf{b} \in \mathbb{R}^{p+1}\}$ 
}
and that each of the three classes satisfies
\begin{equation*}
N(\mathcal{H}_j,L_2(Q),\epsilon) \leq (\widetilde{A}/ \epsilon)^{\widetilde{v}}, 
\end{equation*}
for all $0 < \epsilon\leq 1$ and some constants $\widetilde{A}, \widetilde{v} < \infty$. Hence, by the Cauchy-Schwarz inequality and Lemma 23 in \cite{belloni2019conditional} (note that the proof of this Lemma continues to hold for arbitrary probability measures, discreteness is not required), we find that
\begin{equation*}
N(\mathcal{G}_{3,k,\ell},L_1(Q),\epsilon) \leq (A/\epsilon)^v, 
\end{equation*}
for some $A,v < \infty$. Next, note that under Assumption~\ref{B1} the series of random vectors $\{\xi_{i,t} := (Y_{it}, \vecz_{it}, Y_{it+j}, \vecz_{it+j})\}_{t\in\mathbb{Z}}$ is $\beta$-mixing with mixing coefficients $\widetilde{\beta}(t)$ satisfying $\widetilde{\beta}(t) \leq \beta(0\vee (t-j))$. Since the functions in $\mathcal{G}_{3,k,\ell}$ are uniformly bounded, Lemma C.1 in \cite{kato2012asymptotics}(applied with $\delta = 1$ in the notation of that Lemma) yields 
\[
|Cov(g(\xi_{i,t}),g(\xi_{i,t+j}))| \leq C \tilde \beta(j)^{1/2}
\]
for a constant $C$ independent of $n,T,i$. For $g \in \mathcal{G}_{3,k,\ell}$ let
\begin{equation*}
\sigma_{q,i,j}^2(g) := Var \Big(\frac{1}{\sqrt{q}} \sum_{t=1}^q f(Y_{it},\vecz_{it},Y_{it+j},\vecz_{it+j})\Big).
\end{equation*}
We have
\begin{align*}
\sigma_{q,i,j}^2(g) &= Var(f(\xi_{i,t})) + 2\sum_{j=1}^{q-1} \Big(1-\frac{j}{q}\Big) Cov(f(\xi_{i,1}), f(\xi_{i,1+j}))
\\
&\leq C + 2C \sum_{j=1}^{m_T} \Big(1-\frac{j}{q}\Big) + 2C \sum_{j=m_T}^{q-1} \tilde \beta(j)^{1/2}
\\
&\leq 2(m_T+1)C + 2C \sum_{j=1}^{\infty} \tilde \beta(j)^{1/2}
\\
&\leq \tilde C(m_T+1)
\end{align*}
for a constant $\tilde C$ independent of $i,n,T$. The claim in~\eqref{eq:mixG3} follows by an application of Proposition~C.2 in \cite{kato2012asymptotics}. To obtain~\eqref{eq:mixG3_2}, set $s = 4\log(nT)$ and $q_{n,T} = C \log(nT)$ with $C$ chosen such that $\beta(C) \leq e^{-4}$. 

\medskip

\textbf{Proof of~\eqref{eq:mixG2} and~\eqref{eq:mixG2_2}} By the proof of Lemma 5 in~\cite{galvao2018}, it holds for any $g\in\mathcal{G}_2(\delta)$ that
\[
\normsup{g}\le U_2\,
\]
with some constant $U_2>0,$ 
and  it also  holds for large $n,T$ satisfying $\frac{1}{nT}\le \delta\le 1 $ that
\begin{align*}
\sigma^2_{q,i}(g)
=\Var\Big(\frac{1}{\sqrt{q}}\sum_{t=1}^q g(\vecz_{it},Y_{it})\Big)
\le & C_{\sigma,2} \delta \log(nT),\quad i=1,\dots,n\,,
\end{align*}
where $C_{\sigma,2}$ is a constant. Moreover, by the first display in the proof of Lemma 5 in~\cite{galvao2018}, it holds for any probability measure~$Q$ and any $0<\epsilon<1$ that
\[
N(\mathcal{G}_2(\delta), L_1(Q),\epsilon)\le (A/\epsilon)^{\nu}\,
\]
with some positive constants $A,\nu<\infty.$
Invoking Proposition C.2 of~\cite{kato2012asymptotics} gives~\eqref{eq:mixG2}. To prove~\eqref{eq:mixG2_2} pick 
\[
q_{n,T}:=C_1\log(nT)
\]
with the universal constant $C_1\ge1$ satisfying $b_{\beta}^{C_1}\le e^{-2},$ and set 
{
\[
\sigma^2_{n,T} := C_{\sigma,2} \log(nT) \delta\,. %
\]
}
With this choice~\eqref{eq:sigma0} holds by definition and we have
\begin{align*}
\frac{q_{n,T}^2}{\sigma_{n,T}^2}\log\Big(\frac{q_{n,T}}{\sigma_{n,T}^2}\Big) \lesssim \sqrt{T} \log(nT)^{1/2} \log(T)^{1/2} = o(T) 
\end{align*}
so that~\eqref{eq:sigma} holds for $n,T$ large enough. Let $s_{n,T}:=2\log n.$ By elementary computations using the fact that $\log(n)^3=o(T)$ by assumption the claim follows by applying the union bound.

\medskip

\textbf{Proof of~\eqref{eq:norm_phi} } Denote the $k$-th element of vector~$\vecz_{it}$ by $\vecz_{it,k}.$ Define the function~$f_k$ via
\begin{align*}
f_k:\rd\times \R&\to \R \\
(\vecz_{it},Y_{it})& \mapsto \vecz_{it,k}\Big( \1 \{U_{it} \leq \tau + d_T\} - \1 \{U_{it} \leq \tau - d_T\} - 2d_T\Big) \,.
\end{align*}
By Lemma 11 in~\cite{galvao2018}, we have
\[
\sup_{i,t,k}|f_k(\vecz_{it},Y_{it})|\le C_1,\quad \E[f_k(\vecz_{it},Y_{it})]=0\,,
\]
and
\[
\sigma^2_{q,i}(f):=\Var\Big(\frac{1}{\sqrt{q}}\sum_{t=1}^q f_k(\vecz_{it},Y_{it}) \Big) \le C_2 d_T|\log (d_T)|\,,
\]
where $C_1,C_2$ are constants independent of $i,T.$
Note that the constants are independent of $i$, throughout the proof we drop the dependence of $\sigma^2_{q,i}(f)$ on $i$, and denote it by $\sigma^2_{q}(f)$ instead.
Applying Corollary C.1 in~\cite{kato2012asymptotics}, we have for some constant $C$ independent of $i,T,k$ and $q_{n,T}\in[1,\frac{T}{2}]$  and for some $s_{n,T}>0$ 
\[
\mpr\Bigg( \Bigg| \frac{1}{T}\sum_{t=1}^T f_k(\vecz_{it},Y_{it})\Bigg| \ge C\Bigg(  \frac{\sqrt{(s_{n,T}\vee 1)}}{\sqrt{T}}\sigma_q(f) + \frac{s_{n,T}q_{n,T}}{T}  \Bigg) \Bigg)
\le 2e^{-s_{n,T}} + 2T\beta(q_{n,T})\,.
\]
Set $s_{n,T}:=2\log n$ and let  
\[
q_{n,T}:= C_1\log (nT)
\]
where $C_1>1$ is a constant satisfying $b_\beta^{C_1} \le e^{-2}.$
Then, it holds for some large $n$ and $T$ and a small $d_T$ that 
\[
\frac{\sqrt{(s_{n,T}\vee 1)}}{\sqrt{T}}\sigma_q(f) + \frac{s_{n,T}q_{n,T}}{T}  
\le 
\sqrt{\frac{2\log n}{T}}\sqrt{d_T}\sqrt{\log\Big(\frac{1}{d_T}\Big)} + \frac{C_1 \log n\log(nT)}{T}\,.
\]
By Assumption~\ref{B3} and $T$ grow at most polynomial in $n$, we have
\[
\sqrt{\frac{2\log n}{T}}\sqrt{d_T}\sqrt{\log\Big(\frac{1}{d_T}\Big)} + \frac{C_1 \log n\log(nT)}{T}\lesssim \frac{\log n}{\sqrt{T}}\sqrt{d_T} + \frac{(\log n)^2}{T}\,.
\]
Moreover, note that
\[
2e^{-s_{n,T}}+ 2T\beta(q_{n,T}) \lesssim \frac{1}{n^2}+\frac{1}{n^2T}\,.
\]
Taking the union bound over $i=1,\dots,n$ then gives
\begin{equation*}
\sup_{i} \normtwo{ \frac{1}{2Td_T}  \sum_{t=1}^{T} \vecz_{it} \Big( \1 \{U_{it} \leq \tau + d_T\} - \1 \{U_{it} \leq \tau - d_T\} - 2d_T\Big)  } =\bigo_{\mpr} \Bigg(\frac{\log n}{\sqrt{T d_T}} \Bigg)\,.
\end{equation*}
This completes the proof of~\eqref{eq:norm_phi}. Now the proofs of all results in Lemma~\ref{lem:tailboundsdep} are complete. \hfill $\Box$

\subsubsection{Proof of Theorem~\ref{quantreg_est1} (i)}

The Lemma~C.2 in~\cite{chao2017} can be used in our framework by setting $t=2, n=T, \xi_m= \kappa, g_n=0,$ which implies the following for each $i\in\{1,\dots,n\}$
\begin{equation}
\label{eq:supset}
\bigg\{\sup_{\tau\in\mathcal{T}}\normtwo{\est_i-\vecg^*_i}\le \frac{4~ \norm{\mpr_{i,T}-\mpr_i}_{\mathcal{G}_1}}{\inf_{\tau\in\mathcal{T}}\lambda_{\min}(\widetilde J_i)} \bigg\}
\supseteq
\bigg\{ \norm{\mpr_{i,T}-\mpr_i}_{\mathcal{G}_1} < \frac{\inf_{\tau\in\mathcal{T}}\lambda^2_{\min}(\widetilde J_i)}{8\kappa\overline{f'}\lambda_{\max} \big( \E\bigl[ \vecz_{it}\vecz_{it}^\top \bigr] \big)}\bigg\}\,,
\end{equation}
where $\widetilde J_i:=\E\bigl[ \vecz_{it}\vecz_{it}^\top f_{\response_{it}|\vecz_{it}}(\vecz_{it}^\top \vecg^*_i |\vecz_{it}) \bigr],$ with the function class $\mathcal{G}_1$ defined in~\eqref{def:G1}. By the assumption that $(\log n)^3 = o(T)$  and Assumptions~\ref{A1}-\ref{A3}, it holds for sufficiently large $n,T$ that 
\begin{equation}
\label{eq:subset}
\tilde C_{\mathcal{G}_1}\sqrt{\frac{\log(nT)}{T}}
\le
\frac{\inf_{\tau\in\mathcal{T}}\lambda^2_{\min}(\widetilde J_i)}{8\kappa\overline{f'}\lambda_{\max} \big( \E\bigl[ \vecz_{it}\vecz_{it}^\top \bigr] \big)}\,.
\end{equation}
Define the event 
\begin{equation}\label{eq:OmegaG1}
\Omega_{\mathcal{G}_1} :=\Big\{\norm{\mpr_{i,T}-\mpr_i}_{\mathcal{G}_1} \le \tilde C_{\mathcal{G}_1}\sqrt{\frac{\log(nT)}{T}} \Big\}.
\end{equation}
By the relation~\eqref{eq:supset}, we obtain that on the event $\Omega_{\mathcal{G}_1},$ it holds that 
\[
\sup_{\tau\in\mathcal{T}}\normtwo{\est_i-\vecg^*_i}\le C_3\sqrt{\frac{\log(nT)}{T}} \,,
\]
where $C_3>0$ is a constant independent of $i,n,T$. Combined with~\eqref{eq:mix1_2} we find that for all sufficiently large $n,T$ 
\begin{equation}\label{eq:tailgammamix}
\mpr\Big(\sup_{\tau\in\mathcal{T}}\normtwo{\est_i-\vecg^*_i}\le C_3\sqrt{\frac{\log(nT)}{T}}\Big) \ge 1 - \frac{1}{nT}.
\end{equation}
This completes the proof of Theorem~\ref{quantreg_est1} (i). \hfill $\Box$

\subsubsection{Proof of Theorem~\ref{quantreg_est1} (ii).} The assumptions made imply that the smallest eigenvalues of the matrices $B_i$ are bounded away from zero uniformly in $i$. Since we work in fixed dimension, it suffices to show that 
{
\[
\max_{i,k,\ell} |\widehat{B}_{iT,k,\ell} - B_{i,k,\ell}| + \max_{i,k,\ell}|\widehat{H}'_{iT,k,\ell} - \widetilde H_{i,k,\ell}| = o_\mpr(1). 
\]
}
We will consider the two sums separately, starting with $\widehat{B}_{iT}$. Note that
{
\begin{align*}
\normspec{\widehat{B}_{iT} - B_i} \leq& \normspec{\frac{1}{T} \sum_{t=1}^T \vecz_{it}\vecz_{it}^\top(\hat f_{it} - f_{it})}
+ \normspec{\frac{1}{T} \sum_{t=1}^T \vecz_{it}\vecz_{it}^\top f_{it} - \E\big[\vecz_{it}\vecz_{it}^\top f_{it}\big]}
\\ 
&+ \normspec{|\E\big[\vecz_{it}\vecz_{it}^\top f_{it}\big] - B_i }.   	
\end{align*}
}
The bound{$\max_i \normspec{ \E\big[\vecz_{it}\vecz_{it}^\top f_{it}\big] - B_i }= o(1)$} follows from standard Taylor expansions similarly to the proof of Lemma 8 in \cite{galvao2018}. Further, we have
{
\[
\max_i  \normspec{ \frac{1}{T} \sum_{t=1}^T \vecz_{it}\vecz_{it}^\top(\hat f_{it} - f_{it}) } \leq \kappa^2 \max_{i,t} |\hat f_{it} - f_{it}| = o(1)
\]
}
by Lemma~\ref{lem:mix_aux1}. To bound {$\max_i \normspec{T^{-1} \sum_{t=1}^T \vecz_{it}\vecz_{it}^\top f_{it} - \E\big[\vecz_{it}\vecz_{it}^\top f_{it}\big] }$} note that the entries of $\vecz_{it}\vecz_{it}^\top f_{it}$ are uniformly bounded. Thus an application of Lemma C.1 from \cite{kato2012asymptotics} shows that $\Var(T^{-1} \sum_{t=1}^T \vecz_{it}\vecz_{it}^\top f_{it}) \leq C_1$ for a constant $C_1$. Now apply Corollary C.1 from \cite{kato2012asymptotics} with $s = 2\log n, q = c\log(nT)$ for a suitable constant $c$ to obtain {$\max_i  \normspec{ T^{-1} \sum_{t=1}^T \vecz_{it}\vecz_{it}^\top f_{it} - \E\big[\vecz_{it}\vecz_{it}^\top f_{it}\big] } = o_\mpr(1)$.}

Next we proceed to bound {$\normspec{\widehat H_{iT}' - \widetilde H_i }.$ } Recall the notation from the paragraph before Lemma~\ref{lem:tailboundsdep}. Observe the decomposition
{
\begin{align*}
[\widehat{H}'_{iT}]_{k,\ell} - [\widetilde H_{i}]_{k,\ell} %
=~&  \tau (1-\tau)\frac{1}{T} \sum_{t=1}^{T} \Big\{[\vecz_{it}\vecz_{it}^\top]_{k,\ell} - \E\big[[\vecz_{it}\vecz_{it}^\top]_{k,\ell}\big]\Big\}
\\
& + \sum_{1\le j\le m_T}\Big(1-\frac{j}{T}\Big)\Big( \mpr_{i,j,T}g_{\hat\vecg_i(\tau),k,\ell} - \mu_{3,k,\ell}(\hat\vecg_i(\tau),i,j)\Big)
\\
& + \sum_{1\le j\le m_T}\Big(1-\frac{j}{T}\Big)\Big( \mu_{3,k,\ell}(\hat\vecg_i(\tau),i,j) - \mu_{3,k,\ell}(\vecg_{i}^*(\tau),i,j)\Big)
\\
& + \sum_{1\le j\le m_T}\Big(1-\frac{j}{T}\Big) \mu_{3,k,\ell}(\vecg_{i}^*(\tau),i,j) - \sum_{j=1}^\infty \mu_{3,k,\ell}(\vecg_{i}^*(\tau),i,j)
\\
& + R_{n,T,k,\ell,i}(\tau) 
\\
=:~& \sum_{j=1}^4 \Delta_{i,k,\ell,n,T}^{(j)}(\tau) +  R_{n,T,k,\ell,i}(\tau). 
\end{align*}
}
where $R_{n,T,k,\ell,i}(\tau)$ arises due to the summation range over $T_j$. Note that
\[
\sup_{i,k,\ell}|R_{n,T,k\ell,i}| \leq \frac{2m_T^2\kappa^2}{T} = o(1)
\]
since $T \geq |T_j| \geq T - m_T$ and $\|g_{\hat\vecg_i(\tau),k,\ell}\|_\infty \leq 2\kappa^2$. The bound $\max_{i,k,\ell}\sup_{\tau \in \mathcal{T}} |\Delta_{i,k,\ell,n,T}^{(1)}(\tau)| = o_\mpr(1)$ follows by combining Lemma C.1 and Proposition C.2 from \cite{kato2012asymptotics} with $s = 2\log n, q = c\log n$ for a suitable constant $c$. The bound $\max_{i,k,\ell}\sup_{\tau \in \mathcal{T}} |\Delta_{i,k,\ell,n,T}^{(4)}(\tau)| = o(1)$ follows from the arguments in the last paragraph in the proof of Lemma 12 in \cite{galvao2018}. To bound {$\max_{i,k,\ell}\sup_{\tau \in \mathcal{T}} |\Delta_{i,k,\ell,n,T}^{(3)}(\tau)|  = o(1)$} note that under Assumption~\ref{B2} the maps $\bm{b} \mapsto \mu_{3,k,\ell}(\bm{b},i,j)$ are Lipshitz continuous with Lipshitz constant \red{$\kappa^2$} bounded uniformly in $n,T,i,j,k,\ell$. Thus
{
\[
\max_{i,k,\ell}\sup_{\tau \in \mathcal{T}} |\Delta_{i,k,\ell,n,T}^{(3)}(\tau)| \leq \kappa^2 m_T \max_i \sup_{\tau \in \mathcal{T}} \|\hat \vecg_i(\tau) - \vecg_{i}^*(\tau)\|_2 = O_\mpr\Big(m_T \sqrt{\frac{\log (nT)}{T}} \Big) = o_\mpr(1)
\]
}
by the first part of the {theorem} and Assumption~\ref{B3}. Finally, observe that
\[
\max_{i,k,\ell}\sup_{\tau \in \mathcal{T}} |\Delta_{i,k,\ell,n,T}^{(2)}(\tau)| \leq m_T \max_{i,k,\ell,j} \|\mpr_{i,j,T} - \tilde\mpr_{i,j}\|_{\mathcal{G}_{3,k,\ell}} = O_\mpr\Big(\sqrt{\frac{m_T^3 \log(nT)}{T}} \Big) = o_\mpr(1) 
\] 
where we used~\eqref{eq:mixG3_2} and the assumption on $m_T$. This completes the proof of Theorem~\ref{quantreg_est1} (ii). \hfill $\Box$

\subsubsection{Technical results used in the proof of Theorem~\ref{quantreg_est1} (ii)}

\begin{lemma}\label{lem:mix_aux1}
Let the assumptions stated in Theorem~\ref{quantreg_est1}(i) and Assumption~\ref{B3} hold. Then
\[
\sup_{i,t} |\hat f_{it}-f_{it}| = o(1).
\]
\end{lemma}

\textit{Proof of Lemma~\ref{lem:mix_aux1}}
The proof strategy follows from Lemma  11 in~\cite{galvao2018}, where we employ the Bernstein inequality for $\beta$-mixing sequences (Corollary C.1 in~\cite{kato2012asymptotics}). Define $\widehat{e}_{it} := \widehat{f}_{it}^{-1}$ and $e_{it} := 1/f_{it}$. By definitions of $\widehat{e}_{is}$ and $e_{is} $ it holds that
\begin{equation*}
\widehat{e}_{is} - e_{is} = \vecz_{is}^\top \Big(\big(\est_{i}(\tau + d_T) - \vecg^*_i(\tau + d_T) \big)- \big(\est_{i}(\tau - d_T) - \vecg^*_i(\tau - d_T) \big)\Big)/2d_T\,.
\end{equation*}
By Lemma~\ref{lem:decmp_gamma}and the assumptions $\frac{\log (nT)}{ T d_T^2 }=o(1), \log(n)^3 = o(T)$ we obtain
\[
\max_{i,s} |\widehat{e}_{is} - e_{is}| \leq 2\kappa^2 \max_i \normop{B_i^{-1}} \max_i \Big\| \frac{1}{2Td_T}\sum_{t=1}^T \psi_{i,\tau + d_T}(\vecz_{it},Y_{it}) - \psi_{i,\tau - d_T}(\vecz_{it},Y_{it}) \Big\|_2 + o_\mpr(1).
\] 
Letting $U_{it} := F_{Y_{it}|X_{it}}(Y_{it}|X_{it})$ we see that {$\1 \{Y_{it} \leq \vecg_{i}^*(\tau \pm d_T)\} = \1 \{U_{it} \leq \tau \pm d_T\}$} and hence 
\begin{multline*}
\frac{1}{2Td_T}\sum_{t=1}^T \psi_{i,\tau + d_T}(\vecz_{it},Y_{it}) - \psi_{i,\tau - d_T}(\vecz_{it},Y_{it})
\\
= \frac{1}{2Td_T}  \sum_{t=1}^{T} \vecz_{it} \Big( \1 \{U_{it} \leq \tau + d_T\} - \1 \{U_{it} \leq \tau - d_T\} - 2d_T\Big)
\end{multline*}
Thus $\max_{i,s} |\widehat{e}_{is} - e_{is}| = o_\mpr(1)$ by~\eqref{eq:norm_phi}. Finally, under the assumptions made we have $\min_{i,t} e_{it} \geq 1/f_{max}$, see \eqref{eitlowb}. The claim follows by a Taylor expansion of $x \mapsto 1/x$. \hfill $\Box$

\begin{lemma}\label{lem:decmp_gamma}
Let the assumptions stated in Theorem~\ref{quantreg_est1}(i) and Assumption~\ref{B3} hold. It holds for every $i\in\{1,\dots,n\}$ that
\[
\est_i(\tau)-\vecg^*_i(\tau)=-\frac{1}{T}B_i^{-1}\sum_{t=1}^T \psi_{i,\tau}(\vecz_{it},\response_{it})+R_{n,T,i}(\tau)\,,
\]
where 
\[
B_i := \E[f_{\response\mid\vecz}(q_{i,\tau}(\vecz_{i1})\mid \vecz_{i1})\vecz_{i1}\vecz_{i1}^\top],\quad   \psi_{i, \tau} (\vecz, \response) := \vecz (\1(\response \le q_{i,\tau}(\vecz))-\tau\,,
\]
and 
\[
\sup_i\sup_{\tau\in\mathcal{T}}\normtwo{R_{n,T,i}(\tau)} =  \bigo_{\mpr}\Bigg( \frac{(\log(nT))^{5/4}}{T^{3/4}} \Bigg)\,.
\]
\end{lemma}

\textit{Proof of Lemma~\ref{lem:decmp_gamma}} Observe the decomposition
{
\begin{equation*}
\widehat{\bm{\gamma}}_{i}(\eta) - \bm{\gamma}_{i}^*(\eta) = -\frac{1}{T}B_i^{-1}\sum_{t=1}^{T} \vecz_{it}(\1(Y_{it} \leq q_{i,\eta}(\vecz_{it})) - \eta) + r_{i,1}(\eta) + r_{i,2}(\eta) + r_{i,3}(\eta),
\end{equation*}
}
where
{
\begin{align*}
r_{i,1}(\eta) &:= \frac{1}{T}B_i^{-1}\sum_{t=1}^{T} \vecz_{it}(\1(Y_{it} \leq \vecz_{it}^\top \widehat{\bm{\gamma}}_{i}(\eta)) - \eta),
\\
r_{i,2}(\eta) &:= - \frac{1}{T}B_i^{-1} \sum_{t=1}^{T} \Big\{\vecz_{it}\Big( \1(Y_{it} \leq \vecz_{it}^\top \widehat{\bm{\gamma}}_{i}(\eta)) -\1(Y_{it} \leq \vecz_{it}^\top \bm{\gamma}_{i}^*(\eta)) \Big)
\\
& \quad \quad \quad \quad \quad \quad \quad \quad- \int z[F_{Y|Z}(z^\top \widehat{\bm{\gamma}}_{i}(\eta) \mid z) - F_{Y|Z}(z^\top {\bm{\gamma}}_{i}^*(\eta)\mid z)] dP^{\vecz_{i1}}(z)\Big\},
\\
r_{i,3}(\eta) &:= - B_i^{-1}\Big[ \int z[F_{Y|Z}(z^\top \widehat{\bm{\gamma}}_{i}(\eta) \mid z) - F_{Y|Z}(z^\top {\bm{\gamma}}_{i}^*(\eta)\mid z)] dP^{\vecz_{i1}}(z) - B_i (\widehat{\bm{\gamma}}_{i}(\eta) - \bm{\gamma}_{i}^*(\eta))\Big].
\end{align*}
}
Let $R_{iT}^{(1)}(\eta) := r_{i,2}(\eta), R_{iT}^{(2)}(\eta) := r_{i,1}(\eta) + r_{i,3}(\eta)$. Following the arguments in the proof of Theorem 5.1 in \cite{chao2017} with $n=T, m=p+1, \xi_m= \kappa, g_n=0,$ and $c_n=0$ we have almost surely
\[
\sup_{\eta \in \mathcal{T}} \|r_{i,1}(\eta)\| \lesssim T^{-1}. 
\]
Moreover, on the event
{
\[
\max_{i} \sup_{\eta \in \mathcal{T}}\|\widehat{\bm{\gamma}}_{i}(\eta) - \bm{\gamma}^*_{i}(\eta)\| \leq \delta
\]
}
we have 
\[
\sup_{\eta \in \mathcal{T}} \|r_{i,2}(\eta)\| \lesssim \max_i \norm{\mpr_{i,T}-\mpr_i}_{\mathcal{G}_2(\delta)}
\]
and $\sup_{\eta \in \mathcal{T}} \|r_{i,2}(\eta)\| \lesssim \delta^2$ where the constants in $\lesssim$ depend on the constants from Assumption~\ref{A1}--\ref{A3} only. Letting $\delta = C_3 \sqrt{T^{-1}\log(nT)}$ and recalling~\eqref{eq:tailgammamix} and~\eqref{eq:mixG2_2} completes he proof. \hfill $\Box$

\bigskip

\begin{proof}[Proof sketch of Theorem~\ref{qrint_est1}]
The results for the first part can be established by following the arguments given in the proof of Theorem 5.1 in~\cite{kato2012asymptotics}, which are parallel to the step 1-3 in the proof of Theorem 3.2 therein. The results for the second part can be proved similarly to those for the second part of Theorem~\ref{quantreg_est1}. 
\end{proof}

\end{document}